\documentclass[a4paper,cleveref,USenglish]{lipics-v2021}

\usepackage{libertine}

\usepackage{datetime}
\usepackage{url}
\usepackage{paralist}

\usepackage{algorithm,algpseudocode}
\usepackage{multicol}

\usepackage[font=small,labelfont=bf,justification=centering]{caption}

\usepackage{amsmath,amsthm}

\theoremstyle{definition}

\theoremstyle{definition}

\usepackage{graphicx}

\usepackage{arydshln}

\usepackage{multirow}

\usepackage{threeparttable}

\usepackage{cleveref}
\crefname{section}{\S}{\S\S}
\Crefname{section}{\S}{\S\S}
\crefformat{section}{\S#2#1#3}

\Crefname{assumption}{Assumption}{Assumptions}

\Crefname{invariant}{Invariant}{Invariants}

\Crefname{observation}{Observation}{Observations}

\usepackage{algorithmicx}
\usepackage{algorithm} 
\usepackage{algpseudocode}

\usepackage{soul}

\usepackage{xspace}

\newcommand{\name}{\textsc{Quad}\xspace}
\newcommand{\sname}{\textsc{SQuad}\xspace}
\newcommand{\namebis}{\textsc{Quad}$^{+}$\xspace}

\newcommand{\rare}{\textsc{RareSync}\xspace}
\newcommand{\core}{\text{core}}

\usepackage{xcolor}
\usepackage{framed}
\definecolor{lightgray}{gray}{0.90}
\renewenvironment{leftbar}[1][\hsize]
{%
\MakeFramed{\hsize#1\advance\hsize-\width\FrameRestore}%
}
{\endMakeFramed}

\algsetblockdefx[Class]{Class}{EndClass}{}{}[1]{\textbf{class} {#1}:}{\phantom}
\algtext*{EndClass}
\algsetblockdefx[Interface]{Interface}{EndInterface}{}{}{\textbf{Interface:}}{\phantom}
\algtext*{EndInterface}
\algsetblockdefx[Requests]{Requests}{EndRequests}{}{}{\textbf{Requests:}}{\phantom}
\algtext*{EndRequests}
\algsetblockdefx[Logging]{Logging}{EndLogging}{}{}{\textbf{Logging:}}{\phantom}
\algtext*{EndLogging}
\algsetblockdefx[Algorithm]{Algorithm}{EndAlgorithm}{}{}{\textbf{Algorithm:}}{\phantom}
\algtext*{EndAlgorithm}
\algsetblockdefx[Constants]{Constants}{EndConstants}{}{}{\textbf{Constants:}}{\phantom}
\algtext*{EndConstants}
\algsetblockdefx[Upon]{Upon}{EndUpon}{}{}[1]{\textbf{upon} $#1$}{\textbf{end upon}}
\algsetblockdefx[UponReceipt]{UponReceipt}{EndUponReceipt}{}{}[2]{\textbf{upon receipt of} $#1$ \textbf{ from} $#2$}{\textbf{end upon receipt}}
\algsetblockdefx[Structure]{Structure}{EndStructure}{}{}[1]{$#1 = \{$}{$\}$}

\algsetblockdefx[Procedure]{Procedure}{EndProcedure}{}{}[2]{\textbf{procedure} {#1}($#2$)}{\phantom}
\algtext*{EndProcedure}
\algsetblockdefx[Function]{Function}{EndFunction}{}{}[2]{\textbf{function} {#1}($#2$)}{\phantom}
\algtext*{EndFunction}

\algsetblockdefx[Leader]{Leader}{EndLeader}{}{}{\textbf{as} a leader}{\phantom}
\algtext*{EndLeader}
\algsetblockdefx[Replica]{Replica}{EndReplica}{}{}{\textbf{as} a replica}{\phantom}
\algtext*{EndReplica}

\algnewcommand{\BlueComment}[1]{\textcolor{blue}{\hfill\(\triangleright\) #1}}
\algnewcommand{\LineComment}[1]{\State \(\triangleright\) #1}

\usepackage{bold-extra} 

\usepackage{listings}
\crefname{lstlisting}{listing}{listings}
\Crefname{lstlisting}{Listing}{Listings}

\crefname{code}{line}{lines}
\Crefname{code}{Line}{Lines}

\usepackage{xcolor}

\definecolor{mygreen}{rgb}{0.254,0.572,0.294}
\definecolor{mygray}{rgb}{0.5,0.5,0.5}
\definecolor{myorange}{rgb}{1,0.35,0}
\definecolor{mymauve}{rgb}{0.58,0,0.82}
\definecolor{myblue}{rgb}{0.2,0.4,0.6}

\definecolor{rakos4orange}{RGB}{255,165,0}
\definecolor{rakos4blue}{RGB}{14,48,173}
\definecolor{rakos4lblue}{RGB}{92,172,238}
\definecolor{rakos4dgray}{RGB}{77,77,77}
\definecolor{plainred}{RGB}{211,63,63}
\definecolor{plainorange}{RGB}{221,105,41}

\lstdefinelanguage{Golang}%
  {morekeywords=[1]{package,import,struct,defer,panic,%
     recover,select,var,const,iota, class},%
   morekeywords=[2]{string,uint,uint8,uint16,uint32,uint64,int,int8,int16,%
     int32,int64,bool,float32,float64,complex64,complex128,byte,rune,uintptr,%
     error,interface,node},%
   morekeywords=[3]{map,slice,make,new,nil,len,cap,copy,close,%
     delete,append,real,imag,complex,chan,},%
   morekeywords=[4]{break,continue,goto,switch,case,fallthrough,%
    default,},%
   morekeywords=[5]{Println,Printf,Error,Send},%
   sensitive=true,%
   morecomment=[l]{//},%
   morecomment=[s]{/*}{*/},%
   morestring=[b]",%
   morestring=[s]{`}{`},%
   }

\usepackage{algorithmicx}
\usepackage{algorithm} 
\usepackage{algpseudocode}
\usepackage{xifthen}
\usepackage{graphicx}
\usepackage{wrapfig}
\usepackage{dblfloatfix}

\usepackage{stmaryrd}
\usepackage{amsmath}
\usepackage{amsthm}

\newcommand{\pierre}[1]{{\color{blue}{[Pierre: #1}]}}

\newcommand{\blue}[1]{{\color{blue}{#1}}}

\newcommand{\remove}[1]{}

\usepackage{cancel}
\usepackage{soul}

\newif\ifcomments
\commentstrue

\author{Pierre Civit}{Sorbonne University, France}{}{}{}
\author{Muhammad Ayaz Dzulfikar}{NUS Singapore, Singapore}{}{}{}
\author{Seth Gilbert}{NUS Singapore, Singapore}{}{}{Supported in part by Singapore MOE grant MOE2018-T2-1-160.}
\author{Vincent Gramoli}{University of Sydney and Redbelly Network, Australia}{}{}{}
\author{Rachid Guerraoui}{Ecole Polytechnique Fédérale de Lausanne (EPFL), Switzerland}{}{}{}
\author{Jovan Komatovic}{Ecole Polytechnique Fédérale de Lausanne (EPFL), Switzerland}{}{}{}
\author{Manuel Vidigueira}{Ecole Polytechnique Fédérale de Lausanne (EPFL), Switzerland}{}{}{}

\nolinenumbers 

\hideLIPIcs  

\begin{document}


\title{Byzantine Consensus is $\Theta(n^2)$ \\ \smallskip \Large{The Dolev-Reischuk Bound is Tight even in Partial Synchrony!} \\ \footnotesize{\normalfont{(Extended Version)}}}

\titlerunning{Communication of Deterministic Byzantine Consensus is $\Theta(n^2)$} 


\authorrunning{P. Civit, M. A. Dzulfikar, S. Gilbert, V. Gramoli, R. Guerraoui, J. Komatovic, M. Vidigueira} 

\Copyright{Anon} 



\bibliographystyle{plainurl}

\maketitle
  

\begin{abstract}
The Dolev-Reischuk bound says that any deterministic Byzantine consensus protocol has (at least) quadratic communication complexity in the worst case.
While it has been shown that the bound is tight in synchronous environments, it is still unknown whether a consensus protocol with quadratic communication complexity can be obtained in partial synchrony.
Until now, the most efficient known solutions for Byzantine consensus in partially synchronous settings had cubic communication complexity (e.g., HotStuff, binary DBFT).

This paper closes the existing gap by introducing \sname, a partially synchronous Byzantine consensus protocol with quadratic worst-case communication complexity.
In addition, \sname is optimally-resilient and achieves linear worst-case latency complexity.
The key technical contribution underlying \sname lies in the way we solve \emph{view synchronization}, the problem of bringing all correct processes to the same view with a correct leader for sufficiently long.
Concretely, we present \rare, a view synchronization protocol with quadratic communication complexity and linear latency complexity, which we utilize in order to obtain \sname.
\end{abstract}

\maketitle

\section{Introduction} \label{section:introduction}

Byzantine consensus~\cite{Lamport1982} is a fundamental distributed computing problem.
In recent years, it has become the target of widespread attention due to the advent of blockchain~\cite{Crain2017a,solida,Gramoli20} and decentralized cloud computing~\cite{cloud}, where it acts as a key primitive.
The demand of these contexts for high performance has given a new impetus to research towards Byzantine consensus with optimal communication guarantees.

Intuitively, Byzantine consensus enables processes to agree on a common value despite Byzantine failures. 
Formally, each process is either correct or faulty; correct processes follow a prescribed protocol, whereas faulty processes (up to $f > 0$) can arbitrarily deviate from it.
Each correct process \emph{proposes} a value, and should eventually \emph{decide} a value.
The following properties are guaranteed:
\begin{compactitem}
    \item \emph{Validity:} If all correct processes propose the same value, then only that value can be decided by a correct process.
    
    \item \emph{Agreement:} No two correct processes decide different values.
    
    \item \emph{Termination:} All correct processes eventually decide.
\end{compactitem}

The celebrated Dolev-Reischuk bound~\cite{Dolev1985} says that any deterministic solution of the Byzantine consensus problem requires correct processes to exchange (at least) a quadratic number of bits of information.
It has been shown that the bound is tight in synchronous environments~\cite{berman,Momose2021}.
However, for the partially synchronous environments~\cite{Dwork1988} in which the network becomes synchronous only after some unknown Global Stabilization Time ($\mathit{GST}$), no Byzantine consensus protocol achieving quadratic communication complexity is known.\footnote{No deterministic protocol solves Byzantine consensus in a completely asynchronous environment~\cite{Fischer1985}.}
Therefore, the question remains whether a partially synchronous Byzantine consensus with quadratic communication complexity exists~\cite{DBLP:journals/sigact/CohenKN21}.
Until now, the most efficient known solutions in partially synchronous environments had cubic communication complexity (e.g., HotStuff~\cite{Yin2019}, binary DBFT~\cite{Crain2017a}).

We close the gap by introducing \sname, a partially synchronous Byzantine consensus protocol with quadratic worst-case communication complexity, matching the Dolev-Reischuk~\cite{Dolev1985} bound.
In addition, \sname is optimally-resilient and achieves optimal linear worst-case latency.

\smallskip
\noindent \textbf{Partially synchronous ``leader-based'' Byzantine consensus.}
Partially synchronous ``leader-based'' consensus protocols~\cite{Yin2019, Diem2021,Castro2002,Buchman2018} operate in \emph{views}, each with a designated leader whose responsibility is to drive the system towards a decision.
If a process does not decide in a view, the process moves to the next view with a different leader and tries again.
Once all correct processes overlap in the same view with a correct leader for sufficiently long, a decision is reached.
Sadly, ensuring such an overlap is non-trivial; for example, processes can start executing the protocol at different times or their local clocks may drift before $\mathit{GST}$, thus placing them in views which are arbitrarily far apart.

Typically, these protocols contain two independent modules:
\begin{compactenum}
    \item View \core: The core of the protocol, responsible for executing the protocol logic of each view.
    
    \item View synchronizer: Auxiliary to the view core, responsible for ``moving'' processes to new views with the goal of ensuring a sufficiently long overlap to allow the view core to decide.
\end{compactenum}
Immediately after $\mathit{GST}$, the view synchronizer brings all correct processes together to the view of the most advanced correct process and keeps them in that view for sufficiently long.
At this point, if the leader of the view is correct, the processes decide.
Otherwise, they ``synchronously'' transit to the next view with a different leader and try again.
In summary, the communication complexity of such protocols can be approximated by $n \cdot C + S$, where:
\begin{compactitem}
    
    \item $C$ denotes the maximum number of bits a correct process sends while executing its view core during $[\mathit{GST}, t_d]$, where $t_d$ is the first time by which all correct processes have decided, and
    
    \item $S$ denotes the communication complexity of the view synchronizer during $[\mathit{GST}, t_d]$.
\end{compactitem}
    
    

Since the adversary can corrupt up to $f$ processes, correct processes must transit through at least $f + 1$ views after $\mathit{GST}$, in the worst case, before reaching a correct leader.
In fact, PBFT~\cite{Castro2002} and HotStuff~\cite{Yin2019} show that passing through $f + 1$ views is sufficient to reach a correct leader.
Furthermore, HotStuff employs the ``leader-to-all, all-to-leader'' communication pattern in each view.
As (1) each process is the leader of at most one view during $[\mathit{GST}, t_d]$, and (2) a process sends $O(n)$ bits in a view if it is the leader of the view, and $O(1)$ bits otherwise, HotStuff achieves $C = 1 \cdot O(n) + f \cdot O(1) = O(n)$.
Unfortunately, $S = (f + 1) \cdot O(n^2) = O(n^3)$ in HotStuff due to ``all-to-all'' communication exploited by its view synchronizer in \emph{every} view.\footnote{While HotStuff~\cite{Yin2019} does not explicitly state how the view synchronization is achieved, we have that $S = O(n^3)$ in Diem BFT~\cite{Diem2021}, which is a mature implementation of the HotStuff protocol.} 
Thus, $S = O(n^3)$ dominates the communication complexity of HotStuff, preventing it from matching the Dolev-Reischuk bound.
If we could design a consensus algorithm for which $S = O(n^2)$ while preserving $C = O(n)$, we would obtain a Byzantine consensus protocol with optimal communication complexity.
The question is if a view synchronizer achieving $S = O(n^2)$ in partial synchrony exists.

\smallskip
\noindent \textbf{Warm-up: View synchronization in complete synchrony.}
Solving the synchronization problem in a completely synchronous environment is not hard. 
As all processes start executing the protocol at the same time and their local clocks do not drift, the desired overlap can be achieved without any communication: processes stay in each view for the fixed, overlap-required time.
However, this simple method \emph{cannot} be used in a partially synchronous setting as it is neither guaranteed that all processes start at the same time nor that their local clocks do not drift (before $\mathit{GST}$).
Still, the observation that, if the system is completely synchronous, processes are not required to communicate in order to synchronize plays a crucial role in developing our view synchronizer which achieves quadratic communication complexity in partially synchronous environments.

\smallskip
\noindent \textbf{\rare.}
The main technical contribution of this work is \rare, a partially synchronous view synchronizer that achieves synchronization within $O(f)$ time after $\mathit{GST}$, and has $O(n^2)$ worst-case communication complexity.
In a nutshell, \rare adapts the ``no-communication'' technique of synchronous view synchronizers to partially synchronous environments.





Namely, \rare groups views into \emph{epochs}; each epoch contains $f + 1$ sequential views.
Instead of performing ``all-to-all'' communication in each view (like the ``traditional'' view synchronizers~\cite{Diem2021}), \rare performs a \emph{single} ``all-to-all'' communication step per epoch.
Specifically, \emph{only} at the end of each epoch do all correct processes communicate to enable further progress.
Once a process has entered an epoch, the process relies \emph{solely} on its local clock (without any communication) to move forward to the next view within the epoch.

Let us give a (rough) explanation of how \rare ensures synchronization.
Let $E$ be the smallest epoch entered by \emph{all} correct processes at or after $\mathit{GST}$; let the first correct process enter $E$ at time $t_E \geq \mathit{GST}$.
Due to (1) the ``all-to-all'' communication step performed at the end of the previous epoch $E - 1$, and (2) the fact that message delays are bounded by a known constant $\delta$ after $\mathit{GST}$, all correct processes enter $E$ by time $t_E + \delta$.
Hence, from the epoch $E$ onward, processes do not need to communicate in order to synchronize: it is sufficient for processes to stay in each view for $\delta + \Delta$ time to achieve $\Delta$-time overlap.
In brief, \rare uses communication to synchronize processes, while relying on local timeouts (and not communication!) to keep them synchronized.

\smallskip
\noindent \textbf{\sname.}
The second contribution of our work is \sname, an optimally-resilient partially synchronous Byzantine consensus protocol with (1) $O(n^2)$ worst-case communication complexity, and (2) $O(f)$ worst-case latency complexity. 
The view core module of \sname is the same as that of HotStuff;
as its view synchronizer, \sname uses \rare.
The combination of the HotStuff's view core and \rare ensures that $C = O(n)$ and $S = O(n^2)$.
By the aforementioned complexity formula, \sname achieves $n \cdot O(n) + O(n^2) = O(n^2)$ communication complexity.
\sname's linear latency is a direct consequence of \rare's ability to synchronize processes within $O(f)$ time after $\mathit{GST}$.

\smallskip
\noindent \textbf{Roadmap.}
We discuss related work in \Cref{section:related_work}.
In \Cref{section:model}, we define the system model.
We introduce \rare in \Cref{section:rare}.
In \Cref{section:quad}, we present \sname.
We conclude the paper in \Cref{section:conclusion}.
Detailed proofs are delegated to the optional appendix.

\section{Related Work}\label{section:related_work}

In this section, we discuss existing results in two related contexts: synchronous networks and randomized algorithms.  In addition, we discuss some precursor (and concurrent) results to our own. 

\smallskip
\noindent \textbf{Synchronous networks.} The first natural question is whether we can achieve synchronous Byzantine agreement with optimal latency and optimal communication complexity.  Momose and Ren answer that question in the affirmative, giving a synchronous  Byzantine agreement protocol with optimal $n/2$ resiliency, optimal $O(n^2)$ worst-case communication complexity and optimal $O(f)$ worst-case latency~\cite{Momose2021}. Optimality follows from two lower bounds: Dolev and Reischuk show that any Byzantine consensus protocol has an execution with quadratic communication complexity~\cite{Dolev1985}; Dolev and Strong show that any synchronous Byzantine consensus protocol has an execution with $f+1$ rounds \cite{Dolev1983}. 
Various other works have tackled the problem of minimizing the latency of Byzantine consensus~\cite{DBLP:conf/wdag/AbrahamDN017,Locher2020,Micali2017ByzantineA}.

\smallskip
\noindent \textbf{Randomization.} A classical approach to circumvent the FLP impossibility~\cite{Fischer1985} is using randomization~\cite{Ben-Or1983b}, where termination is not ensured deterministically. 
Exciting recent results by Abraham \textit{et al.}~\cite{Abraham2019a} and Lu \textit{et al.}~\cite{Lu2020} give fully asynchronous randomized Byzantine consensus with optimal $n/3$ resiliency, optimal $O(n^2)$ expected communication complexity and optimal $O(1)$ expected latency complexity.  Spiegelman~\cite{Spiegelman2021} took a neat \emph{hybrid} approach that achieved optimal results for both synchrony and randomized asynchrony simultaneously: if the network is synchronous, his algorithm yields optimal (deterministic) synchronous complexity; if the network is asynchronous, it falls back on a randomized algorithm and achieves optimal randomized complexity.

Recently, it has been shown that even randomized Byzantine agreement requires $\Omega(n^2)$ expected communication complexity, at least for achieving guaranteed safety against an \emph{adaptive adversary} in an asynchronous setting or against a \emph{strongly rushing adaptive adversary} in a synchronous setting~\cite{Abraham2019b,abraham2019asymptotically}.  (See the papers for details.)  Amazingly, it is possible to break the $O(n^2)$ barrier by accepting a non-zero (but $o(1)$) probability of disagreement~\cite{Chen2018, Cohen2020, King2011}. 


\smallskip 
\noindent \textbf{Authentication.} Most of the results above are \emph{authenticated}: they assume a trusted setup phase\footnote{A trusted setup phase is notably different from randomized algorithms where randomization is used throughout.} wherein devices establish and exchange cryptographic keys; this allows for messages to be signed in a way that proves who sent them. Recently, many of the communication-efficient agreement protocols (such as~\cite{Abraham2019a, Lu2020}) rely on \emph{threshold signatures} (such as~\cite{Libert2016}).  The Dolev-Reischuk~\cite{Dolev1985} lower bound shows that quadratic communication is needed even in such a case (as it looks at the message complexity of authenticated agreement).

Among deterministic, non-authenticated Byzantine agreement protocols, DBFT~\cite{Crain2017a} achieves $O(n^3)$ communication complexity.  For randomized non-authenticated Byzantine agreement protocols, Mostefaoui \textit{et al.}~\cite{MostefaouiMR15} achieve $O(n^2)$ communication complexity---but they assume a perfect common coin, for which efficient implementations may also require signatures. 

We note that it is possible to (1) work towards an authenticated setting from a non-authenticated one by rolling out a public key infrastructure (PKI) \cite{Bracha87, AD15, GKLP18}, (2) set up a threshold scheme \cite{AbrahamJMMST21} without a \emph{trusted dealer}, and (3) asynchronously emulate a perfect common coin \cite{CKS05} used by randomized Byzantine consensus protocols \cite{Rabin83, MostefaouiMR15, Abraham2019a, Lu2020}.


\smallskip
\noindent \textbf{Other related work.}   In this paper, we focus on the partially synchronous setting~\cite{Dwork1988}, where the question of optimal communication complexity of Byzantine agreement has remained open. 
The question can be addressed precisely with the help of rigorous frameworks \cite{Gafni1998, Guerraoui2004, Keidar2006} that were developed to express partially synchronous protocols using a round-based paradigm. More specifically,  state-of-the-art partially synchronous BFT protocols~\cite{Diem2021, Buchman2018, Yin2019, GolanGueta2019} have been developed within a view-based paradigm with a rotating leader, e.g., the seminal PBFT protocol \cite{Castro2002}.
While many approaches improve the complexity for some optimistic scenarios~\cite{Martin2005,Ramasamy2006, Kotla2009,Kuznetsov2021,Pass2018a}, none of them were able to reach the quadratic worst-case Dolev-Reischuk bound.  

The problem of view synchronization was defined in \cite{Naor2021}.
An existing implementation of this abstraction \cite{GolanGueta2019} was based on Bracha's double-echo reliable broadcast at each view, inducing a cubic communication complexity in total. This communication complexity has been reduced for some optimistic scenarios~\cite{Naor2021} and in terms of \emph{expected} complexity~\cite{Naor2020}. The problem has been formalized more precisely in~\cite{Bravo2020} to facilitate formal verification of PBFT-like protocols.

It might be worthwhile highlighting some connections between the view synchronization abstraction and the leader election abstraction $\Omega$~\cite{ Chandra1996, Chandra1992}, capturing the  weakest failure detection information needed to solve consensus (and extended to the Byzantine context in~\cite{Kihlstrom2003a}).   Leaderless partially synchronous Byzantine consensus protocols have also been proposed~\cite{Antoniadis2021a}, somehow indicating that the notion of a leader is not necessary in the mechanisms of a consensus protocol,  even if $\Omega$ is the weakest failure detector needed to solve the problem.
Clock synchronization \cite{ Dolev1995, Srikanth1987} and view synchronization are orthogonal problems.

\smallskip
\noindent \textbf{Concurrent research.}
We have recently discovered concurrent and independent research by Lewis-Pye~\cite{LewisPye}.  Lewis-Pye appears to have discovered a similar approach to the one that we present in this paper, giving an algorithm for state machine replication in a partially synchronous model with quadratic message complexity.  As in this paper, Lewis-Pye makes the key observation that we do not need to synchronize in every view; views can be grouped together, with synchronization occurring only once every fixed number of views.  This yields essentially the same algorithmic approach. Lewis-Pye focuses on state machine replication, instead of Byzantine agreement (though state machine replication is implemented via repeated Byzantine agreement).  
The other useful property of his algorithm is \emph{optimistic responsiveness}, which applies to the multi-shot case and ensures that, in good portions of the executions, decisions happen as quickly as possible. We encourage the reader to look at~\cite{LewisPye} for a different presentation of a similar approach.

Moreover, the similar approach to ours and Lewis-Pye's has been proposed in the first version of HotStuff~\cite{abraham2018hot}: processes synchronize once per \emph{level}, where each level consists of $n$ views.
The authors mention that this approach guarantees the quadratic communication complexity; however, this claim was not formally proven in their work.
The claim was dropped in later versions of HotStuff (including the published version).
We hope readers of our paper will find an increased appreciation of the ideas introduced by HotStuff.

\section{System Model} \label{section:model}

\noindent\textbf{Processes.}
We consider a static set $\{P_1, P_2, ..., P_n\}$ of $n = 3f + 1$ processes out of which at most $f$ can be Byzantine, i.e., can behave arbitrarily.
If a process is Byzantine, the process is \emph{faulty}; otherwise, the process is \emph{correct}.
Processes communicate by exchanging messages over an authenticated point-to-point network.
The communication network is \emph{reliable:} if a correct process sends a message to a correct process, the message is eventually received.
We assume that processes have local hardware clocks.
Furthermore, we assume that local steps of processes take zero time, as the time needed for local computation is negligible compared to message delays.
Finally, we assume that no process can take infinitely many steps in finite time.

\smallskip
\noindent\textbf{Partial synchrony.}
We consider the partially synchronous model introduced in~\cite{Dwork1988}.
For every execution, there exists a Global Stabilization Time ($\mathit{GST}$) and a positive duration $\delta$ such that message delays are bounded by $\delta$ after $\mathit{GST}$.
Furthermore, $\mathit{GST}$ is not known to processes, whereas $\delta$ is known to processes.
We assume that all correct processes start executing their protocol by $\mathit{GST}$.
The hardware clocks of processes may drift arbitrarily before $\mathit{GST}$, but do not drift thereafter.

    

\smallskip
\noindent \textbf{Cryptographic primitives.}
We assume a $(k, n)$-threshold signature scheme~\cite{Libert2016},
where $k = 2f + 1 = n - f$.
In this scheme, each process holds a distinct private key and there is a single public key.
Each process $P_i$ can use its private key to produce a partial signature of a message $m$ by invoking $\mathit{ShareSign}_i(m)$.
A partial signature $\mathit{tsignature}$ of a message $m$ produced by a process $P_i$ can be verified by $\mathit{ShareVerify}_i(m, \mathit{tsignature})$.
Finally, set $S = \{\mathit{tsignature}_i\}$ of partial signatures, where $|S| = k$ and, for each $\mathit{tsignature}_i \in S$, $\mathit{tsignature}_i = \mathit{ShareSign}_i(m)$, can be combined into a \emph{single} (threshold) signature by invoking $\mathit{Combine}(S)$; a combined signature $\mathit{tcombined}$ of message $m$ can be verified by $\mathit{CombinedVerify}(m, \mathit{tcombined})$.
Where appropriate, invocations of $\mathit{ShareVerify}(\cdot)$ and $\mathit{CombinedVerify}(\cdot)$ are implicit in our descriptions of protocols.
$\mathsf{P\_Signature}$ and $\mathsf{T\_Signature}$ denote a partial signature and a (combined) threshold signature, respectively. 

\smallskip
\noindent \textbf{Complexity of Byzantine consensus.}
Let $\mathsf{Consensus}$ be a partially synchronous Byzantine consensus protocol and let $\mathcal{E}(\mathsf{Consensus})$ denote the set of all possible executions.
Let $\alpha \in \mathcal{E}(\mathsf{Consensus})$ be an execution and $t_d(\alpha)$ be the first time by which all correct processes have decided in $\alpha$.

A \emph{word} contains a constant number of signatures and values.
Each message contains at least a single word.
We define the communication complexity of $\alpha$ as the number of words sent in messages by all correct processes during the time period $[\mathit{GST}, t_d(\alpha)]$; if $\mathit{GST} > t_d(\alpha)$, the communication complexity of $\alpha$ is $0$. 
The latency complexity of $\alpha$ is $\max(0, t_d(\alpha) - \mathit{GST})$.

The \emph{communication complexity} of $\mathsf{Consensus}$ is defined as
\begin{equation*}
\max_{\alpha \in \mathcal{E}(\mathsf{Consensus})}\bigg\{\text{communication complexity of } \alpha\bigg\}.
\end{equation*}

Similarly, the \emph{latency complexity} of $\mathsf{Consensus}$ is defined as 
\begin{equation*}
\max_{\alpha \in \mathcal{E}(\mathsf{Consensus})}\bigg\{\text{latency complexity of } \alpha\bigg\}.
\end{equation*}

We underline that the number of words sent by correct processes before $\mathit{GST}$ is unbounded in any partially synchronous Byzantine consensus protocol~\cite{Spiegelman2021}.
Moreover, not a single correct process is guaranteed to decide before $\mathit{GST}$ in any partially synchronous Byzantine consensus protocol~\cite{Fischer1985}; that is why the latency complexity of such protocols is measured from $\mathit{GST}$.

\section{\rare} \label{section:rare}

This section presents \rare, a partially synchronous view synchronizer that achieves synchronization within $O(f)$ time after $\mathit{GST}$, and has $O(n^2)$ worst-case communication complexity.
First, we define the problem of view synchronization (\Cref{subsection:view_synchronization_problem_definition}).
Then, we describe \rare, and present its pseudocode (\Cref{subsection:rare}).
Finally, we reason about \rare's correctness and complexity (\Cref{subsection:rare_proof_sketch}).


\subsection{Problem Definition} \label{subsection:view_synchronization_problem_definition}

View synchronization is defined as the problem of bringing all correct processes to the same view with a correct leader for sufficiently long~\cite{Bravo2020,Naor2020, Naor2021}.
More precisely, let $\mathsf{View} = \{1, 2, ...\}$ denote the set of views.
For each view $v \in \mathsf{View}$, we define $\mathsf{leader}(v)$ to be a process that is the \emph{leader} of view $v$.
The view synchronization problem is associated with a predefined time $\Delta > 0$, which denotes the desired duration during which processes must be in the same view with a correct leader in order to synchronize.
View synchronization provides the following interface:
\begin{compactitem}
    \item \textbf{Indication} $\mathsf{advance}(\mathsf{View} \text{ } v)$: The process advances to a view $v$.
\end{compactitem}
We say that a correct process \emph{enters} a view $v$ at time $t$ if and only if the $\mathsf{advance}(v)$ indication occurs at time $t$.
Moreover, a correct process \emph{is in view} $v$ between the time $t$ (including $t$) at which the $\mathsf{advance}(v)$ indication occurs and the time $t'$ (excluding $t'$) at which the next $\mathsf{advance}(v' \neq v)$ indication occurs.
If an $\mathsf{advance}(v' \neq v)$ indication never occurs, the process remains in the view $v$ from time $t$ onward.

Next, we define a \emph{synchronization time} as a time at which all correct processes are in the same view with a correct leader for (at least) $\Delta$ time.

\begin{definition} [Synchronization time] \label{definition:synchronization_time}
Time $t_s$ is a \emph{synchronization time} if (1) all correct processes are in the same view $v$ from time $t_s$ to (at least) time $t_s + \Delta$, and (2) $\mathsf{leader}(v)$ is correct.
\end{definition}

View synchronization ensures the \emph{eventual synchronization} property which states that there exists a synchronization time at or after $\mathit{GST}$.

\smallskip
\noindent \textbf{Complexity of view synchronization.}
Let $\mathsf{Synchronizer}$ be a partially synchronous view synchronizer and let $\mathcal{E}(\mathsf{Synchronizer})$ denote the set of all possible executions.
Let $\alpha \in \mathcal{E}(\mathsf{Synchronizer})$ be an execution and $t_s(\alpha)$ be the first synchronization time at or after $\mathit{GST}$ in $\alpha$ ($t_s(\alpha) \geq \mathit{GST}$).
We define the communication complexity of $\alpha$ as the number of words sent in messages by all correct processes during the time period $[\mathit{GST}, t_s(\alpha) + \Delta]$.
The latency complexity of $\alpha$ is $t_s(\alpha) + \Delta - \mathit{GST}$.

The \emph{communication complexity} of $\mathsf{Synchronizer}$ is defined as
\begin{equation*}
\max_{\alpha \in \mathcal{E}(\mathsf{Synchronizer})}\bigg\{\text{communication complexity of } \alpha\bigg\}.
\end{equation*}

Similarly, the \emph{latency complexity} of $\mathsf{Synchronizer}$ is defined as 
\begin{equation*}
\max_{\alpha \in \mathcal{E}(\mathsf{Synchronizer})}\bigg\{\text{latency complexity of } \alpha\bigg\}.
\end{equation*}

\subsection{Protocol} \label{subsection:rare}

This subsection details \rare (\Cref{algorithm:synchronizer}).
In essence, \rare achieves $O(n^2)$ communication complexity and $O(f)$ latency complexity by exploiting ``all-to-all'' communication only once per $f + 1$ views.

\smallskip 
\noindent \textbf{Intuition.}
We group views into \emph{epochs}, where each epoch contains $f + 1$ sequential views; $\mathsf{Epoch} = \{1, 2, ...\}$ denotes the set of epochs.
Processes move through an epoch solely by means of local timeouts (without any communication).
However, at the end of each epoch, processes engage in an ``all-to-all'' communication step to obtain permission to move onto the next epoch: (1) Once a correct process has completed an epoch, it broadcasts a message informing other processes of its completion; (2) Upon receiving $2f + 1$ of such messages, a correct process enters the future epoch. Note that (2) applies to \emph{all} processes, including those in arbitrarily ``old'' epochs.
Overall, this ``all-to-all'' communication step is the \emph{only} communication processes perform within a single epoch, implying that per-process communication complexity in each epoch is $O(n)$.
\Cref{fig:raresync} illustrates the main idea behind \rare.

\begin{figure}[h]
    \centering
    \includegraphics[scale=0.2]{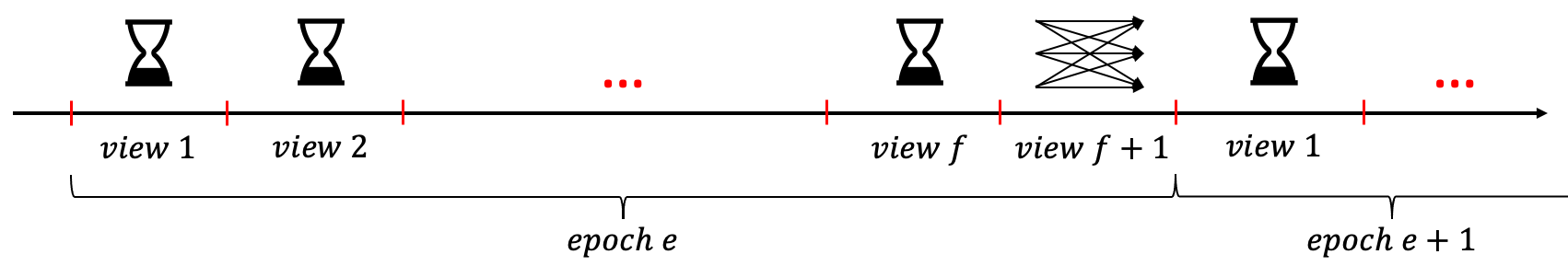}
    \caption{Intuition behind \rare: Processes communicate only in the last view of an epoch; before the last view, they rely solely on local timeouts.}
    \label{fig:raresync}
\end{figure}

Roughly speaking, after $\mathit{GST}$, all correct processes simultaneously enter the same epoch within $O(f)$ time.
After entering the same epoch, processes are guaranteed to synchronize in that epoch, which takes (at most) an additional $O(f)$ time.
Thus, the latency complexity of \rare is $O(f)$.
The communication complexity of \rare is $O(n^2)$ as every correct process executes at most a constant number of epochs, each with $O(n)$ per-process communication, after $\mathit{GST}$.

\smallskip
\noindent \textbf{Protocol description.}
We now explain how \rare works.
The pseudocode of \rare is given in \Cref{algorithm:synchronizer}, whereas all variables, constants, and functions are presented in \Cref{algorithm:variable_constants}.

\begin{algorithm} [b]
\caption{\rare: Variables (for process $P_i$), constants, and functions}
\label{algorithm:variable_constants}
\begin{algorithmic} [1]

\State \textbf{Variables:}
\State \hskip2em $\mathsf{Epoch}$ $\mathit{epoch}_i \gets 1$ \BlueComment{current epoch}
\State \hskip2em $\mathsf{View}$ $\mathit{view}_i \gets 1$ \BlueComment{current view within the current epoch; $\mathit{view}_i \in [1, f + 1]$}  \label{line:init_view}
\State \hskip2em $\mathsf{Timer}$ $\mathit{view\_timer}_i$ \BlueComment{measures the duration of the current view}
\State \hskip2em $\mathsf{Timer}$ $\mathit{dissemination\_timer}_i$ \BlueComment{measures the duration between two communication steps}
\State \hskip2em $\mathsf{T\_Signature}$ $\mathit{epoch\_sig}_i \gets \bot$ \BlueComment{proof that $\mathit{epoch}_i$ can be entered}

\smallskip
\State \textbf{Constants:}
\State \hskip2em $\mathsf{Time}$ $\mathit{view\_duration} = \Delta + 2\delta$ \BlueComment{duration of each view} \label{line:view_duration}

\smallskip
\State \textbf{Functions:}
\State \hskip2em $\mathsf{leader}(\mathsf{View} \text{ } v) \equiv P_{(v \text{ } \mathsf{mod} \text{ } n) + 1}$ \BlueComment{a round-robin function} \label{line:leader_rare}
\end{algorithmic}
\end{algorithm} 

\begin{algorithm} [t]
\caption{\rare: Pseudocode (for process $P_i$)}
\label{algorithm:synchronizer}
\begin{algorithmic} [1]

\State \textbf{upon} $\mathsf{init}$: \BlueComment{start of the protocol} \label{line:start_algorithm}

\State \hskip2em
$\mathit{view\_timer}_i.\mathsf{measure(}\mathit{view\_duration}\mathsf{)}$ \BlueComment{measure the duration of the first view} \label{line:view_timer_measure_first_view} 

\State \hskip2em \textbf{trigger} $\mathsf{advance(}1\mathsf{)}$ \label{line:start_view_1} \BlueComment{enter the first view}

\smallskip
\State \textbf{upon} $\mathit{view\_timer}_i$ \textbf{expires}:
\label{line:rule_view_expired}
\State \hskip2em \textbf{if} $\mathit{view}_i < f + 1$: \BlueComment{check if the current view is not the last view of the current epoch} \label{line:check_last_view}
\State \hskip4em $\mathit{view}_i \gets \mathit{view}_i + 1$ \label{line:increment_view}
\State \hskip4em $\mathsf{View}$ $\mathit{view\_to\_advance} \gets (\mathit{epoch}_i - 1) \cdot (f + 1) + \mathit{view}_i$

\State \hskip4em
$\mathit{view\_timer}_i.\mathsf{measure(}\mathit{view\_duration}\mathsf{)}$ \BlueComment{measure the duration of the view} \label{line:view_timer_measure_without_msg_exchange} 

\State \hskip4em \textbf{trigger} $\mathsf{advance(}\mathit{view\_to\_advance}\mathsf{)}$ \label{line:start_view_without_msg_exchange} \BlueComment{enter the next view}

\State \hskip2em \textbf{else:}
\State \hskip4em \textcolor{blue}{\(\triangleright\) inform other processes that the epoch is completed}
\State \hskip4em \textbf{broadcast} $\langle \textsc{epoch-completed}, \mathit{epoch}_i, \mathit{ShareSign}_i(\mathit{epoch}_i) \rangle$ \label{line:broadcast_epoch_over}

\smallskip
\State \textbf{upon} exists $\mathsf{Epoch}$ $e$ such that $e \geq \mathit{epoch}_i$ and $\langle \textsc{epoch-completed}, e, \mathsf{P\_Signature} \text{ } \mathit{sig} \rangle$ is received from $2f + 1$ processes: \label{line:receive_epoch_over}
\State \hskip2em $\mathit{epoch\_sig}_i \gets \mathit{Combine}\big(\{\mathit{sig} \,|\, \mathit{sig} \text{ is received in an } \textsc{epoch-completed} \text{ message}\}\big)$ \label{line:combine-signatures}
\State \hskip2em $\mathit{epoch}_i \gets e + 1$  \label{line:update_epoch_1} 
\State \hskip2em
$\mathit{view\_timer}_i.\mathsf{cancel()}$ \label{line:cancel_view_timer_1}
\State \hskip2em $\mathit{dissemination\_timer}_i.\mathsf{cancel()}$ \label{line:sync_timer_cancel_1}
\State \hskip2em $\mathit{dissemination\_timer}_i.\mathsf{measure(}\delta\mathsf{)}$ \label{line:sync_timer_measure} \BlueComment{wait $\delta$ time before broadcasting \textsc{enter-epoch}}

\smallskip
\State \textbf{upon} reception of $\langle \textsc{enter-epoch}, \mathsf{Epoch} \text{ } e, \mathsf{T\_Signature} \text{ } \mathit{sig} \rangle$ such that $e > \mathit{epoch}_i$: \label{line:receive_epoch_over_complete}
\State \hskip2em $\mathit{epoch\_sig}_i \gets \mathit{sig}$ \label{line:adopt_signature} \BlueComment{$\mathit{sig}$ is a threshold signature of epoch $e - 1$}
\State \hskip2em $\mathit{epoch}_i \gets e$  \label{line:update_epoch_2} 
\State \hskip2em
$\mathit{view\_timer}_i.\mathsf{cancel()}$ \label{line:cancel_view_timer_2}
\State \hskip2em $\mathit{dissemination\_timer}_i.\mathsf{cancel()}$  \label{line:sync_timer_cancel_2}
\State \hskip2em $\mathit{dissemination\_timer}_i.\mathsf{measure(}\delta\mathsf{)}$ \label{line:sync_timer_measure_2} \BlueComment{wait $\delta$ time before broadcasting \textsc{enter-epoch}}

\smallskip
\State \textbf{upon} $\mathit{dissemination\_timer}_i$ \textbf{expires}: \label{line:sync_timer_expires}
\State \hskip2em \textbf{broadcast} $\langle \textsc{enter-epoch}, \mathit{epoch}_i, \mathit{epoch\_sig}_i \rangle$ \label{line:broadcast_epoch_over_complete}
\State \hskip2em $\mathit{view}_i \gets 1$ \BlueComment{reset the current view to $1$} \label{line:reset_view}
\State \hskip2em $\mathsf{View}$ $\mathit{view\_to\_advance} \gets (\mathit{epoch}_i - 1) \cdot (f + 1) + \mathit{view}_i$

\State \hskip2em
$\mathit{view\_timer}_i.\mathsf{measure(}\mathit{view\_duration}\mathsf{)}$   \label{line:view_timer_measure} \BlueComment{measure the duration of the view}

\State \hskip2em \textbf{trigger} $\mathsf{advance(}\mathit{view\_to\_advance}\mathsf{)}$ \label{line:start_view_2} \BlueComment{enter the first view of the new epoch}

\end{algorithmic}
\end{algorithm}

We explain \rare's pseudocode (\Cref{algorithm:synchronizer}) from the perspective of a correct process $P_i$.
Process $P_i$ utilizes two timers: $\mathit{view\_timer}_i$ and $\mathit{dissemination\_timer}_i$.
A timer has two methods:
\begin{compactenum}
    \item $\mathsf{measure}(\mathsf{Time} \text{ } x)$: After exactly $x$ time as measured by the local clock, an expiration event is received by the host.
    Note that, as local clocks can drift before $\mathit{GST}$, $x$ time as measured by the local clock may not amount to $x$ real time (before $\mathit{GST}$).
    \item $\mathsf{cancel()}$: 
    This method cancels all previously invoked $\mathsf{measure}(\cdot)$ methods on that timer, i.e., all pending expiration events (pertaining to that timer) are removed from the event queue.
\end{compactenum}
In \rare, $\mathsf{leader}(\cdot)$ is a round-robin function (line~\ref{line:leader_rare} of \Cref{algorithm:variable_constants}).

Once $P_i$ starts executing $\rare$ (line~\ref{line:start_algorithm}), it instructs $\mathit{view\_timer}_i$ to measure the duration of the first view (line~\ref{line:view_timer_measure_first_view}) and it enters the first view (line~\ref{line:start_view_1}).

Once $\mathit{view\_timer}_i$ expires (line~\ref{line:rule_view_expired}), $P_i$ checks whether the current view is the last view of the current epoch, $\mathit{epoch}_i$ (line~\ref{line:check_last_view}).
If that is not the case, the process advances to the next view of $\mathit{epoch}_i$ (line~\ref{line:start_view_without_msg_exchange}).
Otherwise, the process broadcasts an $\textsc{epoch-completed}$ message (line~\ref{line:broadcast_epoch_over}) signaling that it has completed $\mathit{epoch}_i$.
At this point in time, the process does not enter any view.

If, at any point in time, $P_i$ receives either (1) $2f + 1$ $\textsc{epoch-completed}$ messages for some epoch $e \geq \mathit{epoch}_i$ (line~\ref{line:receive_epoch_over}), or (2) an \textsc{enter-epoch} message for some epoch $e' > \mathit{epoch}_i$ (line~\ref{line:receive_epoch_over_complete}), the process obtains a proof that a new epoch $E > \mathit{epoch}_i$ can be entered.
However, before entering $E$ and propagating the information that $E$ can be entered, $P_i$ waits $\delta$ time (either line~\ref{line:sync_timer_measure} or line~\ref{line:sync_timer_measure_2}).
This $\delta$-waiting step is introduced to limit the number of epochs $P_i$ can enter within any $\delta$ time period after $\mathit{GST}$ and is crucial for keeping the communication complexity of \rare quadratic.
For example, suppose that processes are allowed to enter epochs and propagate \textsc{enter-epoch} messages without waiting.
Due to an accumulation (from before $\mathit{GST}$) of \textsc{enter-epoch} messages for different epochs, a process might end up disseminating an arbitrary number of these messages by receiving them all at (roughly) the same time.
To curb this behavior, given that message delays are bounded by $\delta$ after $\mathit{GST}$, we force a process to wait $\delta$ time, during which it receives all accumulated messages, before entering the largest known epoch.

Finally, after $\delta$ time has elapsed (line~\ref{line:sync_timer_expires}), $P_i$ disseminates the information that the epoch $E$ can be entered (line~\ref{line:broadcast_epoch_over_complete}) and it enters the first view of $E$ (line~\ref{line:start_view_2}).

\subsection{Correctness and Complexity: Proof Sketch} \label{subsection:rare_proof_sketch}

This subsection presents a proof sketch of the correctness, latency complexity, and communication complexity of \rare. The full proof can be found in \Cref{section:appendix_rare}.

In order to prove the correctness of \rare, we must show that the eventual synchronization property is ensured, i.e., there is a synchronization time $t_s \geq \mathit{GST}$.
For the latency complexity, it suffices to bound $t_s + \Delta - \mathit{GST}$ by $O(f)$.
This is done by proving that synchronization happens within (at most) 2 epochs after $\mathit{GST}$.
As for the communication complexity, we prove that any correct process enters a constant number of epochs during the time period $[\mathit{GST}, t_s + \Delta]$. 
Since every correct process sends $O(n)$ words per epoch, the communication complexity of \rare is $O(n^2) = O(1) \cdot O(n) \cdot n$.
We work towards these conclusions by introducing some key concepts and presenting a series of intermediate results.

A correct process \emph{enters} an epoch $e$ at time $t$ if and only if the process enters the first view of $e$ at time $t$ (either line~\ref{line:start_view_1} or line~\ref{line:start_view_2}).
We denote by $t_e$ the first time a correct process enters epoch $e$.

\smallskip
\noindent \textbf{Result 1:} \emph{If a correct process enters an epoch $e > 1$, then (at least) $f + 1$ correct processes have previously entered epoch $e - 1$.} 
\smallskip
\\ The goal of the communication step at the end of each epoch is to prevent correct processes from arbitrarily entering future epochs. 
In order for a new epoch $e > 1$ to be entered, at least $f + 1$ correct processes must have entered and ``gone through'' each view of the previous epoch, $e - 1$. 
This is indeed the case: in order for a correct process to enter $e$, the process must either (1) collect $2f + 1$ $\textsc{epoch-completed}$ messages for $e - 1$ (line~\ref{line:receive_epoch_over}), or (2) receive an \textsc{enter-epoch} message for $e$, which contains a threshold signature of $e - 1$ (line~\ref{line:receive_epoch_over_complete}). 
In either case, at least $f + 1$ correct processes must have broadcast $\textsc{epoch-completed}$ messages for epoch $e - 1$ (line~\ref{line:broadcast_epoch_over}), which requires them to go through epoch $e - 1$.
Furthermore, $t_{e - 1} \leq t_{e}$; recall that local clocks can drift before $\mathit{GST}$.

\smallskip
\noindent \textbf{Result 2:} \emph{Every epoch is eventually entered by a correct process.}
\smallskip
\\By contradiction, consider the greatest epoch ever entered by a correct process, $e^*$.
In brief, every correct process will eventually (1) receive the \textsc{enter-epoch} message for $e^*$ (line~\ref{line:receive_epoch_over_complete}), (2) enter $e^*$ after its $\mathit{dissemination\_timer}$ expires (lines~\ref{line:sync_timer_expires} and~\ref{line:start_view_2}), (3) send an $\textsc{epoch-completed}$ message for $e^*$ (line~\ref{line:broadcast_epoch_over}), (4) collect $2f+1$ $\textsc{epoch-completed}$ messages for $e^*$ (line~\ref{line:receive_epoch_over}), and, finally, (5) enter $e^* + 1$ (lines~\ref{line:update_epoch_1},~\ref{line:sync_timer_measure},~\ref{line:sync_timer_expires} and~\ref{line:start_view_2}), resulting in a contradiction.
Note that, if $e^* = 1$, no \textsc{enter-epoch} message is sent: all correct processes enter $e^* = 1$ once they start executing \rare (line~\ref{line:start_view_1}).

\smallskip
We now define two epochs: $e_{\mathit{max}}$ and $e_{\mathit{final}} = e_{\mathit{max}} + 1$.
These two epochs are the main protagonists in the proof of correctness and complexity of \rare.

\smallskip
\noindent \textbf{Definition of $\boldsymbol{e_{\mathit{max}}}$:}
\emph{Epoch $e_{\mathit{max}}$ is the greatest epoch entered by a correct process before $\mathit{GST}$; if no such epoch exists, $e_{\mathit{max}} = 0$.\footnote{Epoch $0$ is considered as a special epoch. Note that $0 \notin \mathsf{Epoch}$, where $\mathsf{Epoch}$ denotes the set of epochs (see \Cref{subsection:rare}).}}

\smallskip
\noindent \textbf{Definition of $\boldsymbol{e_{\mathit{final}}}$:} \emph{Epoch $e_{\mathit{final}}$ is the smallest epoch first entered by a correct process at or after $\mathit{GST}$. Note that $\mathit{GST} \leq t_{e_{\mathit{final}}}$. Moreover, $e_{\mathit{final}} = e_{\mathit{max}} + 1$ (by Result 1).}

\smallskip
\noindent \textbf{Result 3:} \emph{For any epoch $e \geq e_{\mathit{final}}$, no correct process broadcasts an \textsc{epoch-completed} message for $e$ (line~\ref{line:broadcast_epoch_over}) before time $t_e + \mathit{epoch\_duration}$, where $\mathit{epoch\_duration} = (f + 1) \cdot \mathit{view\_duration}$.}
\smallskip
\\ This statement is a direct consequence of the fact that, after $\mathit{GST}$, it takes exactly $\mathit{epoch\_duration}$ time for a process to go through $f+1$ views of an epoch; local clocks do not drift after $\mathit{GST}$.
Specifically, the earliest a correct process can broadcast an \textsc{epoch-completed} message for $e$ (line~\ref{line:broadcast_epoch_over}) is at time $t_e + \mathit{epoch\_duration}$, where $t_e$ denotes the first time a correct process enters epoch $e$.

\smallskip
\noindent \textbf{Result 4:} \emph{Every correct process enters epoch $e_{\mathit{final}}$ by time $t_{e_{\mathit{final}}} + 2\delta$.} 
\smallskip
\\ Recall that the first correct process enters $e_{\mathit{final}}$ at time $t_{e_{\mathit{final}}}$. 
If $e_{\mathit{final}} = 1$, all correct processes enter $e_{\mathit{final}}$ at $t_{e_{\mathit{final}}}$.
Otherwise, by time $t_{e_{\mathit{final}}} + \delta$, all correct processes will have received an \textsc{enter-epoch} message for $e_{\mathit{final}}$ and started the $\mathit{dissemination\_timer}_i$ with $epoch_i = e_{\mathit{final}}$ (either lines~\ref{line:update_epoch_1},~\ref{line:sync_timer_measure} or~\ref{line:update_epoch_2},~\ref{line:sync_timer_measure_2}).
By results 1 and 3, no correct process sends an \textsc{epoch-completed} message for an epoch $\geq e_{\mathit{final}}$ (line~\ref{line:broadcast_epoch_over}) before time $t_{e_{\mathit{final}}} + \mathit{epoch\_duration}$, which implies that the $\mathit{dissemination\_timer}$ will not be cancelled.
Hence, the $\mathit{dissemination\_timer}$ will expire by time $t_{e_{\mathit{final}}} + 2\delta$, causing all correct processes to enter $e_{\mathit{final}}$ by time $t_{e_{\mathit{final}}} + 2\delta$.

\smallskip
\noindent \textbf{Result 5:} \emph{In every view of $e_{\mathit{final}}$, processes overlap for (at least) $\Delta$ time. In other words, there exists a synchronization time $t_s \leq t_{e_{\mathit{final}}} + \mathit{epoch\_duration} - \Delta$.}
\smallskip
\\ By Result 3, no future epoch can be entered before time $t_{e_{\mathit{final}}} + \mathit{epoch\_duration}$.
This is precisely enough time for the first correct process (the one to enter $e_{\mathit{final}}$ at $t_{e_{\mathit{final}}}$) to go through all $f+1$ views of $e_{\mathit{final}}$, spending $\mathit{view\_duration}$ time in each view.
Since clocks do not drift after $\mathit{GST}$ and processes spend the same amount of time in each view, the maximum delay of $2\delta$ between processes (Result 4) applies to every view in $e_{\mathit{final}}$.
Thus, all correct processes overlap with each other for (at least) $\mathit{view\_duration} - 2\delta = \Delta$ time in every view of $e_{\mathit{final}}$.
As the $\mathsf{leader}(\cdot)$ function is round-robin, at least one of the $f+1$ views must have a correct leader.
Therefore, synchronization must happen within epoch $e_{\mathit{final}}$, i.e., there is a synchronization time $t_s$ such that $t_{e_{\mathit{final}}} + \Delta \leq t_s + \Delta \leq t_{e_{\mathit{final}}} + \mathit{epoch\_duration}$.



\smallskip
\noindent \textbf{Result 6:} $t_{e_\mathit{final}} \leq \mathit{GST} + \mathit{epoch\_duration} + 4\delta$\emph{.}
\smallskip
\\ If $e_{\mathit{final}} = 1$, all correct processes started executing \rare at time $\mathit{GST}$.
Hence, $t_{e_{\mathit{final}}} = \mathit{GST}$.
Therefore, the result trivially holds in this case.

Let $e_{\mathit{final}} > 1$; recall that $e_{\mathit{final}} = e_{\mathit{max}} + 1$.
(1) By time $\mathit{GST} + \delta$, every correct process receives an \textsc{enter-epoch} message for $e_{\mathit{max}}$ (line~\ref{line:receive_epoch_over_complete}) as the first correct process to enter $e_{\mathit{max}}$ has broadcast this message before $\mathit{GST}$ (line~\ref{line:broadcast_epoch_over_complete}).
Hence, (2) by time $\mathit{GST} + 2\delta$, every correct process enters $e_{\mathit{max}}$.\footnote{If $e_{\mathit{max}} = 1$, every correct process enters $e_{\mathit{max}}$ by time $\mathit{GST}$.}
Then, (3) every correct process broadcasts an \textsc{epoch-completed} message for $e_{\mathit{max}}$ at time $\mathit{GST} + \mathit{epoch\_duration} + 2\delta$ (line~\ref{line:broadcast_epoch_over}), at latest.
(4) By time $\mathit{GST} + \mathit{epoch\_duration} + 3\delta$, every correct process receives $2f + 1$ \textsc{epoch-completed} messages for $e_{\mathit{max}}$ (line~\ref{line:receive_epoch_over}), and triggers the $\mathsf{measure}(\delta)$ method of $\mathit{dissemination\_timer}$ (line~\ref{line:sync_timer_measure}).
Therefore, (5) by time $\mathit{GST} + \mathit{epoch\_duration} + 4\delta$, every correct process enters $e_{\mathit{max}} + 1 = e_{\mathit{final}}$.
\Cref{fig:raresync_proof} depicts this scenario.

Note that for the previous sequence of events \emph{not} to unfold would imply an even lower bound on $t_{e_{\mathit{final}}}$: a correct process would have to receive $2f+1$ \textsc{epoch-completed} messages for $e_{\mathit{max}}$ or an \textsc{enter-epoch} message for $e_{\mathit{max}} + 1 = e_{\mathit{final}}$ before step (4) (i.e., before time $\mathit{GST} + \mathit{epoch\_duration} + 3\delta$), thus showing that $t_{e_{\mathit{final}}} < \mathit{GST} + \mathit{epoch\_duration} + 4\delta$.

\smallskip
\noindent \textbf{Latency:} \emph{Latency complexity of \rare is $O(f)$}. 
\smallskip
\\ By Result 5, $t_s \leq t_{e_{\mathit{final}}} + \mathit{epoch\_duration} - \Delta$.
By Result 6, $t_{e_{\mathit{final}}} \leq \mathit{GST} + \mathit{epoch\_duration} + 4\delta$.
Therefore, $t_s \leq \mathit{GST} + \mathit{epoch\_duration} + 4\delta + \mathit{epoch\_duration} - \Delta = \mathit{GST} + 2\mathit{epoch\_duration} + 4\delta - \Delta$.
Hence, $t_s + \Delta - \mathit{GST} \leq 2\mathit{epoch\_duration} + 4\delta = O(f)$.

\smallskip
\noindent \textbf{Communication:} \emph{Communication complexity of \rare is $O(n^2)$}. 
\smallskip
\\ Roughly speaking, every correct process will have entered $e_{\mathit{max}}$ (or potentially $e_{\mathit{final}} = e_{\mathit{max}} + 1$) by time $\mathit{GST} + 2\delta$ (as seen in the proof of Result 6).
From then on, it will enter at most one other epoch ($e_{\mathit{final}}$) before synchronizing (which is completed by time $t_s + \Delta$).
As for the time interval $[\mathit{GST}, \mathit{GST}+2\delta)$, due to $\mathit{dissemination\_timer}$'s interval of $\delta$, a correct process can enter (at most) two other epochs during this period.
Therefore, a correct process can enter (and send messages for) at most $O(1)$ epochs between $\mathit{GST}$ and $t_s + \Delta$.
The individual communication cost of a correct process is bounded by $O(n)$ words per epoch: $O(n)$ \textsc{epoch-completed} messages (each with a single word), and $O(n)$ \textsc{enter-epoch} messages (each with a single word, as a threshold signature counts as a single word).
Thus, the communication complexity of \rare is $O(n^2) = O(1) \cdot O(n) \cdot n$.

\begin{figure}[h]
    \centering
    \includegraphics[scale=0.2]{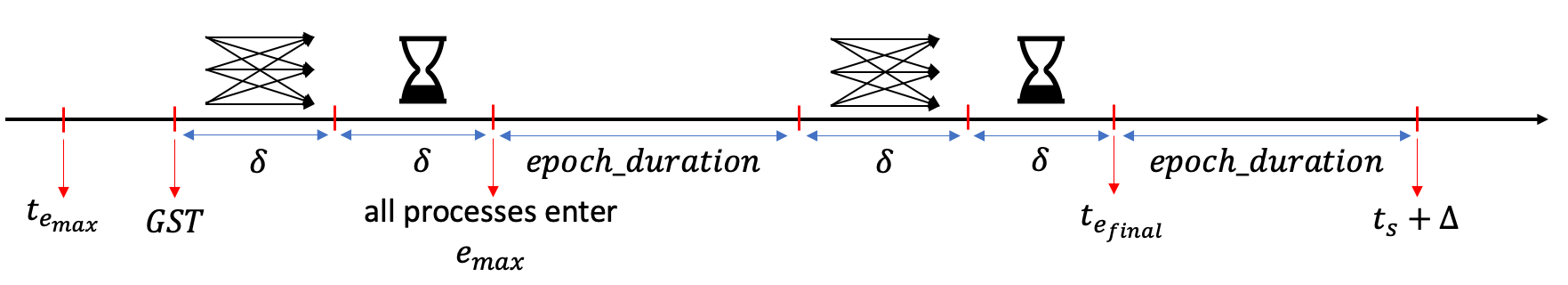}
    \caption{Worst-case latency of \rare: $t_s + \Delta - \mathit{GST} \leq 2\mathit{epoch\_duration} + 4\delta$.}
    \label{fig:raresync_proof}
\end{figure}

\smallskip
The formal proof of the following theorem is delegated to \Cref{section:appendix_rare}.

\begin{theorem}
\rare is a partially synchronous view synchronizer with (1) $O(n^2)$ communication complexity, and (2) $O(f)$ latency complexity.
\end{theorem}

\section{\sname} \label{section:quad}

This section introduces \sname, a partially synchronous Byzantine consensus protocol with optimal resilience~\cite{Dwork1988}.
\sname simultaneously achieves (1) $O(n^2)$ communication complexity, matching the Dolev-Reischuk bound~\cite{Dolev1985}, and (2) $O(f)$ latency complexity, matching the Dolev-Strong bound~\cite{Dolev1983}.

First, we present \name, a partially synchronous Byzantine consensus protocol ensuring weak validity (\Cref{subsection:quad_algorithm}).
\name achieves quadratic communication complexity and linear latency complexity.
Then, we construct \sname by adding a simple preprocessing phase to \name (\Cref{subsection:squad}).


\subsection{\name} \label{subsection:quad_algorithm}

\name is a partially synchronous Byzantine consensus protocol satisfying the weak validity property:
\begin{compactitem}
    \item \emph{Weak validity:} If all processes are correct, then a value decided by a process was proposed.
\end{compactitem}
\name achieves (1) quadratic communication complexity, and (2) linear latency complexity.
Interestingly, the Dolev-Reischuk lower bound~\cite{Dolev1985} does not apply to Byzantine protocols satisfying weak validity; hence, we do not know whether \name has optimal communication complexity.
As explained in \Cref{subsection:squad}, we accompany \name by a preprocessing phase to obtain \sname.

\name (\Cref{algorithm:quad}) uses the same view core module as HotStuff~\cite{Yin2019}, i.e., the view logic of \name is identical to that of HotStuff.
Moreover, \name uses \rare as its view synchronizer, achieving synchronization with $O(n^2)$ communication.
The combination of HotStuff's view core and \rare ensures that each correct process sends $O(n)$ words after $\mathit{GST}$ (and before the decision), i.e., $C = O(n)$ in \name.
Following the formula introduced in \Cref{section:introduction}, \name indeed achieves $n \cdot C + S = n \cdot O(n) + O(n^2) = O(n^2)$ communication complexity.
Due to the linear latency of \rare, \name also achieves $O(f)$ latency complexity.

\smallskip
\noindent \textbf{View core.}
We now give a brief description of the view core module of \name.
The complete pseudocode of this module can be found in \Cref{section:appendix_quad} (and in~\cite{Yin2019}).

Each correct process keeps track of two critical variables: (1) the \emph{prepare} quorum certificate (QC), and (2) the \emph{locked} QC.
Each of these represents a process' estimation of the value that will be decided, although with a different degree of certainty.
For example, if a correct process decides a value $v$, it is guaranteed that (at least) $f + 1$ correct processes have $v$ in their locked QC.
Moreover, it is ensured that no correct process updates (from this point onward) its prepare or locked QC to any other value, thus ensuring agreement.
Lastly, a QC is a (constant-sized) threshold signature.

The structure of a view follows the ``all-to-leader, leader-to-all'' communication pattern.
Specifically, each view is comprised of the following four phases:
\begin{compactenum}
    \item \textbf{Prepare:} A process sends to the leader a \textsc{view-change} message containing its prepare QC.
    Once the leader receives $2f + 1$ \textsc{view-change} messages, it selects the prepare QC from the ``latest'' view.
    The leader sends this QC to all processes via a \textsc{prepare} message.
    
    Once a process receives the \textsc{prepare} message from the leader, it supports the received prepare QC if (1) the received QC is consistent with its locked QC, or (2) the received QC is ``more recent'' than its locked QC.
    If the process supports the received QC, it acknowledges this by sending a \textsc{prepare-vote} message to the leader.
    
    \item \textbf{Precommit:}
    Once the leader receives $2f + 1$ \textsc{prepare-vote} messages, it combines them into a cryptographic proof $\sigma$ that ``enough'' processes have supported its ``prepare-phase'' value; $\sigma$ is a threshold signature.
    Then, it disseminates $\sigma$ to all processes via a \textsc{precommit} message.
    Once a process receives the \textsc{precommit} message carrying $\sigma$, it updates its prepare QC to $\sigma$ and sends back to the leader a \textsc{precommit-vote} message. 
    
    \item \textbf{Commit:}
    Once the leader receives $2f + 1$ \textsc{precommit-vote} messages, it combines them into a cryptographic proof $\sigma'$ that ``enough'' processes have adopted its ``precommit-phase'' value (by updating their prepare QC); $\sigma'$ is a threshold signature.
    Then, it disseminates $\sigma'$ to all processes via a \textsc{commit} message.
    Once a process receives the \textsc{commit} message carrying $\sigma'$, it updates its locked QC to $\sigma'$ and sends back to the leader a \textsc{commit-vote} message. 
    
    \item \textbf{Decide:}
    Once the leader receives $2f + 1$ \textsc{commit-vote} messages, it combines them into a threshold signature $\sigma''$, and relays $\sigma''$ to all processes via a \textsc{decide} message.
    When a process receives the \textsc{decide} message carrying $\sigma''$, it decides the value associated with $\sigma''$.
\end{compactenum}
As a consequence of the ``all-to-leader, leader-to-all'' communication pattern and the constant size of messages, the leader of a view sends $O(n)$ words, while a non-leader process sends $O(1)$ words.

\smallskip
The view core module provides the following interface:
\begin{compactitem}
    \item \textbf{Request} $\mathsf{start\_executing(View} \text{ } v\mathsf{)}$: The view core starts executing the logic of view $v$ and abandons the previous view.
    Concretely, it stops accepting and sending messages for the previous view, and it starts accepting, sending, and replying to messages for view $v$. 
    The state of the view core is kept across views (e.g., the prepare and locked QCs).
    
    \item \textbf{Indication} $\mathsf{decide(Value} \text{ } \mathit{decision}\mathsf{)}$: The view core decides value $\mathit{decision}$ (this indication is triggered at most once).
\end{compactitem}

\smallskip
\noindent \textbf{Protocol description.}
The protocol (\Cref{algorithm:quad}) amounts to a composition of \rare and the aforementioned view core.
Since the view core requires $8$ communication steps in order for correct processes to decide, a synchronous overlap of $8\delta$ is sufficient.
Thus, we parameterize \rare with $\Delta = 8\delta$ (line~\ref{line:init_rare_sync}).
In short, the view core is subservient to \rare, i.e., when \rare triggers the $\mathsf{advance(}v\mathsf{)}$ event (line~\ref{line:synchronizer_advance}), the view core starts executing the logic of view $v$ (line~\ref{line:start_executing_logic}).
Once the view core decides (line~\ref{line:view_core_decide}), \name decides (line~\ref{line:quad_decide}).

\begin{algorithm}
\caption{\name: Pseudocode (for process $P_i$)}
\label{algorithm:quad}
\begin{algorithmic} [1]
\State \textbf{Modules:}
\State \hskip2em $\mathsf{View\_Core}$ $\mathit{core}$ 
\State \hskip2em $\mathsf{View\_Synchronizer}$ $\mathit{synchronizer} \gets \text{\rare}(\Delta = 8\delta)$ \label{line:init_rare_sync}

\smallskip
\State \textbf{upon} $\mathsf{init}(\mathsf{Value} \text{ } \mathit{proposal})$: \BlueComment{propose value $\mathit{proposal}$}
\State \hskip2em $\mathit{core}.\mathsf{init(}\mathit{proposal}\mathsf{)}$ \BlueComment{initialize the view core with the proposal}
\State \hskip2em $\mathit{synchronizer}.\mathsf{init}$ \BlueComment{start \rare}

\smallskip
\State \textbf{upon} $\mathit{synchronizer.}\mathsf{advance(View} \text{ } v\mathsf{)}$: \label{line:synchronizer_advance}
\State \hskip2em $\mathit{core}.\mathsf{start\_executing(}v\mathsf{)}$ \label{line:start_executing_logic}

\smallskip
\State \textbf{upon} $\mathit{core.}\mathsf{decide(Value} \text{ }\mathit{decision}\mathsf{)}$: \label{line:view_core_decide}
\State \hskip2em \textbf{trigger} $\mathsf{decide(}\mathit{decision}\mathsf{)}$ \label{line:quad_decide} \BlueComment{decide value $\mathit{decision}$}
\end{algorithmic}
\end{algorithm}

\smallskip
\noindent \textbf{Proof sketch.}
The agreement and weak validity properties of \name are ensured by the view core's implementation.
As for the termination property, the view core, and therefore \name, is guaranteed to decide as soon as processes have synchronized in the same view with a correct leader for $\Delta = 8\delta$ time at or after $\mathit{GST}$.
Since \rare ensures the eventual synchronization property, this eventually happens, which implies that \name satisfies termination.
As processes synchronize within $O(f)$ time after $\mathit{GST}$, the latency complexity of \name is $O(f)$.

As for the total communication complexity, it is the sum of the communication complexity of (1) \rare, which is $O(n^2)$, and (2) the view core, which is also $O(n^2)$. 
The view core's complexity is a consequence of the fact that:
\begin{compactitem}
    \item each process executes $O(1)$ epochs between $\mathit{GST}$ and the time by which every process decides,
    \item each epoch has $f + 1$ views,
    \item a process can be the leader in only one view of any epoch, and
    \item a process sends $O(n)$ words in a view if it is the leader, and $O(1)$ words otherwise, for an average of $O(1)$ words per view in any epoch.
\end{compactitem}
Thus, the view core's communication complexity is $O(n^2) = O(1) \cdot (f+1) \cdot O(1) \cdot n$.
Therefore, \name indeed achieves $O(n^2)$ communication complexity. 
The formal proof of the following theorem can be found in \Cref{section:appendix_quad}.

\begin{theorem}
\name is a Byzantine consensus protocol ensuring weak validity with (1) $O(n^2)$ communication complexity, and (2) $O(f)$ latency complexity.
\end{theorem}

\subsection{\sname: Protocol Description} \label{subsection:squad}

At last, we present \sname, which we derive from \name.


\smallskip
\noindent \textbf{Deriving \sname from \name.}
Imagine a locally-verifiable, \emph{constant-sized} cryptographic proof $\sigma_v$ vouching that value $v$ is \emph{valid}.
Moreover, imagine that it is impossible, in the case in which all correct processes propose $v$ to \name, for any process to obtain a proof for a value different from $v$:
\begin{compactitem}
    \item Computability: If all correct processes propose $v$ to \name, then no process (even if faulty) obtains a cryptographic proof $\sigma_{v'}$ for a value $v' \neq v$.
\end{compactitem}
If such a cryptographic primitive were to exist, then the \name protocol could be modified in the following manner in order to satisfy the validity property introduced in \Cref{section:introduction}:
\begin{compactitem}
    \item A correct process accompanies each value by a cryptographic proof that the value is valid.
    
    \item A correct process ignores any message with a value not accompanied by the value's proof.
\end{compactitem}
Suppose that all correct processes propose the same value $v$ and that a correct process $P_i$ decides $v'$ from the modified version of \name.
Given that $P_i$ ignores messages with non-valid values, $P_i$ has obtained a proof for $v'$ before deciding.
The computability property of the cryptographic primitive guarantees that $v' = v$, implying that validity is satisfied.
Given that the proof is of constant size, the communication complexity of the modified version of \name remains $O(n^2)$.

Therefore, the main challenge in obtaining \sname from \name, while preserving \name's complexity, lies in implementing the introduced cryptographic primitive.

\smallskip
\noindent \textbf{Certification phase.}
\sname utilizes its \emph{certification phase} (\Cref{algorithm:groundwork}) to obtain the introduced constant-sized cryptographic proofs; we call these proofs \emph{certificates}.\footnote{Note the distinction between certificates and prepare and locked QCs of the view core.}
Formally, $\mathsf{Certificate}$ denotes the set of all certificates.
Moreover, we define a locally computable function $\mathsf{verify}\text{: } \mathsf{Value} \times \mathsf{Certificate} \to \{\mathit{true}, \mathit{false}\}$.
We require the following properties to hold:
\begin{compactitem}
    \item \emph{Computability:} If all correct processes propose the same value $v$ to \sname, then no process (even if faulty) obtains a certificate $\sigma_{v'}$ with $\mathsf{verify}(v', \sigma_{v'}) = \mathit{true}$ and $v' \neq v$.
    
    \item \emph{Liveness:} Every correct process eventually obtains a certificate $\sigma_v$ such that $\mathsf{verify}(v, \sigma_v) = \mathit{true}$, for some value $v$.
\end{compactitem}
The computability property states that, if all correct processes propose the same value $v$ to \sname, then no process (even if Byzantine) can obtain a certificate for a value different from $v$.
The liveness property ensures that all correct processes eventually obtain a certificate.
Hence, if all correct processes propose the same value $v$, all correct processes eventually obtain a certificate for $v$ and no process obtains a certificate for a different value.

In order to implement the certification phase, we assume an $(f + 1, n)$-threshold signature scheme (see \Cref{section:model}) used throughout the entirety of the certification phase.
The $(f + 1, n)$-threshold signature scheme allows certificates to count as a single word, as each certificate is a threshold signature.
Finally, in order to not disrupt \name's communication and latency, the certification phase itself incurs $O(n^2)$ communication and $O(1)$ latency.

\begin{algorithm} [ht]
\caption{Certification Phase: Pseudocode (for process $P_i$)}
\label{algorithm:groundwork}
\begin{algorithmic} [1]

\State \textbf{upon} $\mathsf{init(Value} \text{ } \mathit{proposal}\mathsf{)}$: \BlueComment{propose value $\mathit{proposal}$} \label{line:certification_phase_start}

\State \hskip2em \textcolor{blue}{\(\triangleright\) inform other processes that $\mathit{proposal}$ was proposed}
\State \hskip2em \textbf{broadcast} $\langle \textsc{disclose}, \mathit{proposal}, \mathit{ShareSign}_i(\mathit{proposal}) \rangle$ \label{line:send_disclose}

\smallskip
\State \textbf{upon} exists $\mathsf{Value} \text{ } v$ such that $\langle \textsc{disclose}, v, \mathsf{P\_Signature} \text{ } \mathit{sig}\rangle$ is received from $f+1$ processes: \label{line:f+1_disclose}
\State \hskip2em \textcolor{blue}{\(\triangleright\) a certificate for $v$ is obtained}
\State \hskip2em $\mathsf{Certificate} \text{ } \sigma_v \gets \mathit{Combine}\big(\{\mathit{sig} \,|\, \mathit{sig} \text{ is received in a } \textsc{disclose} \text{ message}\}\big)$ \label{line:combine_certification_phase}
\State \hskip2em \textbf{broadcast} $\langle \textsc{certificate}, v, \sigma_v \rangle$ \BlueComment{disseminate the certificate} \label{line:broadcast_certificate_1}
\State \hskip2em exit the certification phase

\smallskip
\State \textbf{upon} for the first time (1) \textsc{disclose} message is received from $2f + 1$ processes, and (2) not exist $\mathsf{Value} \text{ } v$ such that $\langle \textsc{disclose}, v, \mathsf{P\_Signature} \text{ } \mathit{sig}\rangle$ is received from $f + 1$ processes: \label{line:rule_allow_any}
\State \hskip2em \textcolor{blue}{\(\triangleright\) inform other processes that any value can be ``accepted''}
\State \hskip2em \textbf{broadcast} $\langle \textsc{allow-any}, \mathit{ShareSign}_i(\text{``any value''}) \rangle$ \label{line:broadcast_allow_any}

\smallskip
\State \textbf{upon} $\langle \textsc{allow-any}, \mathsf{P\_Signature} \text{ } \mathit{sig} \rangle$ is received from $f + 1$ processes \label{line:f+1_allow_any}:
\State \hskip2em \textcolor{blue}{\(\triangleright\) a certificate for ``any value'' is obtained}
\State \hskip2em $\mathsf{Certificate} \text{ } \sigma_{\bot} \gets \mathit{Combine}\big(\{\mathit{sig} \,|\, \mathit{sig} \text{ is received in an } \textsc{allow-any} \text{ message}\}\big)$ \label{line:combine_allow_any}
\State \hskip2em \textbf{broadcast} $\langle \textsc{certificate}, \bot, \sigma_{\bot} \rangle$ \BlueComment{disseminate the certificate} \label{line:broadcast_certificate_2}
\State \hskip2em exit the certification phase

\smallskip
\State \textcolor{blue}{\(\triangleright\) a certificate for $v$ is obtained; $v$ can be $\bot$, meaning that $\sigma_v$ vouches for any value}
\State \textbf{upon} reception of $\langle \textsc{certificate}, \mathsf{Value} \text{ } v, \mathsf{Certificate} \text{ } \sigma_v \rangle$: \label{line:receive_certificate}
\State \hskip2em \textbf{broadcast} $\langle \textsc{certificate}, v, \sigma_v \rangle$ \BlueComment{disseminate the certificate} \label{line:broadcast_certificate_3}
\State \hskip2em exit the certification phase

\smallskip
\State \textbf{function} $\mathsf{verify(Value} \text{ } v, \mathsf{Certificate} \text{ } \sigma)$: \label{line:function_verify}
\State \hskip2em \textbf{if} $\mathit{CombinedVerify}(\text{``any value''}, \sigma) = \mathit{true}$: \textbf{return} $\mathit{true}$ \label{line:verify_any_value}
\State \hskip2em \textbf{else if} $\mathit{CombinedVerify}(v, \sigma) = \mathit{true}$: \textbf{return} $\mathit{true}$ \label{line:verify_v}
\State \hskip2em \textbf{else return} $\mathit{false}$

\end{algorithmic}
\end{algorithm}

A certificate $\sigma$ vouches for a value $v$ (the $\mathsf{verify}(\cdot)$ function at line~\ref{line:function_verify}) if (1) $\sigma$ is a threshold signature of the predefined string ``any value'' (line~\ref{line:verify_any_value}), or (2) $\sigma$ is a threshold signature of $v$ (line~\ref{line:verify_v}).
Otherwise, $\mathsf{verify}(v, \sigma)$ returns $\mathit{false}$.

Once $P_i$ enters the certification phase (line~\ref{line:certification_phase_start}), $P_i$ informs all processes about the value it has proposed by broadcasting a \textsc{disclose} message (line~\ref{line:send_disclose}).
Process $P_i$ includes a partial signature of its proposed value in the message.
If $P_i$ receives \textsc{disclose} messages for the same value $v$ from $f + 1$ processes (line~\ref{line:f+1_disclose}), $P_i$ combines the received partial signatures into a threshold signature of $v$ (line~\ref{line:combine_certification_phase}), which represents a certificate for $v$.
To ensure liveness, $P_i$ disseminates the certificate (line~\ref{line:broadcast_certificate_1}).

If $P_i$ receives $2f + 1$ \textsc{disclose} messages and there does not exist a ``common'' value received in $f + 1$ (or more) \textsc{disclose} messages (line~\ref{line:rule_allow_any}), the process concludes that it is fine for a certificate for \emph{any} value to be obtained.
Therefore, $P_i$ broadcasts an \textsc{allow-any} message containing a partial signature of the predefined string ``any value'' (line~\ref{line:broadcast_allow_any}).

If $P_i$ receives $f + 1$ \textsc{allow-any} messages (line~\ref{line:f+1_allow_any}), it combines the received partial signatures into a certificate that vouches for \emph{any} value (line~\ref{line:combine_allow_any}), and it disseminates the certificate (line~\ref{line:broadcast_certificate_2}).
Since \textsc{allow-any} messages are received from $f + 1$ processes, there exists a correct process that has verified that it is indeed fine for such a certificate to exist.

If, at any point, $P_i$ receives a certificate (line~\ref{line:receive_certificate}), it adopts the certificate, and disseminates it (line~\ref{line:broadcast_certificate_3}) to ensure liveness.

Given that each message of the certification phase contains a single word, the certification phase incurs $O(n^2)$ communication.
Moreover, each correct process obtains a certificate after (at most) $2 = O(1)$ rounds of communication.
Therefore, the certification phase incurs $O(1)$ latency.

We explain below why the certification phase (\Cref{algorithm:groundwork}) ensures computability and liveness:
\begin{compactitem}
    \item Computability: If all correct processes propose the same value $v$ to \sname, all correct processes broadcast a \textsc{disclose} message for $v$ (line~\ref{line:send_disclose}).
    Since $2f + 1$ processes are correct, no process obtains a certificate $\sigma_{v'}$ for a value $v' \neq v$ such that $\mathit{CombinedVerify}(v', \sigma_{v'}) = \mathit{true}$ (line~\ref{line:verify_v}).
    
    Moreover, as every correct process receives $f + 1$ \textsc{disclose} messages for $v$ within any set of $2f + 1$ received \textsc{disclose} messages, no correct process sends an \textsc{allow-any} message (line~\ref{line:broadcast_allow_any}).
    Hence, no process obtains a certificate $\sigma_{\bot}$ such that $\mathit{CombinedVerify}(\text{``any value''}, \sigma_{\bot}) = \mathit{true}$ (line~\ref{line:verify_any_value}).
    Thus, computability is ensured.
    
    \item Liveness: If a correct process receives $f + 1$ \textsc{disclose} messages for a value $v$ (line~\ref{line:f+1_disclose}), the process obtains a certificate for $v$ (line~\ref{line:combine_certification_phase}).
    Since the process disseminates the certificate (line~\ref{line:broadcast_certificate_1}), every correct process eventually obtains a certificate (line~\ref{line:receive_certificate}), ensuring liveness in this scenario.
    
    Otherwise, all correct processes broadcast an \textsc{allow-any} message (line~\ref{line:broadcast_allow_any}).
    Since there are at least $2f + 1$ correct processes, every correct process eventually receives $f + 1$ \textsc{allow-any} messages (line~\ref{line:f+1_allow_any}), thus obtaining a certificate.
    Hence, liveness is satisfied in this case as well.
\end{compactitem}

\smallskip
\noindent \textbf{\sname = Certification phase + \name.}
We obtain \sname by combining the certification phase with \name.
The pseudocode of \sname is given in \Cref{algorithm:squad}.

\begin{algorithm}
\caption{\sname: Pseudocode (for process $P_i$)}
\label{algorithm:squad}
\begin{algorithmic} [1]
\State \textbf{upon} $\mathsf{init(Value} \text{ } \mathit{proposal}\mathsf{)}$: \BlueComment{propose value $\mathit{proposal}$}
\State \hskip2em start the certification phase with $\mathit{proposal}$ \label{line:start_certification_phase_squad}

\smallskip
\State \textbf{upon} exiting the certification phase with a certificate $\sigma_v$ for a value $v$: \label{line:obtained_certificate_squad}

\State \hskip2em \textcolor{blue}{\(\triangleright\) in $\name_{\mathit{cer}}$, processes ignore messages with values not accompanied by their certificates}
\State \hskip2em start executing $\name_{\mathit{cer}}$ with the proposal $(v, \sigma_v)$ \label{line:start_quad_squad}

\smallskip
\State \textbf{upon} $\name_{\mathit{cer}}$ decides $\mathsf{Value} \text{ } \mathit{decision}$: \label{line:decide_quad_squad}
\State \hskip2em \textbf{trigger} $\mathsf{decide}(\mathit{decision})$ \label{line:decide_squad} \BlueComment{decide value $\mathit{decision}$}
\end{algorithmic}
\end{algorithm}


A correct process $P_i$ executes the following steps in \sname:
\begin{compactenum}
    \item $P_i$ starts executing the certification phase with its proposal (line~\ref{line:start_certification_phase_squad}).
    
    \item Once the process exits the certification phase with a certificate $\sigma_v$ for a value $v$, it proposes $(v, \sigma_v)$ to $\name_{\mathit{cer}}$, a version of \name ``enriched'' with certificates (line~\ref{line:start_quad_squad}).
    While executing $\name_{\mathit{cer}}$, correct processes \emph{ignore} messages containing values not accompanied by their certificates.
    
    \item Once $P_i$ decides from $\name_{\mathit{cer}}$ (line~\ref{line:decide_quad_squad}), $P_i$ decides the same value from \sname (line~\ref{line:decide_squad}).
\end{compactenum}

\smallskip
The proof of the following theorem is delegated to \Cref{section:squad_appendix}.

\begin{theorem}
\sname is a Byzantine consensus protocol with (1) $O(n^2)$ communication complexity, and (2) $O(f)$ latency complexity.
\end{theorem}

\section{Concluding Remarks} \label{section:conclusion}

This paper shows that the Dolev-Reischuk lower bound can be met by a partially synchronous Byzantine consensus protocol.
Namely, we introduce \sname, an optimally-resilient partially synchronous Byzantine consensus protocol with optimal $O(n^2)$ communication complexity, and optimal $O(f)$ latency complexity.
\sname owes its complexity to \rare, an ``epoch-based'' view synchronizer ensuring synchronization with quadratic communication and linear latency in partial synchrony.
In the future, we aim to address the following limitations of \rare.

\smallskip
\noindent \textbf{Lack of adaptiveness.}
\rare is not \emph{adaptive}, i.e., its complexity does not depend on the \emph{actual} number $b$, but rather on the upper bound $f$, of Byzantine processes.
Consider a scenario $S$ in which all processes are correct; we separate them into three disjoint groups: (1) group $A$, with $|A| = f$, (2) group $B$, with $|B| = f$, and (3) group $C$, with $|C| = f + 1$.
At $\mathit{GST}$, group $A$ is in the first view of epoch $e_{\mathit{max}}$, group $B$ is in the second view of $e_{\mathit{max}}$, and group $C$ is in the third view of $e_{\mathit{max}}$.\footnote{Recall that $e_{\mathit{max}}$ is the greatest epoch entered by a correct process before $\mathit{GST}$; see \Cref{subsection:rare_proof_sketch}.}
Unfortunately, it is impossible for processes to synchronize in epoch $e_{\mathit{max}}$.
Hence, they will need to wait for the end of epoch $e_{\mathit{max}}$ in order to synchronize in the next epoch: thus, the latency complexity is $O(f)$ (since $e_{\mathit{max}}$ has $f + 1$ views) and the communication complexity is $O(n^2)$ (because of the ``all-to-all'' communication step at the end of $e_{\mathit{max}}$).
In contrast, the view synchronizer presented in~\cite{Naor2020} achieves $O(1)$ latency and $O(n)$ communication complexity in $S$.

\smallskip
\noindent \textbf{Suboptimal expected complexity.}
A second limitation of \rare is that its \emph{expected complexity} is the same as its worst-case complexity.
Namely, the expected complexity considers a weaker adversary which does not have a knowledge of the $\mathsf{leader}(\cdot)$ function. 
Therefore, this adversary is unable to corrupt $f$ processes that are scheduled to be leaders right after $\mathit{GST}$.

As the previously introduced scenario $S$ does not include any Byzantine process, we can analyze it for the expected complexity of \rare.
Therefore, the expected latency complexity of \rare is $O(f)$ and the expected communication complexity of \rare is $O(n^2)$.
On the other hand, the view synchronizer of Naor and Keidar~\cite{Naor2020} achieves $O(1)$ expected latency complexity and $O(n)$ expected communication complexity.


\smallskip
\noindent \textbf{Limited clock drift tolerance.}
A third limitation of \rare is that its latency is susceptible to clock drifts.
Namely, let $\phi > 1$ denote the bound on clock drifts after $\mathit{GST}$.
To accommodate for the bounded clock drifts after $\mathit{GST}$, \rare increases the duration of a view. 
The duration of the $i$-th view of an epoch becomes $\phi^i \cdot \mathit{view\_duration}$ (instead of only $\mathit{view\_duration}$).
Thus, the latency complexity of \rare becomes $O(f \cdot \phi^f)$.

\paragraph*{Acknowledgments}

The authors would like to thank Gregory Chockler and Alexey Gotsman for helpful conversations. This work is supported in part by the ARC
Future Fellowship funding scheme (\#180100496).

\bibstyle{plainurl}
\bibliography{references}

\begin{thebibliography}{10}

\bibitem{Abraham2019b}
Ittai Abraham, T-H.~Hubert Chan, Danny Dolev, Kartik Nayak, Rafael Pass, Ling
  Ren, and Elaine Shi.
\newblock {Communication Complexity of Byzantine Agreement, Revisited}.
\newblock In {\em Proceedings of the 2019 ACM Symposium on Principles of
  Distributed Computing}, PODC '19, page 317–326, New York, NY, USA, 2019.
  Association for Computing Machinery.
\newblock \href {https://doi.org/10.1145/3293611.3331629}
  {\path{doi:10.1145/3293611.3331629}}.

\bibitem{DBLP:conf/wdag/AbrahamDN017}
Ittai Abraham, Srinivas Devadas, Kartik Nayak, and Ling Ren.
\newblock {Brief Announcement: Practical Synchronous Byzantine Consensus}.
\newblock In Andr{\'{e}}a~W. Richa, editor, {\em 31st International Symposium
  on Distributed Computing, {DISC} 2017, October 16-20, 2017, Vienna, Austria},
  volume~91 of {\em LIPIcs}, pages 41:1--41:4. Schloss Dagstuhl -
  Leibniz-Zentrum f{\"{u}}r Informatik, 2017.
\newblock \href {https://doi.org/10.4230/LIPIcs.DISC.2017.41}
  {\path{doi:10.4230/LIPIcs.DISC.2017.41}}.

\bibitem{abraham2018hot}
Ittai Abraham, Guy Gueta, and Dahlia Malkhi.
\newblock {Hot-Stuff the Linear, Optimal-Resilience, One-Message BFT Devil}.
\newblock {\em CoRR, abs/1803.05069}, 2018.

\bibitem{AbrahamJMMST21}
Ittai Abraham, Philipp Jovanovic, Mary Maller, Sarah Meiklejohn, Gilad Stern,
  and Alin Tomescu.
\newblock {Reaching Consensus for Asynchronous Distributed Key Generation}.
\newblock In Avery Miller, Keren Censor{-}Hillel, and Janne~H. Korhonen,
  editors, {\em {PODC} '21: {ACM} Symposium on Principles of Distributed
  Computing, Virtual Event, Italy, July 26-30, 2021}, pages 363--373. {ACM},
  2021.
\newblock \href {https://doi.org/10.1145/3465084.3467914}
  {\path{doi:10.1145/3465084.3467914}}.

\bibitem{solida}
Ittai Abraham, Dahlia Malkhi, Kartik Nayak, Ling Ren, and Alexander Spiegelman.
\newblock {Solida: {A} Blockchain Protocol Based on Reconfigurable Byzantine
  Consensus}.
\newblock In James Aspnes, Alysson Bessani, Pascal Felber, and Jo{\~{a}}o
  Leit{\~{a}}o, editors, {\em 21st International Conference on Principles of
  Distributed Systems, {OPODIS} 2017, Lisbon, Portugal, December 18-20, 2017},
  volume~95 of {\em LIPIcs}, pages 25:1--25:19. Schloss Dagstuhl -
  Leibniz-Zentrum f{\"{u}}r Informatik, 2017.

\bibitem{Abraham2019a}
Ittai Abraham, Dahlia Malkhi, and Alexander Spiegelman.
\newblock {Asymptotically Optimal Validated Asynchronous Byzantine Agreement}.
\newblock {\em Proceedings of the Annual ACM Symposium on Principles of
  Distributed Computing}, pages 337--346, 2019.

\bibitem{abraham2019asymptotically}
Ittai Abraham, Dahlia Malkhi, and Alexander Spiegelman.
\newblock {Asymptotically Optimal Validated Asynchronous Byzantine Agreement}.
\newblock In {\em Proceedings of the 2019 ACM Symposium on Principles of
  Distributed Computing (PODC)}, pages 337--346, 2019.

\bibitem{AD15}
Marcin Andrychowicz and Stefan Dziembowski.
\newblock {PoW-Based Distributed Cryptography with No Trusted Setup}.
\newblock In Rosario Gennaro and Matthew Robshaw, editors, {\em Advances in
  Cryptology - {CRYPTO} 2015 - 35th Annual Cryptology Conference, Santa
  Barbara, CA, USA, August 16-20, 2015, Proceedings, Part {II}}, volume 9216 of
  {\em Lecture Notes in Computer Science}, pages 379--399. Springer, 2015.
\newblock \href {https://doi.org/10.1007/978-3-662-48000-7\_19}
  {\path{doi:10.1007/978-3-662-48000-7\_19}}.

\bibitem{Antoniadis2021a}
Karolos Antoniadis, Antoine Desjardins, Vincent Gramoli, Rachid Guerraoui, and
  Igor Zablotchi.
\newblock {Leaderless Consensus}.
\newblock In {\em Proceedings - International Conference on Distributed
  Computing Systems}, volume 2021-July, pages 392--402, 2021.

\bibitem{Ben-Or1983b}
Michael Ben-Or.
\newblock {Another Advantage of Free Choice: Completely Asynchronous Agreement
  Protocols.}
\newblock {\em Proceedings of the Second Annual Symposium on Principles of
  Distributed Computing}, pages 27--30, 1983.

\bibitem{berman}
Piotr Berman, Juan~A. Garay, and Kenneth~J. Perry.
\newblock {Bit Optimal Distributed Consensus}.
\newblock {\em Computer Science: Research and Applications}, page 313–321,
  1992.

\bibitem{Bracha87}
Gabriel Bracha.
\newblock {Asynchronous Byzantine Agreement Protocols}.
\newblock {\em Inf. Comput.}, 75(2):130--143, 1987.
\newblock \href {https://doi.org/10.1016/0890-5401(87)90054-X}
  {\path{doi:10.1016/0890-5401(87)90054-X}}.

\bibitem{Bravo2020}
Manuel Bravo, Gregory Chockler, and Alexey Gotsman.
\newblock {Making Byzantine Consensus Live}.
\newblock In {\em 34th International Symposium on Distributed Computing
  (DISC)}, volume 179, pages 1--17, 2020.

\bibitem{Buchman2018}
Ethan Buchman, Jae Kwon, and Zarko Milosevic.
\newblock {The latest gossip on BFT consensus}.
\newblock pages 1--14, 2018.
\newblock URL: \url{https://arxiv.org/pdf/1807.04938.pdf}, \href
  {http://arxiv.org/abs/1807.04938} {\path{arXiv:1807.04938}}.

\bibitem{CKS05}
Christian Cachin, Klaus Kursawe, and Victor Shoup.
\newblock {Random Oracles in Constantinople: Practical Asynchronous Byzantine
  Agreement Using Cryptography}.
\newblock {\em J. Cryptol.}, 18(3):219--246, 2005.
\newblock \href {https://doi.org/10.1007/s00145-005-0318-0}
  {\path{doi:10.1007/s00145-005-0318-0}}.

\bibitem{Castro2002}
Miguel Castro and Barbara Liskov.
\newblock {Practical Byzantine Fault Tolerance}.
\newblock {\em ACM Trans. Comput. Syst.}, (February):359--368, 2002.

\bibitem{Chandra1996}
Tushar Chandra and Sam Toueg.
\newblock {Unreliable Failure Detectors for Reliable Distributed Systems}.
\newblock {\em Proceedings of the 10th ACM Symposium on Principles of
  Distributed Computing}, (2):225--267, 1996.

\bibitem{Chandra1992}
Tushar~Deepak Chandra, Vassos Hadzilacos, and Sam Toueg.
\newblock {The Weakest Failure Detector for Solving Consensus}.
\newblock {\em Proceedings of the Annual ACM Symposium on Principles of
  Distributed Computing}, 43(4):147--158, 1992.

\bibitem{Chen2018}
Jing Chen, Sergey Gorbunov, Silvio Micali, and Georgios Vlachos.
\newblock {Algorand Agreement: Super Fast and Partition Resilient Byzantine
  Agreement}.
\newblock {\em Cryptology ePrint Archive}, 377:1--10, 2018.
\newblock URL: \url{https://eprint.iacr.org/2018/377.pdf}.

\bibitem{DBLP:journals/sigact/CohenKN21}
Shir Cohen, Idit Keidar, and Oded Naor.
\newblock {Byzantine Agreement with Less Communication: Recent Advances}.
\newblock {\em {SIGACT} News}, 52(1):71--80, 2021.
\newblock \href {https://doi.org/10.1145/3457588.3457600}
  {\path{doi:10.1145/3457588.3457600}}.

\bibitem{Cohen2020}
Shir Cohen, Idit Keidar, and Alexander Spiegelman.
\newblock {Brief Announcement: Not a COINcidence: Sub-Quadratic Asynchronous
  Byzantine Agreement WHP}.
\newblock {\em Proceedings of the Annual ACM Symposium on Principles of
  Distributed Computing}, pages 175--177, 2020.

\bibitem{Crain2017a}
Tyler Crain, Vincent Gramoli, Mikel Larrea, and Michel Raynal.
\newblock {DBFT: Efficient Byzantine Consensus with a Weak Coordinator and its
  Application to Consortium Blockchains}.
\newblock In {\em 17th {\{}IEEE{\}} International Symposium on Network
  Computing and Applications, {\{}NCA{\}}}, pages 1--41, 2017.
\newblock \href {http://arxiv.org/abs/1702.03068} {\path{arXiv:1702.03068}}.

\bibitem{Dolev1983}
D.~Dolev and H.~R. Strong.
\newblock {Authenticated Algorithms for Byzantine Agreement}.
\newblock 12(4):656--666, 1983.

\bibitem{Dolev1995}
Danny Dolev, Joseph~Y. Halpern, Barbara Simons, and Ray Strong.
\newblock {Dynamic Fault-Tolerant Clock Synchronization}.
\newblock {\em Journal of the ACM (JACM)}, 42(1):143--185, 1995.

\bibitem{Dolev1985}
Danny Dolev and R{\"{u}}diger Reischuk.
\newblock {Bounds on information exchange for Byzantine agreement}.
\newblock {\em Journal of the ACM (JACM)}, 1985.

\bibitem{Dwork1988}
Cynthia Dwork, Lynch Nancy, and Larry Stockmeyer.
\newblock {Consensus in the Presence of Partial Synchrony}.
\newblock {\em Journal of the ACM (JACM)}, 35(2):288--323, 1988.

\bibitem{Fischer1985}
Michael~J. Fischer, Nancy~A. Lynch, and Michael~S. Paterson.
\newblock {Impossibility of Distributed Consensus with One Faulty Process}.
\newblock {\em Journal of the Association for Computing Machinery,},
  32(2):374--382, 1985.

\bibitem{Gafni1998}
Eli Gafni.
\newblock {Round-by-Round Fault Detectors: Unifying Synchrony and Asynchrony}.
\newblock {\em Proceedings of the Annual ACM Symposium on Principles of
  Distributed Computing}, pages 143--152, 1998.

\bibitem{GKLP18}
Juan~A. Garay, Aggelos Kiayias, Nikos Leonardos, and Giorgos Panagiotakos.
\newblock {Bootstrapping the Blockchain, with Applications to Consensus and
  Fast {PKI} Setup}.
\newblock In Michel Abdalla and Ricardo Dahab, editors, {\em Public-Key
  Cryptography - {PKC} 2018 - 21st {IACR} International Conference on Practice
  and Theory of Public-Key Cryptography, Rio de Janeiro, Brazil, March 25-29,
  2018, Proceedings, Part {II}}, volume 10770 of {\em Lecture Notes in Computer
  Science}, pages 465--495. Springer, 2018.
\newblock \href {https://doi.org/10.1007/978-3-319-76581-5\_16}
  {\path{doi:10.1007/978-3-319-76581-5\_16}}.

\bibitem{GolanGueta2019}
Guy {Golan Gueta}, Ittai Abraham, Shelly Grossman, Dahlia Malkhi, Benny Pinkas,
  Michael Reiter, Dragos~Adrian Seredinschi, Orr Tamir, and Alin Tomescu.
\newblock {SBFT: A Scalable and Decentralized Trust Infrastructure}.
\newblock {\em Proceedings - 49th Annual IEEE/IFIP International Conference on
  Dependable Systems and Networks, DSN 2019}, pages 568--580, 2019.

\bibitem{Gramoli20}
Vincent Gramoli.
\newblock {From blockchain consensus back to Byzantine consensus}.
\newblock {\em Future Gener. Comput. Syst.}, 107:760--769, 2020.

\bibitem{Guerraoui2004}
Rachid Guerraoui and Michel Raynal.
\newblock {The Information Structure of Indulgent Consensus}.
\newblock {\em {\{}IEEE{\}} Trans. Computers}, 53(4):453--466, 2004.

\bibitem{Keidar2006}
Idit Keidar and Alexander Shraer.
\newblock {Timeliness, Failure-Detectors, and Consensus Performance}.
\newblock {\em Proceedings of the Annual ACM Symposium on Principles of
  Distributed Computing}, 2006:169--178, 2006.

\bibitem{Kihlstrom2003a}
Kim~Potter Kihlstrom, Louise~E. Moser, and P.~M. Melliar-Smith.
\newblock {Byzantine Fault Detectors for Solving Consensus}.
\newblock {\em The Computer Journal}, 46(1):16--35, 2003.

\bibitem{King2011}
Valerie King and Jared Saia.
\newblock {Breaking the $O(n^2)$ Bit Barrier: Scalable Byzantine agreement with
  an Adaptive Adversary}.
\newblock {\em Journal of the ACM}, 58(4):1--24, 2011.

\bibitem{Kotla2009}
Ramakrishna Kotla, Lorenzo Alvisi, Mike Dahlin, Allen Clement, and Edmund Wong.
\newblock {Zyzzyva: Speculative Byzantine Fault Tolerance}.
\newblock {\em ACM Transactions on Computer Systems}, 27(4), 2009.

\bibitem{Kuznetsov2021}
Petr Kuznetsov, Andrei Tonkikh, and Yan~X. Zhang.
\newblock {Revisiting Optimal Resilience of Fast Byzantine Consensus}.
\newblock {\em Proceedings of the Annual ACM Symposium on Principles of
  Distributed Computing (PODC)}, 1(1):343--353, 2021.

\bibitem{Lamport1982}
Leslie Lamport, Robert Shostak, and Marshall Pease.
\newblock {The Byzantine Generals Problem}.
\newblock {\em {ACM} Trans. Program. Lang. Syst.}, 4(3):382--401, 1982.

\bibitem{LewisPye}
Andrew Lewis-Pye.
\newblock {Quadratic worst-case message complexity for State Machine
  Replication in the partial synchrony model}, 2022.
\newblock URL: \url{https://arxiv.org/abs/2201.01107}, \href
  {https://doi.org/10.48550/ARXIV.2201.01107}
  {\path{doi:10.48550/ARXIV.2201.01107}}.

\bibitem{Libert2016}
Beno{\^{i}}t Libert, Marc Joye, and Moti Yung.
\newblock {Born and Raised Distributively: Fully Distributed Non-Interactive
  Adaptively-Secure Threshold Signatures with Short Shares}.
\newblock {\em Theoretical Computer Science}, 645:1--24, 2016.

\bibitem{cloud}
JongBeom Lim, Taeweon Suh, Joon{-}Min Gil, and Heon{-}Chang Yu.
\newblock {Scalable and leaderless Byzantine consensus in cloud computing
  environments}.
\newblock {\em Inf. Syst. Frontiers}, 16(1):19--34, 2014.

\bibitem{Locher2020}
Thomas Locher.
\newblock {Fast Byzantine Agreement for Permissioned Distributed Ledgers}.
\newblock {\em Annual ACM Symposium on Parallelism in Algorithms and
  Architectures}, pages 371--382, 2020.

\bibitem{Lu2020}
Yuan Lu, Zhenliang Lu, Qiang Tang, and Guiling Wang.
\newblock {Dumbo-MVBA: Optimal Multi-Valued Validated Asynchronous Byzantine
  Agreement, Revisited}.
\newblock {\em Proceedings of the Annual ACM Symposium on Principles of
  Distributed Computing}, pages 129--138, 2020.

\bibitem{Martin2005}
Jean~Philippe Martin and Lorenzo Alvisi.
\newblock {Fast Byzantine Consensus}.
\newblock {\em Proceedings of the International Conference on Dependable
  Systems and Networks}, pages 402--411, 2005.

\bibitem{Micali2017ByzantineA}
Silvio Micali.
\newblock {Byzantine Agreement , Made Trivial}.
\newblock 2017.

\bibitem{Momose2021}
Atsuki Momose and Ling Ren.
\newblock {Optimal Communication Complexity of Authenticated Byzantine
  Agreement}.
\newblock In {\em 35th International Symposium on Distributed Computing
  (DISC)}, volume 209, pages 32:1--32:0. Schloss Dagstuhl – Leibniz-Zentrum
  f{\"{u}}r Informatik, Dagstuhl Publishing, Germany, 2021.

\bibitem{MostefaouiMR15}
Achour Most{\'{e}}faoui, Hamouma Moumen, and Michel Raynal.
\newblock {Signature-Free Asynchronous Binary Byzantine Consensus with t
  {\textless} n/3, O(n2) Messages, and {O(1)} Expected Time}.
\newblock {\em J. {ACM}}, 62(4):31:1--31:21, 2015.
\newblock \href {https://doi.org/10.1145/2785953} {\path{doi:10.1145/2785953}}.

\bibitem{Naor2021}
Oded Naor, Mathieu Baudet, Dahlia Malkhi, and Alexander Spiegelman.
\newblock {Cogsworth: Byzantine View Synchronization}.
\newblock {\em Cryptoeconomic Systems}, 2021.

\bibitem{Naor2020}
Oded Naor and Idit Keidar.
\newblock {Expected Linear Round Synchronization: The Missing Link for Linear
  Byzantine SMR}.
\newblock {\em 34th International Symposium on Distributed Computing (DISC)},
  179, 2020.

\bibitem{Pass2018a}
Rafael Pass and Elaine Shi.
\newblock {Thunderella: Blockchains with Optimistic Instant Confirmation}.
\newblock {\em Lecture Notes in Computer Science (including subseries Lecture
  Notes in Artificial Intelligence and Lecture Notes in Bioinformatics)}, 10821
  LNCS:3--33, 2018.

\bibitem{Rabin83}
Michael~O. Rabin.
\newblock {Randomized Byzantine Generals}.
\newblock In {\em 24th Annual Symposium on Foundations of Computer Science,
  Tucson, Arizona, USA, 7-9 November 1983}, pages 403--409. {IEEE} Computer
  Society, 1983.
\newblock \href {https://doi.org/10.1109/SFCS.1983.48}
  {\path{doi:10.1109/SFCS.1983.48}}.

\bibitem{Ramasamy2006}
Hari Govind~V. Ramasamy and Christian Cachin.
\newblock {Parsimonious Asynchronous Byzantine-Fault-Tolerant Atomic
  Broadcast}.
\newblock {\em Lecture Notes in Computer Science (including subseries Lecture
  Notes in Artificial Intelligence and Lecture Notes in Bioinformatics)}, 3974
  LNCS:88--102, 2006.

\bibitem{Spiegelman2021}
Alexander Spiegelman.
\newblock {In Search for an Optimal Authenticated Byzantine Agreement}.
\newblock In Seth Gilbert, editor, {\em 35th International Symposium on
  Distributed Computing (DISC 2021)}, volume 209 of {\em Leibniz International
  Proceedings in Informatics (LIPIcs)}, pages 38:1--38:19, Dagstuhl, Germany,
  2021. Schloss Dagstuhl -- Leibniz-Zentrum f{\"u}r Informatik.
\newblock URL: \url{https://drops.dagstuhl.de/opus/volltexte/2021/14840}, \href
  {https://doi.org/10.4230/LIPIcs.DISC.2021.38}
  {\path{doi:10.4230/LIPIcs.DISC.2021.38}}.

\bibitem{Srikanth1987}
T.~K. Srikanth and Sam Toueg.
\newblock {Optimal Clock Synchronization}.
\newblock {\em Journal of the Association for Computing Machinery},
  34(3):71--86, 1987.

\bibitem{Diem2021}
{The Diem Team}.
\newblock {DiemBFT v4: State Machine Replication in the Diem Blockchain}, 2021.
\newblock URL:
  \url{https://developers.diem.com/papers/diem-consensus-state-machine-replication-in-the-diem-blockchain/2021-08-17.pdf}.

\bibitem{Yin2019}
Maofan Yin, Dahlia Malkhi, Michael~K. Reiter, Guy~Golan Gueta, and Ittai
  Abraham.
\newblock {HotStuff: BFT Consensus with Linearity and Responsiveness}.
\newblock {\em Proceedings of the Annual ACM Symposium on Principles of
  Distributed Computing}, pages 347--356, 2019.

\end{thebibliography}

\newpage
\appendix
\section{\rare: Proof of Correctness and Complexity}
\label{section:appendix_rare}

This section proves the correctness and establishes the complexity of \rare (\Cref{algorithm:synchronizer}).
We start by defining the concept of a process' \emph{behavior} and \emph{timer history}.

\smallskip
\textit{Behaviors \& timer histories.}
A \emph{behavior} of a process $P_i$ is a sequence of (1) message-sending events performed by $P_i$, (2) message-reception events performed by $P_i$, and (3) internal events performed by $P_i$ (e.g., invocations of the $\mathsf{measure}(\cdot)$ and $\mathsf{cancel()}$ methods on the local timers).
If an event $e$ belongs to a behavior $\beta_i$, we write $e \in \beta_i$; otherwise, we write $e \notin \beta_i$.
If an event $e_1$ precedes an event $e_2$ in a behavior $\beta_i$, we write $e_1 \stackrel{\beta_i}{\prec} e_2$.
Note that, if $e_1 \stackrel{\beta_i}{\prec} e_2$ and $e_1$ occurs at some time $t_1$ and $e_2$ occurs at some time $t_2$, $t_1 \leq t_2$.

A \emph{timer history} of a process $P_i$ is a sequence of (1) invocations of the $\mathsf{measure}(\cdot)$ and $\mathsf{cancel}()$ methods on $\mathit{view\_timer}_i$ and $\mathit{dissemination\_timer}_i$, and (2) processed expiration events of $\mathit{view\_timer}_i$ and $\mathit{dissemination\_timer}_i$.
Observe that a timer history of a process is a subsequence of the behavior of the process.
We further denote by $h_i|_{\mathit{view}}$ a subsequence of $h_i$ associated with $\mathit{view\_timer}_i$, where $h_i$ is a timer history of a process $P_i$.
If an expiration event $\mathit{Exp}$ of a timer is associated with an invocation $\mathit{Inv}$ of the $\mathsf{measure(\cdot)}$ method on the timer, we say that $\mathit{Inv}$ \emph{produces} $\mathit{Exp}$.
Note that a single invocation of the $\mathsf{measure}(\cdot)$ method can produce at most one expiration event.

Given an execution, we denote by $\beta_i$ and $h_i$ the behavior and the timer history of the process $P_i$, respectively.

\smallskip
\noindent \textbf{Proof of correctness.}
In order to prove the correctness of \rare, we need to prove that \rare ensures the eventual synchronization property (see \Cref{subsection:view_synchronization_problem_definition}).
First, we show that the value of $\mathit{view}_i$ variable at a correct process $P_i$ is never smaller than $1$ or greater than $f + 1$.

\begin{lemma} \label{lemma:view_f+1}
Let $P_i$ be a correct process.
Then, $1 \leq \mathit{view}_i \leq f + 1$ throughout the entire execution.
\end{lemma}
\begin{proof}
First, $\mathit{view}_i \geq 1$ throughout the entire execution since (1) the initial value of $\mathit{view}_i$ is $1$ (line~\ref{line:init_view} of \Cref{algorithm:variable_constants}), and (2) the value of $\mathit{view}_i$ either increases (line~\ref{line:increment_view}) or is set to $1$ (line~\ref{line:reset_view}).

By contradiction, suppose that $\mathit{view}_i = F > f + 1 > 1$ at some time during the execution.
The update of $\mathit{view}_i$ to $F > f + 1$ must have been done at line~\ref{line:increment_view}.
This means that, just before executing line~\ref{line:increment_view}, $\mathit{view}_i \geq f + 1$.
However, this contradicts the check at line~\ref{line:check_last_view}, which concludes the proof.
\end{proof}

The next lemma shows that, if an invocation of the $\mathsf{measure}(\cdot)$ method on $\mathit{dissemination\_timer}_i$ produces an expiration event, the expiration event immediately follows the invocation in the timer history $h_i$ of a correct process $P_i$.

\begin{lemma} \label{lemma:expiration_after_invocation_dissemination}
Let $P_i$ be a correct process.
Let $\mathit{Exp}_d$ be any expiration event of $\mathit{dissemination\_timer}_i$ that belongs to $h_i$ and let $\mathit{Inv}_d$ be the invocation of the $\mathsf{measure}(\cdot)$ method (on $\mathit{dissemination\_timer}_i$) that has produced $\mathit{Exp}_d$.
Then, $\mathit{Exp}_d$ immediately follows $\mathit{Inv}_d$ in $h_i$.
\end{lemma}
\begin{proof}
In order to prove the lemma, we show that only $\mathit{Exp}_d$ can immediately follow $\mathit{Inv}_d$ in $h_i$.
We consider the following scenarios:
\begin{compactitem}
    \item Let an invocation $\mathit{Inv}'_d$ of the $\mathsf{measure}(\cdot)$ method on $\mathit{dissemination\_timer}_i$ immediately follow $\mathit{Inv}_d$ in $h_i$:
    $\mathit{Inv}'_d$ could only have been invoked either at line~\ref{line:sync_timer_measure} or at line~\ref{line:sync_timer_measure_2}.
    However, an invocation of the $\mathsf{cancel()}$ method on $\mathit{dissemination\_timer}_i$ (line~\ref{line:sync_timer_cancel_1} or line~\ref{line:sync_timer_cancel_2}) must immediately precede $\mathit{Inv}'_d$ in $h_i$, which contradicts the fact that $\mathit{Inv}_d$ immediately precedes $\mathit{Inv}'_d$.
    Therefore, this scenario is impossible.
    
    \item Let an invocation $\mathit{Inv}'_d$ of the $\mathsf{cancel}()$ method on $\mathit{dissemination\_timer}_i$ immediately follow $\mathit{Inv}_d$ in $h_i$:
    $\mathit{Inv}'_d$ could only have been invoked either at line~\ref{line:sync_timer_cancel_1} or at line~\ref{line:sync_timer_cancel_2}.
    However, an invocation of the $\mathsf{cancel()}$ method on $\mathit{view\_timer}_i$ (line~\ref{line:cancel_view_timer_1} or line~\ref{line:cancel_view_timer_2}) must immediately precede $\mathit{Inv}'_d$ in $h_i$, which contradicts the fact that $\mathit{Inv}_d$ immediately precedes $\mathit{Inv}'_d$.
    Hence, this scenario is impossible, as well.
    
    \item Let an expiration event $\mathit{Exp}'_d \neq \mathit{Exp}_d$ of $\mathit{dissemination\_timer}_i$ immediately follow $\mathit{Inv}_d$ in $h_i$:
    As $\mathit{Inv}_d$ could have been invoked either at line~\ref{line:sync_timer_measure} or at line~\ref{line:sync_timer_measure_2}, an invocation of the $\mathsf{cancel()}$ method on $\mathit{dissemination\_timer}_i$ (line~\ref{line:sync_timer_cancel_1} or line~\ref{line:sync_timer_cancel_2}) immediately precedes $\mathit{Inv}_d$ in $h_i$.
    This contradicts the fact that $\mathit{Exp}'_d \neq \mathit{Exp}_d$ is produced and immediately follows $\mathit{Inv}_d$, which renders this scenario impossible.
    
    \item Let an invocation $\mathit{Inv}_v$ of the $\mathsf{measure}(\cdot)$ method on $\mathit{view\_timer}_i$ immediately follow $\mathit{Inv}_d$ in $h_i$:
    $\mathit{Inv}_v$ could have been invoked either at line~\ref{line:view_timer_measure_without_msg_exchange} or at line~\ref{line:view_timer_measure}.
    We further consider both cases:
    \begin{compactitem}
        \item If $\mathit{Inv}_v$ was invoked at line~\ref{line:view_timer_measure_without_msg_exchange}, then $\mathit{Inv}_v$ is immediately preceded by an expiration event of $\mathit{view\_timer}_i$ (line~\ref{line:rule_view_expired}).
        This case is impossible as $\mathit{Inv}_v$ is not immediately preceded by $\mathit{Inv}_d$.
        
        \item If $\mathit{Inv}_v$ was invoked at line~\ref{line:view_timer_measure}, then $\mathit{Inv}_v$ is immediately preceded by an expiration event of $\mathit{dissemination\_timer}_i$ (line~\ref{line:sync_timer_expires}).
        This case is also impossible as $\mathit{Inv}_v$ is not immediately preceded by $\mathit{Inv}_d$.
    \end{compactitem}
    As neither of the two cases is possible, $\mathit{Inv}_v$ cannot immediately follow $\mathit{Inv}_d$.
    
    \item Let an invocation $\mathit{Inv}_v$ of the $\mathsf{cancel}()$ method on $\mathit{view\_timer}_i$ immediately follow $\mathit{Inv}_d$ in $h_i$:
    $\mathit{Inv}_v$ could have been invoked either at line~\ref{line:cancel_view_timer_1} or at line~\ref{line:cancel_view_timer_2}.
    In both cases, an invocation of the $\mathsf{cancel}()$ method on $\mathit{dissemination\_timer}$ (line~\ref{line:sync_timer_cancel_1} or line~\ref{line:sync_timer_cancel_2}) immediately follows $\mathit{Inv}_v$ in $h_i$.
    This contradicts the fact that $\mathit{Inv}_d$ produces $\mathit{Exp}_d$, which implies that this case is impossible.
    
    \item Let an expiration event $\mathit{Exp}_v$ of $\mathit{view\_timer}_i$ immediately follow $\mathit{Inv}_d$ in $h_i$:
    As $\mathit{Inv}_d$ could have been invoked either at line~\ref{line:sync_timer_measure} or at line~\ref{line:sync_timer_measure_2}, invocations of the $\mathsf{cancel()}$ method on $\mathit{view\_timer}_i$ and $\mathit{dissemination\_timer}_i$ (lines~\ref{line:cancel_view_timer_1},~\ref{line:sync_timer_cancel_1} or lines~\ref{line:cancel_view_timer_2},~\ref{line:sync_timer_cancel_2}) immediately precede $\mathit{Inv}_d$ in $h_i$.
    This contradicts the fact that $\mathit{Exp}_v$ is produced and immediately follows $\mathit{Inv}_d$, which renders this scenario impossible.
\end{compactitem}
As any other option is impossible, $\mathit{Exp}_d$ must immediately follow $\mathit{Inv}_d$ in $h_i$.
Thus, the lemma.
\end{proof}

The next lemma shows that views entered by a correct process are monotonically increasing.

\begin{lemma} [Monotonically increasing views] \label{lemma:increasing_views}
Let $P_i$ be a correct process.
Let $e_1 = \mathsf{advance}(v)$, $e_2 = \mathsf{advance}(v')$ and $e_1 \stackrel{\beta_i}{\prec} e_2$.
Then, $v' > v$.
\end{lemma}
\begin{proof}
Let $\mathit{epoch}_i = e$ and $\mathit{view}_i = j$ when $P_i$ triggers $\mathsf{advance}(v)$.
Moreover, let $\mathit{epoch}_i = e'$ and $\mathit{view}_i = j'$ when $P_i$ triggers $\mathsf{advance}(v')$.
As the value of the $\mathit{epoch}_i$ variable only increases throughout the execution (lines~\ref{line:receive_epoch_over}, ~\ref{line:update_epoch_1} and lines~\ref{line:receive_epoch_over_complete}, ~\ref{line:update_epoch_2}), $e' \geq e$.

We investigate both possibilities:
\begin{compactitem}
    \item Let $e' > e$.
    In this case, the lemma follows from \Cref{lemma:view_f+1} and the fact that $(e' - 1) \cdot (f + 1) + j' > (e - 1) \cdot (f + 1) + j$, for every $j, j' \in [1, f + 1]$.
    
    \item Let $e' = e$.
    Just before triggering $\mathsf{advance}(v)$ (line~\ref{line:start_view_1} or line~\ref{line:start_view_without_msg_exchange} or line~\ref{line:start_view_2}), $P_i$ has invoked the $\mathsf{measure}(\cdot)$ method on $\mathit{view\_timer}_i$ (line~\ref{line:view_timer_measure_first_view} or line~\ref{line:view_timer_measure_without_msg_exchange} or line~\ref{line:view_timer_measure}); we denote this invocation of the $\mathsf{measure}(\cdot)$ method by $\mathit{Inv}_v$.
    
    Now, we investigate two possible scenarios:
    \begin{compactitem}
        \item Let $P_i$ trigger $\mathsf{advance}(v')$ at line~\ref{line:start_view_without_msg_exchange}.
        By contradiction, suppose that $j' \leq j$.
        Hence, just before triggering $\mathsf{advance}(v')$ (i.e., just before executing line~\ref{line:increment_view}), we have that $\mathit{view}_i < j$.
        Thus, line~\ref{line:reset_view} must have been executed by $P_i$ after triggering $\mathsf{advance}(v)$ and before triggering $\mathsf{advance}(v')$, which means that an expiration event of $\mathit{dissemination\_timer}_i$ (line~\ref{line:sync_timer_expires}) follows $\mathit{Inv}_v$ in $h_i$.
        By \Cref{lemma:expiration_after_invocation_dissemination}, the $\mathsf{measure}(\cdot)$ method on $\mathit{dissemination\_timer}_i$ was invoked by $P_i$ after the invocation of $\mathit{Inv}_v$.
        Hence, when the aforementioned invocation of the $\mathsf{measure}(\cdot)$ method on $\mathit{dissemination\_timer}_i$ was invoked by $P_i$ (line~\ref{line:sync_timer_measure} or line~\ref{line:sync_timer_measure_2}), the $\mathit{epoch}_i$ variable had a value greater than $e$ (line~\ref{line:update_epoch_1} or line~\ref{line:update_epoch_2}) since $\mathit{epoch}_i \geq e$ when processing line~\ref{line:receive_epoch_over} or line~\ref{line:receive_epoch_over_complete}; recall that the value of the $\mathit{epoch}_i$ variable only increases throughout the execution.
        Therefore, we reach a contradiction with the fact that $e' = e$, which means that $j' > j$ and the lemma holds in this case.
        
        \item Let $P_i$ trigger $\mathsf{advance}(v')$ at line~\ref{line:start_view_2}.
        In this case, $P_i$ processes an expiration event of $\mathit{dissemination\_timer}_i$ (line~\ref{line:sync_timer_expires}); therefore, the $\mathsf{measure}(\cdot)$ method on $\mathit{dissemination\_timer}_i$ was invoked by $P_i$ after the invocation of $\mathit{Inv}_v$ (by \Cref{lemma:expiration_after_invocation_dissemination}).
        As in the previous case, when the aforementioned invocation of the $\mathsf{measure}(\cdot)$ method on $\mathit{dissemination\_timer}_i$ was invoked by $P_i$ (line~\ref{line:sync_timer_measure} or line~\ref{line:sync_timer_measure_2}), the $\mathit{epoch}_i$ variable had a value greater than $e$ (line~\ref{line:update_epoch_1} or line~\ref{line:update_epoch_2}); recall that the value of the $\mathit{epoch}_i$ variable only increases throughout the execution.
        Thus, we reach a contradiction with the fact that $e' = e$, which renders this case impossible.
    \end{compactitem}
    In the only possible scenario, we have that $j' > j$, which implies that $v' > v$.
\end{compactitem}
The lemma holds as it holds in both possible cases.
\end{proof}

The next lemma shows that an invocation of the $\mathsf{measure(\cdot)}$ method cannot be immediately followed by another invocation of the same method in a timer history (of a correct process) associated with $\mathit{view\_timer}_i$.

\begin{lemma} \label{lemma:view_timer_behavior}
Let $P_i$ be a correct process.
Let $\mathit{Inv}_v$ be any invocation of the $\mathsf{measure(\cdot)}$ method on $\mathit{view\_timer}_i$ that belongs to $h_i$.
Invocation $\mathit{Inv}_v$ is not immediately followed by another invocation of the $\mathsf{measure(\cdot)}$ method on $\mathit{view\_timer}_i$ in $h_i|_{\mathit{view}}$.
\end{lemma}
\begin{proof}
We denote by $\mathit{Inv}'_v$ the first invocation of the $\mathsf{measure(\cdot)}$ method on $\mathit{view\_timer}_i$ after $\mathit{Inv}_v$ in $h_i|_{\mathit{view}}$.
If $\mathit{Inv}'_v$ does not exist, the lemma trivially holds.
Hence, let $\mathit{Inv}'_v$ exist in the rest of the proof.
We examine two possible cases:
\begin{compactitem}
    \item Let $\mathit{Inv}'_v$ be invoked at line~\ref{line:view_timer_measure_without_msg_exchange}: 
    In this case, there exists an expiration event of $\mathit{view\_timer}_i$ (line~\ref{line:rule_view_expired}) separating $\mathit{Inv}_v$ and $\mathit{Inv}'_v$ in $h_i|_{\mathit{view}}$.
    
    \item Let $\mathit{Inv}'_v$ be invoked at line~\ref{line:view_timer_measure}:
    In this case, $\mathit{Inv}'_v$ is immediately preceded by an expiration event $\mathit{Exp}_d$ of $\mathit{dissemination\_timer}_i$ (line~\ref{line:sync_timer_expires}) in $h_i$.
    By \Cref{lemma:expiration_after_invocation_dissemination}, an invocation $\mathit{Inv}_d$ of the $\mathsf{measure}(\cdot)$ method on $\mathit{dissemination\_timer}_i$ immediately precedes $\mathit{Exp}_d$ in $h_i$.
    As $\mathit{Inv}_d$ could have been invoked either at line~\ref{line:sync_timer_measure} or at line~\ref{line:sync_timer_measure_2}, $\mathit{Inv}_d$ is immediately preceded by invocations of the $\mathsf{cancel()}$ methods on $\mathit{view\_timer}_i$ and $\mathit{dissemination\_timer}_i$ (lines~\ref{line:cancel_view_timer_1},~\ref{line:sync_timer_cancel_1} or lines~\ref{line:cancel_view_timer_2},~\ref{line:sync_timer_cancel_2}).
    Hence, in this case, an invocation of the $\mathsf{cancel()}$ method on $\mathit{view\_timer}_i$ separates $\mathit{Inv}_v$ and $\mathit{Inv}'_v$ in $h_i|_{\mathit{view}}$.
\end{compactitem}
The lemma holds since  $\mathit{Inv}'_v$ does not immediately follow $\mathit{Inv}_v$ in $h_i|_{\mathit{view}}$ in any of the two cases.
\end{proof}

A direct consequence of \Cref{lemma:view_timer_behavior} is that an expiration event of $\mathit{view\_timer}_i$ immediately follows (in a timer history associated with $\mathit{view\_timer}_i$) the $\mathsf{measure(\cdot)}$ invocation that has produced it.

\begin{lemma} \label{lemma:expiration_after_invocation}
Let $P_i$ be a correct process.
Let $\mathit{Exp}_v$ be any expiration event that belongs to $h_i|_{\mathit{view}}$ and let $\mathit{Inv}_v$ be the invocation of the $\mathsf{measure(\cdot)}$ method (on $\mathit{view\_timer}_i$) that has produced $\mathit{Exp}_v$.
Then, $\mathit{Exp}_v$ immediately follows $\mathit{Inv}_v$ in $h_i|_{\mathit{view}}$.
\end{lemma}
\begin{proof}
We prove the lemma by induction.

\smallskip
\noindent \textbf{Base step:} \emph{Let $\mathit{Inv}_v^1$ be the first invocation of the $\mathsf{measure(\cdot)}$ method in $h_i|_{\mathit{view}}$ that produces an expiration event, and let $\mathit{Exp}_v^1$ be the expiration event produced by $\mathit{Inv}_v^1$. Expiration event $\mathit{Exp}_v^1$ immediately follows $\mathit{Inv}_v^1$ in $h_i|_{\mathit{view}}$.}
\smallskip
\\ Since $\mathit{Inv}_v^1$ produces the expiration event $\mathit{Exp}_v^1$, an invocation of the $\mathsf{cancel()}$ method does not immediately follow $\mathit{Inv}_v^1$ in $h_i|_{\mathit{view}}$.
Moreover, no invocation of the $\mathsf{measure(\cdot)}$ method immediately follows $\mathit{Inv}_v^1$ in $h_i|_{\mathit{view}}$ (by \Cref{lemma:view_timer_behavior}).
Finally, no expiration event produced by a different invocation of the $\mathsf{measure(\cdot)}$ method immediately follows $\mathit{Inv}_v^1$ in $h_i|_{\mathit{view}}$ since $\mathit{Inv}_v^1$ is the first invocation of the method in $h_i|_{\mathit{view}}$ that produces an expiration event.
Therefore, the statement of the lemma holds for $\mathit{Inv}_v^1$ and $\mathit{Exp}_v^1$.

\smallskip
\noindent \textbf{Induction step:} \emph{Let $\mathit{Inv}_v^j$ be the $j$-th invocation of the $\mathsf{measure(\cdot)}$ method in $h_i|_{\mathit{view}}$ that produces an expiration event, where $j > 1$, and let $\mathit{Exp}_v^j$ be the expiration event produced by $\mathit{Inv}_v^j$. Expiration event $\mathit{Exp}_v^j$ immediately follows $\mathit{Inv}_v^j$ in $h_i|_{\mathit{view}}$.
\\Induction hypothesis: For every $k \in [1, j - 1]$, the $k$-th invocation of the $\mathsf{measure}(\cdot)$ method in $h_i|_{\mathit{view}}$ that produces an expiration event is immediately followed by the produced expiration event in $h_i|_{\mathit{view}}$.}
\smallskip
\\ An invocation of the $\mathsf{cancel()}$ method does not immediately follow $\mathit{Inv}_v^j$ in $h_i|_{\mathit{view}}$ since $\mathit{Inv}_v^j$ produces $\mathit{Exp}_v^j$.
Moreover, no invocation of the $\mathsf{measure(\cdot)}$ method immediately follows $\mathit{Inv}_v^j$ in $h_i|_{\mathit{view}}$ (by \Cref{lemma:view_timer_behavior}).
Lastly, no expiration event produced by a different invocation of the $\mathsf{measure(\cdot)}$ method immediately follows $\mathit{Inv}_v^j$ in $h_i|_{\mathit{view}}$ by the induction hypothesis.
Therefore, the statement of the lemma holds for $\mathit{Inv}_v^j$ and $\mathit{Exp}_v^j$, which concludes the proof.
\end{proof}

We now prove the statement of \Cref{lemma:expiration_after_invocation_dissemination} for $\mathit{view\_timer}_i$.

\begin{lemma} \label{lemma:expiration_after_invocation_view}
Let $P_i$ be a correct process.
Let $\mathit{Exp}_v$ be any expiration event of $\mathit{view\_timer}_i$ that belongs to $h_i$ and let $\mathit{Inv}_v$ be the invocation of the $\mathsf{measure}(\cdot)$ method (on $\mathit{view\_timer}_i$) that has produced $\mathit{Exp}_v$.
Then, $\mathit{Exp}_v$ immediately follows $\mathit{Inv}_v$ in $h_i$.
\end{lemma}
\begin{proof}
Let us consider all possible scenarios (as in the proof of \Cref{lemma:expiration_after_invocation_dissemination}):
\begin{compactitem}
    \item Let an invocation $\mathit{Inv}_d$ of the $\mathsf{measure}(\cdot)$ method on $\mathit{dissemination\_timer}_i$ immediately follow $\mathit{Inv}_v$ in $h_i$:
    $\mathit{Inv}_d$ could have been invoked either at line~\ref{line:sync_timer_measure} or at line~\ref{line:sync_timer_measure_2}.
    However, an invocation of the $\mathsf{cancel()}$ method on $\mathit{dissemination\_timer}_i$ (line~\ref{line:sync_timer_cancel_1} or line~\ref{line:sync_timer_cancel_2}) must immediately precede $\mathit{Inv}_d$ in $h_i$, which contradicts the fact that $\mathit{Inv}_v$ immediately precedes $\mathit{Inv}_d$.
    Therefore, this scenario is impossible.
    
    \item Let an invocation $\mathit{Inv}_d$ of the $\mathsf{cancel}()$ method on $\mathit{dissemination\_timer}_i$ immediately follow $\mathit{Inv}_v$ in $h_i$:
    $\mathit{Inv}_d$ could have been invoked either at line~\ref{line:sync_timer_cancel_1} or at line~\ref{line:sync_timer_cancel_2}.
    However, an invocation of the $\mathsf{cancel()}$ method on $\mathit{view\_timer}_i$ (line~\ref{line:cancel_view_timer_1} or line~\ref{line:cancel_view_timer_2}) must immediately precede $\mathit{Inv}_d$ in $h_i$, which contradicts the fact that $\mathit{Inv}_v$ immediately precedes $\mathit{Inv}_d$.
    Hence, this scenario is impossible, as well.
    
    \item Let an expiration event $\mathit{Exp}_d$ of $\mathit{dissemination\_timer}_i$ immediately follow $\mathit{Inv}_v$ in $h_i$:
    This is impossible due to \Cref{lemma:expiration_after_invocation_dissemination}.
    
    \item Let the event immediately following $\mathit{Inv}_v$ be (1) an invocation of the $\mathsf{measure(\cdot)}$ method on $\mathit{view\_timer}_i$, or (2) an invocation of the $\mathsf{cancel()}$ method on $\mathit{view\_timer}_i$, or (3) an expiration event $\mathit{Exp}'_v$ of $\mathit{view\_timer}_i$, where $\mathit{Exp}'_v \neq \mathit{Exp}_v$:
    This case is impossible due to \Cref{lemma:expiration_after_invocation}.
\end{compactitem}
As any other option is impossible, $\mathit{Exp}_v$ must immediately follow $\mathit{Inv}_v$ in $h_i$.
\end{proof}

Next, we show that the values of the $\mathit{epoch}_i$ and $\mathit{view}_i$ variables of a correct process $P_i$ do not change between an invocation of the $\mathsf{measure(\cdot)}$ method on $\mathit{view\_timer}_i$ and the processing of the expiration event the invocation produces.

\begin{lemma} \label{lemma:epoch_and_view_no_change}
Let $P_i$ be a correct process.
Let $\mathit{Inv}_v$ denote an invocation of the $\mathsf{measure(\cdot)}$ method on $\mathit{view\_timer}_i$ which produces an expiration event, and let $\mathit{Exp}_v$ denote the expiration event produced by $\mathit{Inv}_v$.
Let $\mathit{epoch}_i = e$ and $\mathit{view}_i = v$ when $P_i$ invokes $\mathit{Inv}_v$.
Then, when $P_i$ processes $\mathit{Exp}_v$ (line~\ref{line:rule_view_expired}), $\mathit{epoch}_i = e$ and $\mathit{view}_i = v$.
\end{lemma}
\begin{proof}
By contradiction, suppose that $\mathit{epoch}_i \neq e$ or $\mathit{view}_i \neq v$ when $P_i$ processes $\mathit{Exp}_v$.
Hence, the value of the variables of $P_i$ must have changed between invoking $\mathit{Inv}_v$ and processing $\mathit{Exp}_v$.
Let us investigate all possible lines of \Cref{algorithm:synchronizer} where $P_i$ could have modified its variables for the first time after invoking $\mathit{Inv}_v$ (the first modification occurs before processing $\mathit{Exp}_v$):
\begin{compactitem}
    \item the $\mathit{view}_i$ variable at line~\ref{line:increment_view}: 
    If $P_i$ has modified its $\mathit{view}_i$ variable here, there exists an expiration event of $\mathit{view\_timer}_i$ (line~\ref{line:rule_view_expired}) which follows $\mathit{Inv}_v$ in $h_i$.
    By \Cref{lemma:expiration_after_invocation_view}, this expiration event cannot occur before processing $\mathit{Exp}_v$, which implies that this case is impossible.
    
    \item the $\mathit{epoch}_i$ variable at line~\ref{line:update_epoch_1}:
    If $P_i$ updates its $\mathit{epoch}_i$ variable here, an invocation of the $\mathsf{cancel()}$ method on $\mathit{view\_timer}_i$ (line~\ref{line:cancel_view_timer_1}) separates $\mathit{Inv}_v$ and $\mathit{Exp}_v$ in $h_i$.
    However, this is impossible due to \Cref{lemma:expiration_after_invocation_view}, which renders this case impossible.
    
    \item the $\mathit{epoch}_i$ variable at line~\ref{line:update_epoch_2}:
    If $P_i$ updates its $\mathit{epoch}_i$ variable here, an invocation of the $\mathsf{cancel()}$ method on $\mathit{view\_timer}_i$ (line~\ref{line:cancel_view_timer_2}) separates $\mathit{Inv}_v$ and $\mathit{Exp}_v$ in $h_i$.
    However, this is impossible due to \Cref{lemma:expiration_after_invocation_view}, which implies that this case is impossible.
    
    \item the $\mathit{view}_i$ variable at line~\ref{line:reset_view}:
    If $P_i$ updates its $\mathit{view}_i$ variable here, an expiration event of $\mathit{dissemination\_timer}_i$ (line~\ref{line:sync_timer_expires}) separates $\mathit{Inv}_v$ and $\mathit{Exp}_v$ in $h_i$, which contradicts \Cref{lemma:expiration_after_invocation_view}. 
\end{compactitem}
Given that $P_i$ does not change the value of neither $\mathit{epoch}_i$ nor $\mathit{view}_i$ between invoking $\mathit{Inv}_v$ and processing $\mathit{Exp}_v$, the lemma holds.
\end{proof}

Finally, we show that correct processes cannot ``jump'' into an epoch, i.e., they must go into an epoch by going into its first view.

\begin{lemma} \label{lemma:no_jump}
Let $P_i$ be a correct process.
Let $\mathsf{advance}(v) \in \beta_i$, where $v$ is the $j$-th view of an epoch $e$ and $j > 1$.
Then, $\mathsf{advance}(v - 1) \stackrel{\beta_i}{\prec} \mathsf{advance}(v)$.
\end{lemma}
\begin{proof}
Since $P_i$ enters view $v$, which is not the first view of epoch $e$, $P_i$ triggers $\mathsf{advance}(v)$ at line~\ref{line:start_view_without_msg_exchange}: $P_i$ could not have triggered $\mathsf{advance}(v)$ neither at line~\ref{line:start_view_1} nor at line~\ref{line:start_view_2} since $v$ is not the first view of epoch $e$.
Due to line~\ref{line:rule_view_expired}, the $\mathsf{measure}(\cdot)$ method was invoked on $\mathit{view\_timer}_i$ before $\mathsf{advance}(v)$ is triggered; we denote by $\mathit{Inv}_v$ this specific invocation of the $\mathsf{measure(\cdot)}$ method on $\mathit{view\_timer}_i$ and by $\mathit{Exp}_v$ its expiration event (processed by $P_i$ just before triggering $\mathsf{advance}(v)$).


When $P_i$ triggers $\mathsf{advance}(v)$ (at line~\ref{line:start_view_without_msg_exchange}), we have that $\mathit{epoch}_i = e$ and $\mathit{view}_i = j$.
Moreover, when processing $\mathit{Exp}_v$, we have that $\mathit{epoch}_i = e$ and $\mathit{view}_i = j - 1$ (due to line~\ref{line:increment_view}).
By \Cref{lemma:epoch_and_view_no_change}, when $P_i$ has invoked $\mathit{Inv}_v$, we had the same state: $\mathit{epoch}_i = e$ and $\mathit{view}_i = j - 1$.
Process $P_i$ could have invoked $\mathit{Inv}_v$ either (1) at line~\ref{line:view_timer_measure_first_view}, or (2) at line~\ref{line:view_timer_measure_without_msg_exchange}, or (3) at line~\ref{line:view_timer_measure}. 
Since $P_i$ triggers $\mathsf{advance(\cdot)}$ immediately after (line~\ref{line:start_view_1}, line~\ref{line:start_view_without_msg_exchange}, or line~\ref{line:start_view_2}), that $\mathsf{advance}(\cdot)$ indication is for $v - 1$ (as $\mathit{epoch}_i = e$ and $\mathit{view}_i = j - 1$ at that time).
Hence, $\mathsf{advance}(v - 1) \stackrel{\beta_i}{\prec} \mathsf{advance}(v)$.
\end{proof}

We say that a correct process \emph{enters} an epoch $e$ at time $t$ if and only if the process enters the first view of $e$ (i.e., the view $(e - 1) \cdot (f + 1) + 1$) at time $t$.
Furthermore, a correct process \emph{is in epoch $e$} between the time $t$ (including $t$) at which it enters $e$ and the time $t'$ (excluding $t'$) at which it enters (for the first time after entering $e$) another epoch $e'$.
If another epoch is never entered, the process is in epoch $e$ from time $t$ onward.
Recall that, by \Cref{lemma:increasing_views}, a correct process enters each view at most once, which means that a correct process enters each epoch at most once.

The following lemma shows that, if a correct process broadcasts an \textsc{epoch-completed} message for an epoch (line~\ref{line:broadcast_epoch_over}), then the process has previously entered that epoch.

\begin{lemma} \label{lemma:epoch_over_previously_started}
Let a correct process $P_i$ send an \textsc{epoch-completed} message for an epoch $e$ (line~\ref{line:broadcast_epoch_over}); let this sending event be denoted by $e_{\mathit{send}}$.
Then, $\mathsf{advance}(v) \stackrel{\beta_i}{\prec} e_{\mathit{send}}$, where $v$ is the first view of the epoch $e$.
\end{lemma}
\begin{proof}
At the moment of sending the message (line~\ref{line:broadcast_epoch_over}), the following holds: (1) $\mathit{epoch}_i = e$, and (2) $\mathit{view}_i = f + 1$ (by the check at line~\ref{line:check_last_view} and \Cref{lemma:view_f+1}).
We denote by $\mathit{Inv}_v$ the invocation of the $\mathsf{measure(\cdot)}$ method on $\mathit{view\_timer}_i$ producing the expiration event $\mathit{Exp}_v$ leading to $P_i$ broadcasting the \textsc{epoch-completed} message for $e$.
Note that $\mathit{Inv}_v$ precedes the sending of the \textsc{epoch-completed} message in $\beta_i$.

When processing $\mathit{Exp}_v$ (line~\ref{line:rule_view_expired}), the following was the state of $P_i$: $\mathit{epoch}_i = e$ and $\mathit{view}_i = f + 1$.
By \Cref{lemma:epoch_and_view_no_change}, when $P_i$ invokes $\mathit{Inv}_v$, $\mathit{epoch}_i = e$ and $\mathit{view}_i = f + 1 > 1$.
Therefore, $\mathit{Inv}_v$ must have been invoked at line~\ref{line:view_timer_measure_without_msg_exchange}: $\mathit{Inv}_v$ could not have invoked neither at line~\ref{line:view_timer_measure_first_view} nor at line~\ref{line:view_timer_measure} since $\mathit{view}_i = f + 1 \neq 1$ at that moment.
Immediately after invoking $\mathit{Inv}_v$, $P_i$ enters the $(f + 1)$-st view of $e$ (line~\ref{line:start_view_without_msg_exchange}), which implies that $P_i$ enters the $(f + 1)$-st view of $e$ before it sends the \textsc{epoch-completed} message.
Therefore, the lemma follows from \Cref{lemma:no_jump}.
\end{proof}

The next lemma shows that, if a correct process $P_i$ updates its $\mathit{epoch}_i$ variable to $e > 1$, then (at least) $f + 1$ correct processes have previously entered epoch $e - 1$.

\begin{lemma} \label{lemma:epoch_update_previous_epoch_entered}
Let a correct process $P_i$ update its $\mathit{epoch}_i$ variable to $e > 1$ at some time $t$.
Then, at least $f + 1$ correct processes have entered $e - 1$ by time $t$.
\end{lemma}
\begin{proof}
Since $P_i$ updates $\mathit{epoch}_i$ to $e > 1$ at time $t$, it does so at either:
\begin{compactitem}
    \item line~\ref{line:update_epoch_1}: In this case, $P_i$ has received $2f + 1$ \textsc{epoch-completed} messages for epoch $e - 1$ (line~\ref{line:receive_epoch_over}), out of which (at least) $f + 1$ were sent by correct processes.
    
    \item line~\ref{line:update_epoch_2}: In this case, $P_i$ has received a threshold signature of epoch $e - 1$ (line~\ref{line:receive_epoch_over_complete}) built out of $2f + 1$ partial signatures, out of which (at least) $f + 1$ must have come from correct processes.
    Such a partial signature from a correct process can only be obtained by receiving an \textsc{epoch-completed} message for epoch $e - 1$ from that process.
\end{compactitem}
In both cases, $f + 1$ correct processes have sent \textsc{epoch-completed} messages (line~\ref{line:broadcast_epoch_over}) for epoch $e - 1$ by time $t$.
By \Cref{lemma:epoch_over_previously_started}, all these correct processes have entered epoch $e - 1$ by time $t$.
\end{proof}

Note that a correct process $P_i$ \emph{does not} enter an epoch immediately upon updating its $\mathit{epoch}_i$ variable, but only upon triggering the $\mathsf{advance(\cdot)}$ indication for the first view of that epoch (line~\ref{line:start_view_1} or line~\ref{line:start_view_2}).
We now prove that, if an epoch $e > 1$ is entered by a correct process at some time $t$, then epoch $e - 1$ is entered by a (potentially different) correct process by time $t$. 

\begin{lemma} \label{lemma:epoch_previously_started}
Let a correct process $P_i$ enter an epoch $e > 1$ at time $t$.
Then, epoch $e - 1$ was entered by a correct process by time $t$. 
\end{lemma}
\begin{proof}
Since $P_i$ enters $e > 1$ at time $t$ (line~\ref{line:start_view_2}), $\mathit{epoch}_i = e$ at time $t$.
Hence, $P_i$ has updated its $\mathit{epoch}_i$ variable to $e > 1$ by time $t$.
Therefore, the lemma follows directly from \Cref{lemma:epoch_update_previous_epoch_entered}.
\end{proof}

The next lemma shows that all epochs are eventually entered by some correct processes. 
In other words, correct processes keep transiting to new epochs forever.

\begin{lemma} \label{lemma:no_decision_constantly_next_epoch}
Every epoch is eventually entered by a correct process. 
\end{lemma}

\begin{proof}
Epoch $1$ is entered by a correct process since every correct process initially triggers the $\mathsf{advance(}1\mathsf{)}$ indication (line~\ref{line:start_view_1}).
Therefore, it is left to prove that all epochs greater than $1$ are entered by a correct process.
By contradiction, let $e + 1$ be the smallest epoch not entered by a correct process, where $e \geq 1$.

\smallskip
\noindent \textbf{Part 1.} \emph{No correct process $P_i$ ever sets $\mathit{epoch}_i$ to an epoch greater than $e$.}
\smallskip
\\Since $e + 1$ is the smallest epoch not entered by a correct process, no correct process ever enters any epoch greater than $e$ (by \Cref{lemma:epoch_previously_started}).
Furthermore, \Cref{lemma:epoch_update_previous_epoch_entered} shows that no correct process $P_i$ ever updates its $\mathit{epoch}_i$ variable to an epoch greater than $e + 1$. 

Finally, $P_i$ never sets $\mathit{epoch}_i$ to $e + 1$ either.
By contradiction, suppose that it does.
In this case, $P_i$ invokes the $\mathsf{measure(}\delta\mathsf{)}$ method on $\mathit{dissemination\_timer}_i$ (either line~\ref{line:sync_timer_measure} or line~\ref{line:sync_timer_measure_2}).
Since $P_i$ does not update $\mathit{epoch}_i$ to an epoch greater than $e + 1$ (as shown in the previous paragraph), the previously invoked $\mathsf{measure(}\delta\mathsf{)}$ method will never be canceled (neither at line~\ref{line:sync_timer_cancel_1} nor at line~\ref{line:sync_timer_cancel_2}).
This implies that $\mathit{dissemination\_timer}_i$ eventually expires (line~\ref{line:sync_timer_expires}), and $P_i$ enters epoch $e + 1$ (line~\ref{line:start_view_2}).
Hence, a contradiction with the fact that epoch $e + 1$ is never entered by a correct process.


\smallskip
\noindent \textbf{Part 2.} \emph{Every correct process eventually enters epoch $e$.}
\smallskip
\\ If $e = 1$, every correct process enters $e$ as every correct process eventually executes line~\ref{line:start_view_1}.

Let $e > 1$.
Since $e > 1$ is entered by a correct process (line~\ref{line:start_view_2}), the process has disseminated an $\textsc{enter-epoch}$ message for $e$ (line~\ref{line:broadcast_epoch_over_complete}).
This message is eventually received by every correct process since the network is reliable.
If a correct process $P_i$ has not previously set its $\mathit{epoch}_i$ variable to $e$, it does so upon the reception of the \textsc{enter-epoch} message (line~\ref{line:update_epoch_2}).
Hence, $P_i$ eventually sets its $\mathit{epoch}_i$ variable to $e$.

Immediately after updating its $\mathit{epoch}_i$ variable to $e$ (line~\ref{line:update_epoch_1} or line~\ref{line:update_epoch_2}), $P_i$ invokes $\mathsf{measure}(\delta)$ on $\mathit{dissemination\_timer}_i$ (line~\ref{line:sync_timer_measure} or line~\ref{line:sync_timer_measure_2}).
Because $P_i$ never updates $\mathit{epoch}_i$ to an epoch greater than $e$ (by Part 1), $\mathit{dissemination\_timer}_i$ expires while $\mathit{epoch}_i = e$.
When this happens (line~\ref{line:sync_timer_expires}), $P_i$ enters epoch $e$ (line~\ref{line:start_view_2}).
Thus, all correct processes eventually enter epoch $e$.

\smallskip
\noindent \textbf{Epilogue.} By Part 2, a correct process $P_i$ eventually enters epoch $e$ (line~\ref{line:start_view_1} or line~\ref{line:start_view_2}); when $P_i$ enters $e$, $\mathit{epoch}_i = e$ and $\mathit{view}_i = 1$.
Moreover, just before entering $e$, $P_i$ invokes the $\mathsf{measure}(\cdot)$ method on $\mathit{view\_timer}_i$ (line~\ref{line:view_timer_measure_first_view} or line~\ref{line:view_timer_measure}); let this invocation be denoted by $\mathit{Inv}_v^1$.
As $P_i$ never updates its $\mathit{epoch}_i$ variable to an epoch greater than $e$ (by Part 1), $\mathit{Inv}_v^1$ eventually expires.
When $P_i$ processes the expiration of $\mathit{Inv}_v^1$ (line~\ref{line:rule_view_expired}), $\mathit{epoch}_i = e$ and $\mathit{view}_i = 1 < f + 1$ (by \Cref{lemma:epoch_and_view_no_change}).
Hence, $P_i$ then invokes the $\mathsf{measure}(\cdot)$ method on $\mathit{view\_timer}_i$ (line~\ref{line:view_timer_measure_without_msg_exchange}); when this occurs, $\mathit{epoch}_i = e$ and $\mathit{view}_i = 2$ (by line~\ref{line:increment_view}).
Following the same argument as for $\mathit{Inv}_v^1$, $\mathit{view\_timer}_i$ expires for each view of epoch $e$.

Therefore, every correct process $P_i$ eventually broadcasts an \textsc{epoch-completed} message for epoch $e$ (line~\ref{line:broadcast_epoch_over}) when $\mathit{view\_timer}_i$ expires for the last view of epoch $e$.
Thus, a correct process $P_j$ eventually receives $2f + 1$ \textsc{epoch-completed} messages for epoch $e$ (line~\ref{line:receive_epoch_over}), and updates $\mathit{epoch}_j$ to $e + 1$ (line~\ref{line:update_epoch_1}).
This contradicts Part 1, which implies that the lemma holds.
\end{proof}

We now introduce $e_{\mathit{final}}$, the first new epoch entered at or after $\mathit{GST}$.

\begin{definition} \label{definition:final}
We denote by $e_{\mathit{final}}$ the smallest epoch such that the first correct process to enter $e_{\mathit{final}}$ does so at time $t_{e_\mathit{final}} \geq \mathit{GST}$.
\end{definition}

Note that $e_{\mathit{final}}$ exists due to \Cref{lemma:no_decision_constantly_next_epoch}; recall that, by $\mathit{GST}$, an execution must be finite as no process is able to perform infinitely many steps in finite time.
It is stated in \Cref{algorithm:variable_constants} that $\mathit{view\_duration} = \Delta + 2\delta$ (line~\ref{line:view_duration}).
However, technically speaking, $\mathit{view\_duration}$ must be greater than $\Delta + 2\delta$ in order to not waste the ``very last'' moment of a $\Delta + 2\delta$ time period, i.e., we set $\mathit{view\_duration} = \Delta + 2\delta + \epsilon$, where $\epsilon$ is any positive constant.
Therefore, in the rest of the section, we assume that $\mathit{view\_duration} = \Delta + 2\delta + \epsilon > \Delta + 2\delta$.

We now show that, if a correct process enters an epoch $e$ at time $t_e \geq \mathit{GST}$ and sends an \textsc{epoch-completed} message for $e$, the \textsc{epoch-completed} message is sent at time $t_e + \mathit{epoch\_duration}$, where $\mathit{epoch\_duration} = (f + 1) \cdot \mathit{view\_duration}$.

\begin{lemma} \label{lemma:epoch_completed_time}
Let a correct process $P_i$ enter an epoch $e$ at time $t_e \geq \mathit{GST}$ and let $P_i$ send an \textsc{epoch-completed} message for epoch $e$ (line~\ref{line:broadcast_epoch_over}).
The \textsc{epoch-completed} message is sent at time $t_e + \mathit{epoch\_duration}$.
\end{lemma}
\begin{proof}
We prove the lemma by backwards induction.
Let $t^*$ denote the time at which the \textsc{epoch-completed} message for epoch $e$ is sent (line~\ref{line:broadcast_epoch_over}).

\smallskip
\noindent \textbf{Base step:} \emph{The $(f + 1)$-st view of the epoch $e$ is entered by $P_i$ at time $t^{f + 1}$ such that $t^* - t^{f + 1} = 1 \cdot \mathit{view\_duration}$.}
\smallskip
\\ When sending the \textsc{epoch-completed} message (line~\ref{line:broadcast_epoch_over}), the following holds: $\mathit{epoch}_i = e$ and $\mathit{view}_i = f + 1$ (due to the check at line~\ref{line:check_last_view} and \Cref{lemma:view_f+1}).
Let $\mathit{Exp}_v^{f + 1}$ denote the expiration event of $\mathit{view\_timer}_i$ processed just before broadcasting the message (line~\ref{line:rule_view_expired}).
When processing $\mathit{Exp}_v^{f + 1}$, we have that $\mathit{epoch}_i = e$ and $\mathit{view}_i = f + 1$.
When $P_i$ has invoked $\mathit{Inv}_v^{f + 1}$, where $\mathit{Inv}_v^{f + 1}$ is the invocation of the $\mathsf{measure}(\cdot)$ method which has produced $\mathit{Exp}_v^{f + 1}$, we have that $\mathit{epoch}_i = e$ and $\mathit{view}_i = f + 1$ (by \Cref{lemma:epoch_and_view_no_change}).
As $f + 1 \neq 1$, $\mathit{Inv}_v^{f + 1}$ is invoked at line~\ref{line:view_timer_measure_without_msg_exchange} at some time $t^{f + 1} \leq t^*$.
Finally, $P_i$ enters the $(f + 1)$-st view of the epoch $e$ at line~\ref{line:start_view_without_msg_exchange} at time $t^{f + 1}$.
By \Cref{lemma:no_jump}, we have that $t^{f + 1} \geq t_e \geq \mathit{GST}$.
As local clocks do not drift after $\mathit{GST}$, we have that $t^* - t^{f + 1} = \mathit{view\_duration}$ (due to line~\ref{line:view_timer_measure_without_msg_exchange}), which concludes the base step.

\smallskip
\noindent \textbf{Induction step:} \emph{Let $j \in [1, f]$. 
The $j$-th view of the epoch $e$ is entered by $P_i$ at time $t^j$ such that $t^* - t^j = (f + 2 - j)\cdot\mathit{view\_duration}$.
\\Induction hypothesis: For every $k \in [j + 1, f + 1]$, the $k$-th view of the epoch $e$ is entered by $P_i$ at time $t^k$ such that $t^* - t^k = (f + 2 - k) \cdot \mathit{view\_duration}$.}
\smallskip
\\ Let us consider the $(j + 1)$-st view of the epoch $e$; note that $j + 1 \neq 1$.
Hence, the $(j + 1)$-st view of the epoch $e$ is entered by $P_i$ at some time $t^{j + 1}$ at line~\ref{line:start_view_without_msg_exchange}, where $t^* - t^{j + 1} = (f + 2 - j - 1) \cdot \mathit{view\_duration} = (f + 1 - j) \cdot \mathit{view\_duration}$ (by the induction hypothesis).
Let $\mathit{Exp}_v^{j}$ denote the expiration event of $\mathit{view\_timer}_i$ processed at time $t^{j + 1}$ (line~\ref{line:rule_view_expired}).
When processing $\mathit{Exp}_v^{j}$, we have that $\mathit{epoch}_i = e$ and $\mathit{view}_i = j$ (due to line~\ref{line:increment_view}).
When $P_i$ has invoked $\mathit{Inv}_v^{j}$ at some time $t^{j}$, where $\mathit{Inv}_v^{j}$ is the invocation of the $\mathsf{measure}(\cdot)$ method which has produced $\mathit{Exp}_v^{j}$, we have that $\mathit{epoch}_i = e$ and $\mathit{view}_i = j$ (by \Cref{lemma:epoch_and_view_no_change}).
$\mathit{Inv}_v^{j}$ could have been invoked either at line~\ref{line:view_timer_measure_first_view}, or at line~\ref{line:view_timer_measure_without_msg_exchange}, or at line~\ref{line:view_timer_measure}:
\begin{compactitem}
    \item line~\ref{line:view_timer_measure_first_view}:
    In this case, $P_i$ enters the $j$-th view of the epoch $e$ at time $t^j$ at line~\ref{line:start_view_1}, where $j = 1$ (by line~\ref{line:start_view_1}).
    Moreover, we have that $t^{j} \geq \mathit{GST}$ as $t^j = t_e$ (by \Cref{lemma:increasing_views}).
    As local clocks do not drift after $\mathit{GST}$, we have that $t^{j + 1} - t^j = \mathit{view\_duration}$, which implies that $t^* - t^j = t^* - t^{j + 1} + \mathit{view\_duration} = (f + 1 - j + 1) \cdot \mathit{view\_duration} = (f + 2 - j) \cdot \mathit{view\_duration}$.
    Hence, in this case, the induction step is concluded.
    
    \item line~\ref{line:view_timer_measure_without_msg_exchange}:
    $P_i$ enters the $j$-th view of the epoch $e$ at line~\ref{line:start_view_without_msg_exchange} at time $t^j$, where $j > 1$ (by \Cref{lemma:view_f+1} and line~\ref{line:increment_view}).
    By lemmas~\ref{lemma:increasing_views} and~\ref{lemma:no_jump}, we have that $t^j \geq t_e \geq \mathit{GST}$.
    As local clocks do not drift after $\mathit{GST}$, we have that $t^{j + 1} - t^j = \mathit{view\_duration}$, which implies that $t^* - t^j = (f + 2 - j) \cdot \mathit{view\_duration}$.
    Hence, the induction step is concluded even in this case.
    
    \item line~\ref{line:view_timer_measure}:
    In this case, $P_i$ enters the $j$-th view of the epoch $e$ at time $t^j$ at line~\ref{line:start_view_2}, where $j = 1$ as $\mathit{view}_i = 1$ (by line~\ref{line:reset_view}).
    Moreover, $t^j = t_e \geq \mathit{GST}$ (by \Cref{lemma:increasing_views}).
    As local clocks do not drift after $\mathit{GST}$, we have that $t^{j + 1} - t^j = \mathit{view\_duration}$, which implies that $t^* - t^j = t^* - t^{j + 1} + \mathit{view\_duration} = (f + 1 - j + 1) \cdot \mathit{view\_duration} = (f + 2 - j) \cdot \mathit{view\_duration}$.
    Hence, even in this case, the induction step is concluded.
\end{compactitem}
As the induction step is concluded in all possible scenarios, the backwards induction holds.
Therefore, $P_i$ enters the first view of the epoch $e$ (and, thus, the epoch $e$) at time $t_e$ (recall that the first view of any epoch is entered at most once by \Cref{lemma:increasing_views}) such that $t^* - t_e = (f + 1) \cdot \mathit{view\_duration} = \mathit{epoch\_duration}$, which concludes the proof.
\end{proof}

The following lemma shows that no correct process broadcasts an \textsc{epoch-completed} message for an epoch $\geq e_{\mathit{final}}$ before time $t_{e_\mathit{final}} + \mathit{epoch\_duration}$.

\begin{lemma} \label{lemma:no_epoch_over_for_epoch_duration}
No correct process broadcasts an \textsc{epoch-completed} message for an epoch $e' \geq e_{\mathit{final}}$ (line~\ref{line:broadcast_epoch_over}) before time $t_{e_\mathit{final}} + \mathit{epoch\_duration}$.
\end{lemma}

\begin{proof}
Let $t^*$ be the first time a correct process, denoted by $P_i$, sends an \textsc{epoch-completed} message for an epoch $e' \geq e_{\mathit{final}}$ (line~\ref{line:broadcast_epoch_over}); if $t^*$ is not defined, the lemma trivially holds.
By \Cref{lemma:epoch_over_previously_started}, $P_i$ has entered epoch $e'$ at some time $t_{e'} \leq t^*$.
If $e' = e_{\mathit{final}}$, then $t_e' \geq t_{e_{\mathit{final}}} \geq \mathit{GST}$.
If $e' > e_{\mathit{final}}$, by \Cref{lemma:epoch_previously_started}, $t_{e'} \geq t_{e_{\mathit{final}}} \geq \mathit{GST}$.
Therefore, $t^* = t_{e'} + \mathit{epoch\_duration}$ (by \Cref{lemma:epoch_completed_time}), which means that $t^* \geq t_{e_{\mathit{final}}} + \mathit{epoch\_duration}$.
\end{proof}

Next, we show during which periods a correct process is in which view of the epoch $e_{\mathit{final}}$.

\begin{lemma} \label{lemma:time_in_final}
Consider a correct process $P_i$.
\begin{compactitem}
    \item For any $j \in [1, f]$, $P_i$ enters the $j$-th view of the epoch $e_{\mathit{final}}$ at some time $t^j$, where $t^j \in \big[t_{e_{\mathit{final}}} + (j - 1)\cdot\mathit{view\_duration}, t_{e_{\mathit{final}}} + (j - 1)\cdot\mathit{view\_duration} + 2\delta\big]$, and stays in the view until (at least) time $t^j + \mathit{view\_duration}$ (excluding time $t^j + \mathit{view\_duration}$).
    
    \item For $j = f + 1$, $P_i$ enters the $j$-th view of the epoch $e_{\mathit{final}}$ at some time $t^j$, where $t^j \in \big[t_{e_{\mathit{final}}} + f\cdot\mathit{view\_duration}, t_{e_{\mathit{final}}} + f\cdot\mathit{view\_duration} + 2\delta\big]$, and stays in the view until (at least) time $t_{e_{\mathit{final}}} + \mathit{epoch\_duration}$ (excluding time $t_{e_{\mathit{final}}} + \mathit{epoch\_duration}$).
\end{compactitem}
\end{lemma}
\begin{proof}
Note that no correct process broadcasts an \textsc{epoch-completed} message for an epoch $\geq e_{\mathit{final}}$ (line~\ref{line:broadcast_epoch_over}) before time $t_{e_{\mathit{final}}} + \mathit{epoch\_duration}$ (by \Cref{lemma:no_epoch_over_for_epoch_duration}).
We prove the lemma by induction.

\smallskip
\noindent \textbf{Base step:} \emph{The statement of the lemma holds for $j = 1$.}
\smallskip
\\ If $e_{\mathit{final}} > 1$, every correct process receives an \textsc{enter-epoch} message (line~\ref{line:receive_epoch_over_complete}) for epoch $e_{\mathit{final}}$ by time $t_{e_\mathit{final}} + \delta$ (since $t_{e_\mathit{final}} \geq \mathit{GST}$).
As no correct process broadcasts an \textsc{epoch-completed} message for an epoch $\geq e_{\mathit{final}}$ before time $t_{e_\mathit{final}} + \mathit{epoch\_duration} > t_{e_{\mathit{final}}} + \delta$, $P_i$ sets its $\mathit{epoch}_i$ variable to $e_{\mathit{final}}$ (line~\ref{line:update_epoch_2}) and invokes the $\mathsf{measure}(\delta)$ method on $\mathit{dissemination\_timer}_i$ (line~\ref{line:sync_timer_measure_2}) by time $t_{e_{\mathit{final}}} + \delta$.
Because of the same reason, the $\mathit{dissemination\_timer}_i$ expires by time $t_{e_{\mathit{final}}} + 2\delta$ (line~\ref{line:sync_timer_expires}); at this point in time, $\mathit{epoch}_i = e_{\mathit{final}}$.
Hence, $P_i$ enters the first view of $e_{\mathit{final}}$ by time $t_{e_{\mathit{final}}} + 2\delta$ (line~\ref{line:start_view_2}).
Observe that, if $e_{\mathit{final}} = 1$, $P_i$ enters $e_{\mathit{final}}$ at time $t_{e_{\mathit{final}}}$ (as every correct process starts executing \Cref{algorithm:synchronizer} at $\mathit{GST} = t_{e_{\mathit{final}}}$).
Thus, $t^1 \in [t_{e_{\mathit{final}}}, t_{e_{\mathit{final}}} + 2\delta]$.

Prior to entering the first view of $e_{\mathit{final}}$, $P_i$ invokes the $\mathsf{measure}(\mathit{view\_duration})$ method on $\mathit{view\_timer}_i$ (line~\ref{line:view_timer_measure_first_view} or line~\ref{line:view_timer_measure}); we denote this invocation by $\mathit{Inv}_v$.
By \Cref{lemma:no_epoch_over_for_epoch_duration}, $\mathit{Inv}_v$ cannot be canceled (line~\ref{line:cancel_view_timer_1} or line~\ref{line:cancel_view_timer_2}) as $t_{e_{\mathit{final}}} + \mathit{epoch\_duration} > t_{e_{\mathit{final}}} + 2\delta + \mathit{view\_duration}$.
Therefore, $\mathit{Inv}_v$ produces an expiration event $\mathit{Exp}_v$ which is processed by $P_i$ at time $t^1 + \mathit{view\_duration}$ (since $t^1 \geq \mathit{GST}$ and local clocks do not drift after $\mathit{GST}$).

Let us investigate the first time $P_i$ enters another view after entering the first view of $e_{\mathit{final}}$.
This could happen at the following places of \Cref{algorithm:synchronizer}:
\begin{compactitem}
    \item line~\ref{line:start_view_without_msg_exchange}:
    By \Cref{lemma:expiration_after_invocation_view}, we conclude that this occurs at time $t^* \geq t^1 + \mathit{view\_duration}$.
    Therefore, in this case, $P_i$ is in the first view of $e_{\mathit{final}}$ during the time period $[t^1, t^1 + \mathit{view\_duration})$.
    The base step is proven in this case.
    
    \item line~\ref{line:start_view_2}:
    By contradiction, suppose that this happens before time $t^1 + \mathit{view\_duration}$.
    Hence, the $\mathsf{measure}(\cdot)$ method was invoked on $\mathit{dissemination\_timer}_i$ (line~\ref{line:sync_timer_measure} or line~\ref{line:sync_timer_measure_2}) before time $t^1 + \mathit{view\_duration}$ and after the invocation of $\mathit{Inv}_v$ (by \Cref{lemma:expiration_after_invocation_dissemination}).
    Thus, $\mathit{Inv}_v$ is canceled (line~\ref{line:cancel_view_timer_1} or line~\ref{line:cancel_view_timer_2}), which is impossible (as previously proven).
    
    Hence, $P_i$ is in the first view of $e_{\mathit{final}}$ during (at least) the time period $[t^1, t^1 + \mathit{view\_duration})$, which implies that the base step is proven even in this case.
\end{compactitem}

\smallskip
\noindent \textbf{Induction step:} 
\emph{The statement of the lemma holds for $j$, where $1 < j \leq f + 1$.
\\ Induction hypothesis: The statement of the lemma holds for every $k \in [1, j - 1]$.}
\smallskip
\\ Consider the $(j - 1)$-st view of $e_{\mathit{final}}$ denoted by $v_{j - 1}$.
Recall that $t^{j - 1}$ denotes the time at which $P_i$ enters $v_{j - 1}$.
Just prior to entering $v_{j - 1}$ (line~\ref{line:start_view_1} or line~\ref{line:start_view_without_msg_exchange} or line~\ref{line:start_view_2}), $P_i$ has invoked the $\mathsf{measure}(\mathit{view\_duration})$ method on $\mathit{view\_timer}_i$ (line~\ref{line:view_timer_measure_first_view} or line~\ref{line:view_timer_measure_without_msg_exchange} or line~\ref{line:view_timer_measure}); let this invocation be denoted by $\mathit{Inv}_v$.
When $P_i$ invokes $\mathit{Inv}_v$, we have that $\mathit{epoch}_i = e_{\mathit{final}}$ and $\mathit{view}_i = j - 1$.
As in the base step, \Cref{lemma:no_epoch_over_for_epoch_duration} shows that $\mathit{Inv}_v$ cannot be canceled (line~\ref{line:cancel_view_timer_1} or line~\ref{line:cancel_view_timer_2}) as $t_{e_{\mathit{final}}} + \mathit{epoch\_duration} > t^{j - 1} + \mathit{view\_duration}$ since $t^{j - 1} \leq t_{e_{\mathit{final}}} + (j - 2)\cdot \mathit{view\_duration} + 2\delta$ (by the induction hypothesis).
We denote by $\mathit{Exp}_v$ the expiration event produced by $\mathit{Inv}_v$.
By \Cref{lemma:epoch_and_view_no_change}, when $P_i$ processes $\mathit{Exp}_v$ (line~\ref{line:rule_view_expired}), we have that $\mathit{epoch}_i = e_{\mathit{final}}$ and $\mathit{view}_i = j - 1 < f + 1$.
Hence, $P_i$ enters the $j$-th view of $e_{\mathit{final}}$ at time $t^j = t^{j - 1} + \mathit{view\_duration}$ (line~\ref{line:start_view_without_msg_exchange}), which means that $t^j \in \big[ t_{e_{\mathit{final}}} + (j - 1) \cdot \mathit{view\_duration}, t_{e_{\mathit{final}}} + (j - 1) \cdot \mathit{view\_duration} + 2\delta \big]$.

We now separate two cases:
\begin{compactitem}
    \item Let $j < f + 1$.
    Just prior to entering the $j$-th view of $e_{\mathit{final}}$ (line~\ref{line:start_view_without_msg_exchange}), $P_i$ invokes the $\mathsf{measure}(\mathit{view\_duration})$ method on $\mathit{view\_timer}_i$ (line~\ref{line:view_timer_measure_without_msg_exchange}); we denote this invocation by $\mathit{Inv}'_v$.
    By \Cref{lemma:no_epoch_over_for_epoch_duration}, $\mathit{Inv}'_v$ cannot be canceled (line~\ref{line:cancel_view_timer_1} or line~\ref{line:cancel_view_timer_2}) as $t_{e_{\mathit{final}}} + \mathit{epoch\_duration} > t_{e_{\mathit{final}}} + (j - 1) \cdot \mathit{view\_duration} + 2\delta + \mathit{view\_duration}$.
    Therefore, $\mathit{Inv}'_v$ produces an expiration event $\mathit{Exp}'_v$ which is processed by $P_i$ at time $t^j + \mathit{view\_duration}$ (since $t^j \geq \mathit{GST}$ and local clocks do not drift after $\mathit{GST}$).

    Let us investigate the first time $P_i$ enters another view after entering the $j$-th view of $e_{\mathit{final}}$.
    This could happen at the following places of \Cref{algorithm:synchronizer}:
    \begin{compactitem}
        \item line~\ref{line:start_view_without_msg_exchange}:
        By \Cref{lemma:expiration_after_invocation_view}, we conclude that this occurs at time $\geq t^j + \mathit{view\_duration}$.
        Therefore, in this case, $P_i$ is in the $j$-th view of $e_{\mathit{final}}$ during the time period $[t^j, t^j + \mathit{view\_duration})$.
        The induction step is proven in this case.
    
        \item line~\ref{line:start_view_2}:
        By contradiction, suppose that this happens before time $t^j + \mathit{view\_duration}$.
        Hence, the $\mathsf{measure}(\cdot)$ method was invoked on $\mathit{dissemination\_timer}_i$ (line~\ref{line:sync_timer_measure} or line~\ref{line:sync_timer_measure_2}) before time $t^j + \mathit{view\_duration}$ and after the invocation of $\mathit{Inv}'_v$ (by \Cref{lemma:expiration_after_invocation_dissemination}).
        Thus, $\mathit{Inv}'_v$ is canceled (line~\ref{line:cancel_view_timer_1} or line~\ref{line:cancel_view_timer_2}), which is impossible (as previously proven).
    
        Hence, $P_i$ is in the $j$-th view of $e_{\mathit{final}}$ during (at least) the time period $[t^j, t^j + \mathit{view\_duration})$, which concludes the induction step even in this case.
    \end{compactitem}
    
    \item Let $j = f + 1$.
    Just prior to entering the $j$-th view of $e_{\mathit{final}}$ (line~\ref{line:start_view_without_msg_exchange}), $P_i$ invokes the $\mathsf{measure}(\mathit{view\_duration})$ method on $\mathit{view\_timer}_i$ (line~\ref{line:view_timer_measure_without_msg_exchange}); we denote this invocation by $\mathit{Inv}'_v$.
    When $\mathit{Inv}_v'$ was invoked, $\mathit{epoch}_i = e_{\mathit{final}}$ and $\mathit{view}_i = f + 1$.
    By \Cref{lemma:no_epoch_over_for_epoch_duration}, we know that the earliest time $\mathit{Inv}_v'$ can be canceled (line~\ref{line:cancel_view_timer_1} or line~\ref{line:cancel_view_timer_2}) is $t_{e_{\mathit{final}}} + \mathit{epoch\_duration}$.
    
    Let us investigate the first time $P_i$ enters another view after entering the $j$-th view of $e_{\mathit{final}}$.
    This could happen at the following places of \Cref{algorithm:synchronizer}:
    \begin{compactitem}
        \item line~\ref{line:start_view_without_msg_exchange}:
        This means that, when processing the expiration event of $\mathit{view\_timer}_i$ (denoted by $\mathit{Exp}_v^*$) at line~\ref{line:rule_view_expired} (before executing the check at line~\ref{line:check_last_view}), $\mathit{view}_i < f + 1$.
        Hence, $\mathit{Exp}_v^*$ is not produced by $\mathit{Inv}_v'$ (by \Cref{lemma:epoch_and_view_no_change}).
        
        By contradiction, suppose that $\mathit{Exp}_v^*$ is processed before time $t_{e_{\mathit{final}}} + \mathit{epoch\_duration}$.
        In this case, $\mathit{Exp}_v^*$ is processed before the expiration event produced by $\mathit{Inv}_v'$ would (potentially) be processed (which is $t_{e_{\mathit{final}}} + \mathit{epoch\_duration}$ at the earliest).
        Thus, $\mathit{Inv}_v'$ must be immediately followed by an invocation of the $\mathsf{cancel()}$ method on $\mathit{view\_timer}_i$ in $h_i|_{\mathit{view}}$ (by lemmas~\ref{lemma:view_timer_behavior} and~\ref{lemma:expiration_after_invocation}).
        As previously shown, the earliest time $\mathit{Inv}'_v$ can be canceled is $t_{e_\mathit{final}} + \mathit{epoch\_duration}$, which implies that $\mathit{Exp}_v^*$ cannot be processed before time $t_{e_{\mathit{final}}} + \mathit{epoch\_duration}$.
        Therefore, $\mathit{Exp}_v^*$ is processed at $t_{e_{\mathit{final}}} + \mathit{epoch\_duration}$ (at the earliest), which concludes the induction step for this case.

        \item line~\ref{line:start_view_2}:
        Suppose that, by contradiction, this happens before time $t_{e_{\mathit{final}}} + \mathit{epoch\_duration}$.
        Hence, the $\mathsf{measure}(\cdot)$ method was invoked on $\mathit{dissemination\_timer}_i$ (line~\ref{line:sync_timer_measure} or line~\ref{line:sync_timer_measure_2}) before time $t_{e_{\mathit{final}}} + \mathit{epoch\_duration}$ (by \Cref{lemma:expiration_after_invocation_dissemination}) and after $P_i$ has entered the $j$-th view of $e_{\mathit{final}}$, which implies that $\mathit{Inv}_v'$ is canceled before time $t_{e_{\mathit{final}}} + \mathit{epoch\_duration}$ (line~\ref{line:cancel_view_timer_1} or line~\ref{line:cancel_view_timer_2}).
        However, this is impossible as the earliest time for $\mathit{Inv}'_v$ to be canceled is $t_{e_{\mathit{final}}} + \mathit{epoch\_duration}$.
        Hence, $P_i$ enters another view at time $t_{e_{\mathit{final}}} + \mathit{epoch\_duration}$ (at the earliest), which concludes the induction step in this case.
    \end{compactitem}
\end{compactitem}
The conclusion of the induction step concludes the proof of the lemma.
\end{proof}

Finally, we prove that \rare ensures the eventual synchronization property.

\begin{theorem} [Eventual synchronization] \label{theorem:termination}
\rare ensures eventual synchronization.
Moreover, the first synchronization time at or after $\mathit{GST}$ occurs by time $t_{e_\mathit{final}} + f\cdot\mathit{view\_duration} + 2\delta$.
\end{theorem}
\begin{proof}
\Cref{lemma:time_in_final} proves that all correct processes overlap in each view of $e_{\mathit{final}}$ for (at least) $\Delta$ time.
As the leader of one view of $e_{\mathit{final}}$ must be correct (since $\mathsf{leader}(\cdot)$ is a round-robin function), the eventual synchronization is satisfied by \rare: correct processes synchronize in (at least) one of the views of $e_{\mathit{final}}$.
Finally, as the last view of $e_{\mathit{final}}$ is entered by every correct process by time $t^* = t_{e_{\mathit{final}}} + f \cdot \mathit{view\_duration} + 2\delta$ (by \Cref{lemma:time_in_final}), the first synchronization time at or after $\mathit{GST}$ must occur by time $t^*$.
\end{proof}

\noindent \textbf{Proof of complexity.}
We start by showing that, if a correct process sends an \textsc{epoch-completed} message for an epoch $e$, then the ``most recent'' epoch entered by the process is $e$.

\begin{lemma} \label{lemma:epoch_completed_last_enter}
Let $P_i$ be a correct process and let $P_i$ send an \textsc{epoch-completed} message for an epoch $e$ (line~\ref{line:broadcast_epoch_over}).
Then, $e$ is the last epoch entered by $P_i$ in $\beta_i$ before sending the \textsc{epoch-completed} message.
\end{lemma}
\begin{proof}
By \Cref{lemma:epoch_over_previously_started}, $P_i$ enters $e$ before sending the \textsc{epoch-completed} message for $e$.
By contradiction, suppose that $P_i$ enters some other epoch $e^*$ after entering $e$ and before sending the \textsc{epoch-completed} message for $e$.
By \Cref{lemma:increasing_views}, $e^* > e$.

When $P_i$ enters $e^*$ (line~\ref{line:start_view_2}), $\mathit{epoch}_i = e^*$.
As the value of the $\mathit{epoch}_i$ variable only increases throughout the execution, $P_i$ does not send the \textsc{epoch-completed} message for $e$ after entering $e^* > e$.
Thus, we reach a contradiction, and the lemma holds.
\end{proof}

Next, we show that, if a correct process sends an \textsc{enter-epoch} message for an epoch $e$ at time $t$, the process enters $e$ at time $t$.

\begin{lemma} \label{lemma:enter_epoch_in_e}
Let a correct process $P_i$ send an \textsc{enter-epoch} message (line~\ref{line:broadcast_epoch_over_complete}) for an epoch $e$ at time $t$.
Then, $P_i$ enters $e$ at time $t$.
\end{lemma}
\begin{proof}
When $P_i$ sends the \textsc{enter-epoch} message, we have that $\mathit{epoch}_i = e$.
Hence, $P_i$ enters $e$ at time $t$ (line~\ref{line:start_view_2}).
\end{proof}

Next, we show that a correct process sends (at most) $O(n)$ \textsc{epoch-completed} messages for a specific epoch $e$.

\begin{lemma} \label{lemma:epoch_completed_n}
For any epoch $e$ and any correct process $P_i$, $P_i$ sends at most $O(n)$ \textsc{epoch-completed} messages for $e$ (line~\ref{line:broadcast_epoch_over}).
\end{lemma}
\begin{proof}
Let $\mathit{Exp}_v$ denote the first expiration event of $\mathit{view\_timer}_i$ which $P_i$ processes (line~\ref{line:rule_view_expired}) in order to broadcast the \textsc{epoch-completed} message for $e$ (line~\ref{line:broadcast_epoch_over}); if $\mathit{Exp}_v$ does not exist, the lemma trivially holds.
Hence, let $\mathit{Exp}_v$ exist.

When $\mathit{Exp}_v$ was processed, $\mathit{epoch}_i = e$.
Let $\mathit{Inv}'_v$ denote the first invocation of the $\mathsf{measure}(\cdot)$ method on $\mathit{view\_timer}_i$ after the processing of $\mathit{Exp}_v$.
If $\mathit{Inv}'_v$ does not exist, there does not exist an expiration event of $\mathit{view\_timer}_i$ processed after $\mathit{Exp}_v$ (by \Cref{lemma:expiration_after_invocation_view}), which implies that the lemma trivially holds.

Let us investigate where $\mathit{Inv}'_v$ could have been invoked:
\begin{compactitem}
    \item line~\ref{line:view_timer_measure_without_msg_exchange}:
    By \Cref{lemma:expiration_after_invocation_view}, we conclude that the processing of $\mathit{Exp}_v$ leads to $\mathit{Inv}'_v$.
    However, this is impossible as the processing of $\mathit{Exp}_v$ leads to the broadcasting of the \textsc{epoch-completed} messages (see the check at line~\ref{line:check_last_view}).
    
    \item line~\ref{line:view_timer_measure}:
    In this case, $P_i$ processes an expiration event $\mathit{Exp}_d$ of $\mathit{dissemination\_timer}_i$ (line~\ref{line:sync_timer_expires}).
    By \Cref{lemma:expiration_after_invocation_dissemination}, the invocation $\mathit{Inv}_d$ of the $\mathsf{measure}(\cdot)$ method on $\mathit{dissemination\_timer}_i$ immediately precedes $\mathit{Exp}_d$ in $h_i$.
    Hence, $\mathit{Inv}_d$ follows $\mathit{Exp}_v$ in $h_i$ and $\mathit{Inv}_d$ could have been invoked either at line~\ref{line:sync_timer_measure} or at line~\ref{line:sync_timer_measure_2}.
    Just before invoking $\mathit{Inv}_d$, $P_i$ changes its $\mathit{epoch}_i$ variable to a value greater than $e$ (line~\ref{line:update_epoch_1} or line~\ref{line:update_epoch_2}; the value of $\mathit{epoch}_i$ only increases throughout the execution).
\end{compactitem}
Therefore, when $\mathit{Inv}'_v$ is invoked, $\mathit{epoch}_i > e$.
As the value of the $\mathit{epoch}_i$ variable only increases throughout the execution, $P_i$ broadcasts the \textsc{epoch-completed} messages for $e$ at most once (by \Cref{lemma:expiration_after_invocation_view}), which concludes the proof.
\end{proof}

The following lemma shows that a correct process sends (at most) $O(n)$ \textsc{enter-epoch} messages for a specific epoch $e$.

\begin{lemma} \label{lemma:enter_epoch_n}
For any epoch $e$ and any correct process $P_i$, $P_i$ sends at most $O(n)$ \textsc{enter-epoch} messages for $e$ (line~\ref{line:broadcast_epoch_over_complete}).
\end{lemma}
\begin{proof}
Let $\mathit{Exp}_d$ denote the first expiration event of $\mathit{dissemination\_timer}_i$ which $P_i$ processes (line~\ref{line:sync_timer_expires}) in order to broadcast the \textsc{enter-epoch} message for $e$ (line~\ref{line:broadcast_epoch_over_complete}); if $\mathit{Exp}_d$ does not exist, the lemma trivially holds.
When $\mathit{Exp}_d$ was processed, $\mathit{epoch}_i = e$.
Let $\mathit{Inv}'_d$ denote the first invocation of the $\mathsf{measure}(\cdot)$ method on $\mathit{dissemination\_timer}_i$ after the processing of $\mathit{Exp}_d$.
If $\mathit{Inv}'_d$ does not exist, there does not exist an expiration event of $\mathit{dissemination\_timer}_i$ processed after $\mathit{Exp}_d$ (by \Cref{lemma:expiration_after_invocation_dissemination}), which implies that the lemma trivially holds.

$\mathit{Inv}'_d$ could have been invoked either at line~\ref{line:sync_timer_measure} or at line~\ref{line:sync_timer_measure_2}.
However, before that (still after the processing of $\mathit{Exp}_d$), $P_i$ changes its $\mathit{epoch}_i$ variable to a value greater than $e$ (line~\ref{line:update_epoch_1} or line~\ref{line:update_epoch_2}).
Therefore, when $\mathit{Inv}'_d$ is invoked, $\mathit{epoch}_i > e$.
As the value of the $\mathit{epoch}_i$ variable only increases throughout the execution, $P_i$ broadcasts the \textsc{enter-epoch} messages for $e$ at most once (by \Cref{lemma:expiration_after_invocation_dissemination}), which concludes the proof.
\end{proof}

Next, we show that, after $\mathit{GST}$, two ``epoch-entering'' events are separated by at least $\delta$ time.

\begin{lemma} \label{lemma:within_delta}
Let $P_i$ be a correct process.
Let $P_i$ trigger $\mathsf{advance}(v)$ at time $t \geq \mathit{GST}$ and let $P_i$ trigger $\mathsf{advance}(v')$ at time $t'$ such that (1) $\mathsf{advance}(v) \stackrel{\beta_i}{\prec} \mathsf{advance}(v')$, and (2) $v$ (resp., $v'$) is the first view of an epoch $e$ (resp., $e'$).
Then, $t' \geq t + \delta$.
\end{lemma}
\begin{proof}
Let $\mathsf{advance}(v^*)$, where $v^*$ is the first view of an epoch $e^*$, be the first ``epoch-entering'' event following $\mathsf{advance}(v)$ in $\beta_i$ (i.e., $\mathsf{advance}(v) \stackrel{\beta_i}{\prec} \mathsf{advance}(v^*)$); let $\mathsf{advance}(v^*)$ be triggered at time $t^*$.
In order to prove the lemma, it suffices to show that $t^* \geq t + \delta$.

The $\mathsf{advance}(v^*)$ upcall is triggered at line~\ref{line:start_view_2}.
Let $\mathit{Exp}_d$ denote the processed expiration event of $\mathit{dissemination\_timer}_i$ (line~\ref{line:sync_timer_expires}) which leads $P_i$ to trigger $\mathsf{advance}(v^*)$.
Let $\mathit{Inv}_d$ denote the invocation of the $\mathsf{measure}(\delta)$ on $\mathit{dissemination\_timer}_i$ that has produces $\mathit{Exp}_d$.
By \Cref{lemma:expiration_after_invocation_dissemination}, $\mathit{Inv}_d$ immediately precedes $\mathit{Exp}_d$ in the timer history $h_i$ of $P_i$.
Note that $\mathit{Inv}_d$ was invoked after $P_i$ has entered $e$ (this follows from \Cref{lemma:expiration_after_invocation_dissemination} and the fact that $P_i$ enters $e$ after invoking $\mathsf{measure}(\cdot)$ on $\mathit{view\_timer}_i$), which means that $\mathit{Inv}_d$ was invoked at some time $\geq t \geq \mathit{GST}$.
As local clocks do not drift after $\mathit{GST}$, $\mathit{Exp}_d$ is processed at some time $\geq t + \delta$, which concludes the proof.
\end{proof}

Next, we define $t_s$ as the first synchronization time at or after $\mathit{GST}$.

\begin{definition} \label{definition:t_s}
We denote by $t_s$ the first synchronization time at or after $\mathit{GST}$ (i.e., $t_s \geq \mathit{GST}$).
\end{definition}

The next lemma shows that no correct process enters any epoch greater than $e_{\mathit{final}}$ by $t_s + \Delta$.
This lemma is the consequence of \Cref{lemma:no_epoch_over_for_epoch_duration} and \Cref{theorem:termination}.

\begin{lemma} \label{lemma:no_process_starts_emax_1}
No correct process enters an epoch greater than $e_{\mathit{final}}$ by time $t_s + \Delta$. 
\end{lemma}
\begin{proof}
By \Cref{lemma:no_epoch_over_for_epoch_duration}, no correct process enters an epoch $> e_{\mathit{final}}$ before time $t_{e_{\mathit{final}}} + \mathit{epoch\_duration}$.
By \Cref{theorem:termination}, we have that $t_s < t_{e_{\mathit{final}}} + \mathit{epoch\_duration} - \Delta$, which implies that $t_{e_{\mathit{final}}} + \mathit{epoch\_duration} > t_s + \Delta$.
Hence, the lemma.
\end{proof}

Next, we define $e_{\mathit{max}}$ as the greatest epoch entered by a correct process before time $\mathit{GST}$.
Note that $e_{\mathit{max}}$ is properly defined in any execution as only finite executions are possible until $\mathit{GST}$.

\begin{definition} \label{definition:max}
We denote by $e_{\mathit{max}}$ the greatest epoch entered by a correct process before $\mathit{GST}$.
If no such epoch exists, $e_{\mathit{max}} = 0$.
\end{definition}

The next lemma shows that $e_{\mathit{final}}$ (\Cref{definition:final}) is $e_{\mathit{max}} + 1$.

\begin{lemma} \label{lemma:e_final_e_max}
$e_{\mathit{final}} = e_{\mathit{max}} + 1$.
\end{lemma}
\begin{proof}
If $e_{\mathit{max}} = 0$, then $e_{\mathit{final}} = 1$. 
Hence, let $e_{\mathit{max}} > 0$ in the rest of the proof.

By the definitions of $e_{\mathit{final}}$ (\Cref{definition:final}) and $e_{\mathit{max}}$ (\Cref{definition:max}) and by \Cref{lemma:epoch_previously_started}, $e_{\mathit{final}} \geq e_{\mathit{max}} + 1$.
Therefore, we need to prove that $e_{\mathit{final}} \leq e_{\mathit{max}} + 1$.

By contradiction, suppose that $e_{\mathit{final}} > e_{\mathit{max}} + 1$.
By \Cref{lemma:epoch_previously_started}, epoch $e_{\mathit{final}} - 1$ was entered by the first correct process at some time $t_{\mathit{prev}} \leq t_{e_\mathit{final}}$.
Note that $e_{\mathit{final}} - 1 \geq e_{\mathit{max}} + 1$.
Moreover, $t_{\mathit{prev}} \geq \mathit{GST}$; otherwise, we would contradict the definition of $e_{\mathit{max}}$. 
Thus, the first new epoch to be entered by a correct process at or after $\mathit{GST}$ is not $e_{\mathit{final}}$, i.e., we contradict \Cref{definition:final}.
Hence, the lemma holds.
\end{proof}

Next, we show that every correct process enters epoch $e_{\mathit{max}}$ by time $\mathit{GST} + 2\delta$ or epoch $e_{\mathit{final}} = e_{\mathit{max}} + 1$ by time $\mathit{GST} + 3\delta$.

\begin{lemma} \label{lemma:3delta}
Every correct process (1) enters epoch $e_{\mathit{max}}$ by $\mathit{GST} + 2\delta$, or (2) enters epoch $ e_{\mathit{max}} + 1$ by $\mathit{GST} + 3\delta$.
\end{lemma}
\begin{proof}
\Cref{lemma:e_final_e_max} shows that $e_{\mathit{final}}$ is $e_{\mathit{max}} + 1$.
Recall that $t_{e_\mathit{final}} \geq \mathit{GST}$.
Consider a correct process $P_i$.
If $e_{\mathit{max}} = 1$ (resp., $e_{\mathit{final}} = 1$), then $P_i$ enters $e_{\mathit{max}}$ (resp., $e_{\mathit{final}}$) by time $\mathit{GST}$, which concludes the lemma.
Hence, let $e_{\mathit{max}} > 1$; thus, $e_{\mathit{final}} > 1$ by \Cref{lemma:e_final_e_max}.

    
    
\Cref{lemma:no_epoch_over_for_epoch_duration} proves that no correct process broadcasts an \textsc{epoch-completed} message for an epoch $\geq e_{\mathit{max}} + 1$ before time $t_{e_\mathit{final}} + \mathit{epoch\_duration} \geq \mathit{GST} + \mathit{epoch\_duration}$.
    
By time $\mathit{GST} + \delta$, every correct process $P_i$ receives an \textsc{enter-epoch} message for epoch $e_{\mathit{max}} > 1$ (line~\ref{line:receive_epoch_over_complete}) sent by the correct process which has entered $e_{\mathit{max}}$ before $\mathit{GST}$ (the message is sent at line~\ref{line:broadcast_epoch_over_complete}).
Therefore, by time $\mathit{GST} + \delta$, $\mathit{epoch}_i$ is either $e_{\mathit{max}}$ or $e_{\mathit{max}} + 1$; note that $\mathit{epoch}_i$ cannot take a value greater than $e_{\mathit{max}} + 1$ before time $\mathit{GST} + \mathit{epoch\_duration} > \mathit{GST} + \delta$ since no correct process broadcasts an \textsc{epoch-completed} message for an epoch $\geq e_{\mathit{max}} + 1$ before this time.

Let us consider both scenarios:
\begin{compactitem}
    \item Let $\mathit{epoch}_i = e_{\mathit{max}} + 1$ by time $\mathit{GST} + \delta$.
    In this case, $\mathit{dissemination\_timer}_i$ expires in $\delta$ time (line~\ref{line:sync_timer_expires}), and $P_i$ enters $e_{\mathit{max}} + 1$ by time $\mathit{GST} + 2\delta$ (line~\ref{line:start_view_2}) as $\mathit{GST} + \mathit{epoch\_duration} > \mathit{GST} + 2\delta$.
    Hence, the statement of the lemma is satisfied in this case.
    
    \item Let $\mathit{epoch}_i = e_{\mathit{max}}$ by time $\mathit{GST} + \delta$.
    If, within $\delta$ time from updating $\mathit{epoch}_i$ to $e_{\mathit{max}}$, $P_i$ does not cancel its $\mathit{dissemination\_timer}_i$, $\mathit{dissemination\_timer}_i$ expires (line~\ref{line:rule_view_expired}), and $P_i$ enters $e_{\mathit{max}}$ by time $\mathit{GST} + 2\delta$.
    
    Otherwise, $\mathit{epoch}_i = e_{\mathit{max}} + 1$ by time $\mathit{GST} + 2\delta$ as $\mathit{dissemination\_timer}_i$ was canceled; $\mathit{epoch}_i$ cannot take any other value as \textsc{epoch-completed} messages are not broadcast before time $\mathit{GST} + \mathit{epoch\_duration} > \mathit{GST} + 2\delta$.
    As in the previous case, $\mathit{dissemination\_timer}_i$ expires in $\delta$ time (line~\ref{line:sync_timer_expires}), and $P_i$ enters $e_{\mathit{max}} + 1$ by time $\mathit{GST} + 3\delta$ (line~\ref{line:start_view_2}) as $\mathit{GST} + \mathit{epoch\_duration} > \mathit{GST} + 3\delta$.
    Hence, the statement of the lemma holds in this case, as well.
\end{compactitem}
Since the lemma is satisfied in both possible scenarios, the proof is concluded.
\end{proof}

The direct consequence of \Cref{lemma:e_final_e_max} is that $t_{e_\mathit{final}} \leq \mathit{GST} + \mathit{epoch\_duration} + 4\delta$.

\begin{lemma} \label{lemma:delta}
$t_{e_\mathit{final}} \leq \mathit{GST} + \mathit{epoch\_duration} + 4\delta$.
\end{lemma}
\begin{proof}
By contradiction, let $t_{e_\mathit{final}} > \mathit{GST} + \mathit{epoch\_duration} + 4\delta$.
\Cref{lemma:3delta} proves that every correct process enters epoch $e_{\mathit{max}}$ by time $\mathit{GST} + 2\delta$ or epoch $e_{\mathit{final}} = e_{\mathit{max}} + 1$ by time $\mathit{GST} + 3\delta$.
Additionally, \Cref{lemma:no_epoch_over_for_epoch_duration} proves that no correct process broadcasts an \textsc{epoch-completed} message for an epoch $\geq e_{\mathit{final}}$ (line~\ref{line:broadcast_epoch_over}) before time $t_{e_\mathit{final}} + \mathit{epoch\_duration} > \mathit{GST} + 2\cdot\mathit{epoch\_duration} + 4\delta$.

If any correct process enters $e_{\mathit{max}} + 1$ by time $\mathit{GST} + 3\delta$, we reach a contradiction with the fact that $t_{e_\mathit{final}} > \mathit{GST} + \mathit{epoch\_duration} + 4\delta$ since $e_{\mathit{final}} = e_{\mathit{max}} + 1$ (by \Cref{lemma:e_final_e_max}).
Therefore, all correct processes enter $e_{\mathit{max}}$ by time $\mathit{GST} + 2\delta$.

Since $t_{e_{\mathit{final}}} > \mathit{GST} + \mathit{epoch\_duration} + 4\delta$, no correct process $P_i$ updates its $\mathit{epoch}_i$ variable to $e_{\mathit{max}} + 1$ (at line~\ref{line:update_epoch_1} or line~\ref{line:update_epoch_2}) by time $\mathit{GST} + \mathit{epoch\_duration} + 3\delta$ (otherwise, $P_i$ would have entered $e_{\mathit{max}} + 1$ by time $\mathit{GST} + \mathit{epoch\_duration} + 4\delta$, which contradicts $t_{e_{\mathit{final}}} > \mathit{GST} + \mathit{epoch\_duration} + 4\delta$).
By time $\mathit{GST} + \mathit{epoch\_duration} + 2\delta$, all correct processes broadcast an \textsc{epoch-completed} message for $e_{\mathit{max}}$ (line~\ref{line:broadcast_epoch_over}).
By time $\mathit{GST} + \mathit{epoch\_duration} + 3\delta$, every correct process $P_i$ receives $2f + 1$ \textsc{epoch-completed} messages for $e_{\mathit{max}}$ (line~\ref{line:receive_epoch_over}), and updates its $\mathit{epoch}_i$ variable to $e_{\mathit{max}} + 1$ (line~\ref{line:update_epoch_1}).
This represents a contradiction with the fact that $P_i$ does not update its $\mathit{epoch}_i$ variable to $e_{\mathit{max}} + 1$ by time $\mathit{GST} + \mathit{epoch\_duration} + 3\delta$, which concludes the proof.
\end{proof}

The final lemma shows that no correct process enters more than $O(1)$ epochs during the time period $[\mathit{GST}, t_s + \Delta]$.

\begin{lemma} \label{lemma:constant-number-of-epochs-after-GST}
No correct process enters more than $O(1)$ epochs in the time period $[\mathit{GST}, t_s + \Delta]$.
\end{lemma}
\begin{proof}
Consider a correct process $P_i$.
Process $P_i$ enters epoch $e_{\mathit{max}}$ by time $\mathit{GST} + 2\delta$ or $P_i$ enters epoch $e_{\mathit{max}} + 1$ by time $\mathit{GST} + 3\delta$ (by \Cref{lemma:3delta}).
\Cref{lemma:e_final_e_max} shows that $e_{\mathit{final}} = e_{\mathit{max}} + 1$.
Finally, no correct process enters an epoch greater than $e_{\mathit{final}} = e_{\mathit{max}} + 1$ by time $t_s + \Delta$ (by \Cref{lemma:no_process_starts_emax_1}).

Let us consider two scenarios according to \Cref{lemma:3delta}:
\begin{compactenum}
    \item By time $\mathit{GST} + 2\delta$, $P_i$ enters $e_{\mathit{max}}$; let $P_i$ enter $e_{\mathit{max}}$ at time $t^* \leq \mathit{GST} + 2\delta$.
    By \Cref{lemma:increasing_views}, during the time period $[t^*, t_s + \Delta]$, $P_i$ enters (at most) $2 = O(1)$ epochs (epochs $e_{\mathit{max}}$ and $e_{\mathit{max}} + 1$).
    Finally, during the time period $[\mathit{GST}, t^*)$, \Cref{lemma:within_delta} shows that $P_i$ enters (at most) $2 = O(1)$ epochs (as $t^* \leq \mathit{GST} + 2\delta$).
    Hence, in this case, $P_i$ enters (at most) $4 = O(1)$ epochs during the time period $[\mathit{GST}, t_s + \Delta]$.
    
    \item By time $\mathit{GST} + 3\delta$, $P_i$ enters $e_{\mathit{max}} + 1$; let $P_i$ enter $e_{\mathit{max}} + 1$ at time $t^* \leq \mathit{GST} + 3\delta$.
    By \Cref{lemma:increasing_views}, during the time period $[t^*, t_s + \Delta]$, $P_i$ enters (at most) $1 = O(1)$ epoch (epoch $e_{\mathit{max}} + 1$).
    Finally, during the time period $[\mathit{GST}, t^*)$, \Cref{lemma:within_delta} shows that $P_i$ enters (at most) $3 = O(1)$ epochs (as $t^* \leq \mathit{GST} + 3\delta$).
    Hence, in this case, $P_i$ enters (at most) $4 = O(1)$ epochs during the time period $[\mathit{GST}, t_s + \Delta]$.
    
\end{compactenum}
Hence, during the time period $[\mathit{GST}, t_s + \Delta]$, $P_i$ enters (at most) $4 = O(1)$ epochs.
\end{proof}

Finally, we prove that \rare achieves $O(n^2)$ communication and $O(f)$ latency.

\begin{theorem} [Complexity] \label{theorem:latency}
\rare achieves $O(n^2)$ communication complexity and $O(f)$ latency complexity. 
\end{theorem}
\begin{proof}
Fix a correct process $P_i$.
For every epoch $e$, $P_i$ sends (at most) $O(n)$ \textsc{epoch-completed} and \textsc{enter-epoch} messages for $e$ (by lemmas~\ref{lemma:epoch_completed_n} and~\ref{lemma:enter_epoch_n}).
Moreover, if $P_i$ sends an \textsc{epoch-completed} message for an epoch $e$ at time $t$, then $e$ is the last epoch entered by $P_i$ prior to sending the message (by \Cref{lemma:epoch_completed_last_enter}).
Similarly, if $P_i$ sends an \textsc{enter-epoch} message for an epoch $e$ at time $t$, then $P_i$ enters $e$ at $t$ (by \Cref{lemma:enter_epoch_in_e}).
Hence, during the time period $[\mathit{GST}, t_s + \Delta]$, $P_i$ sends \textsc{epoch-completed} or \textsc{enter-epoch} messages for (at most) $O(1)$ epochs (by \Cref{lemma:constant-number-of-epochs-after-GST}).
Thus, $P_i$ sends (at most) $O(1) \cdot O(n) = O(n)$ messages during the time period $[\mathit{GST}, t_s + \Delta]$, which implies that $P_i$ sends $O(n)$ words in this time period (as each \textsc{epoch-completed} and \textsc{enter-epoch} message contains a single word).
Therefore, the communication complexity of \rare is indeed $n \cdot O(n) = O(n^2)$.

By \Cref{theorem:termination}, $t_s + \Delta < t_{e_\mathit{final}} + \mathit{epoch\_duration}$.
Moreover, \Cref{lemma:delta} shows that $t_{e_\mathit{final}} \leq \mathit{GST} + \mathit{epoch\_duration} + 4\delta$.
Therefore, $t_s + \Delta < \mathit{GST} + 2\cdot\mathit{epoch\_duration} + 4\delta$.
Furthermore, $t_s + \Delta - \mathit{GST} < 2\cdot\mathit{epoch\_duration} + 4\delta$.
Since $\mathit{epoch\_duration} = (f + 1) \cdot \mathit{view\_duration} = O(f)$ (recall that $\mathit{view\_duration}$ is constant), $t_s + \Delta - \mathit{GST} = O(f)$, which proves the linear latency complexity of \rare.
\end{proof}
\section{\name: Pseudocode \& Proof of Correctness and Complexity}\label{section:appendix_quad}

In this section, we give the complete pseudocode of \name's view core module (\cref{algorithm:utilities,algorithm:view_1}), and we formally prove that \name solves consensus (with weak validity) with $O(n^2)$ communication complexity and $O(f)$ latency complexity.

\begin{algorithm} [h]
\caption{\name: View core's utilities (for process $P_i$)}
\label{algorithm:utilities}
\begin{algorithmic} [1]
\State \textbf{function} $\mathsf{msg(String} \text{ } \mathit{type}, \mathsf{Value} \text{ } \mathit{value}, \mathsf{Quorum\_Certificate} \text{ } \mathit{qc}, \mathsf{View} \text{ }\mathit{view}\mathsf{)}$:
\State \hskip2em $m.\mathit{type} \gets \mathit{type}$; $m.\mathit{value} \gets \mathit{value}$; $m.\mathit{qc} \gets \mathit{qc}$; $m.\mathit{view} \gets \mathit{view}$
\State \hskip2em \textbf{return} $m$

\smallskip
\State \textbf{function} $\mathsf{vote\_msg(String} \text{ } \mathit{type}, \mathsf{Value} \text{ } \mathit{value}, \mathsf{Quorum\_Certificate} \text{ } \mathit{qc}, \mathsf{View} \text{ }\mathit{view}\mathsf{)}$:
\State \hskip2em $m \gets \mathsf{msg(}\mathit{type}, \mathit{value}, \mathit{qc}, \mathit{view}\mathsf{)}$
\State \hskip2em $m.\mathit{partial\_sig} \gets \mathit{ShareSign}_i([m.\mathit{type}, m.\mathit{value}, m.\mathit{view}])$
\State \hskip2em \textbf{return} $m$

\smallskip
\State \textcolor{blue}{\(\triangleright\) All the messages in $M$ have the same type, value and view}
\State \textbf{function} $\mathsf{qc(Set(Vote\_Message)} \text{ } M\mathsf{)}$:
\State \hskip2em $\mathit{qc}.\mathit{type} \gets m.\mathit{type}$, where $m \in M$
\State \hskip2em $\mathit{qc}.\mathit{value} \gets m.\mathit{value}$, where $m \in M$
\State \hskip2em $\mathit{qc}.\mathit{view} \gets m.\mathit{view}$, where $m \in M$
\State \hskip2em $\mathit{qc}.\mathit{sig} \gets \mathit{Combine}\big(\{\mathit{partial\_sig} \,|\, \mathit{partial\_sig} \text{ is in a message that belongs to } M\}\big)$
\State \hskip2em \textbf{return} $\mathit{qc}$

\smallskip
\State \textbf{function} $\mathsf{matching\_msg(Message} \text{ }m, \mathsf{String} \text{ }\mathit{type}, \mathsf{View} \text{ } \mathit{view}\mathsf{)}$:
\State \hskip2em \textbf{return} $m.\mathit{type} = \mathit{type}$ and $m.\mathit{view} = \mathit{view}$

\smallskip
\State \textbf{function} $\mathsf{matching\_qc(Quorum\_Certificate} \text{ } \mathit{qc}, \mathsf{String} \text{ } \mathit{type}, \mathsf{View} \text{ } \mathit{view}\mathsf{)}$:
\State \hskip2em \textbf{return} $\mathit{qc}.\mathit{type} = \mathit{type}$ and $\mathit{qc}.\mathit{view} = \mathit{view}$
\end{algorithmic}
\end{algorithm}

\begin{algorithm} 
\caption{\name: View core (for process $P_i$)}
\label{algorithm:view_1}
\begin{algorithmic} [1]
\State \textbf{upon} $\mathsf{init(Value } \text{ } \mathit{proposal})$:
\State \hskip2em $\mathit{proposal}_i \gets \mathit{proposal}$ \BlueComment{$P_i$'s proposal}

\smallskip
\State \textbf{upon} $\mathsf{start\_executing(View} \text{ } \mathit{view}\mathsf{)}$:

\State \hskip2em \textcolor{blue}{\(\triangleright\) Prepare phase}
\State \hskip2em \textbf{send} $\mathsf{msg(}\textsc{view-change}, \bot, \mathit{prepareQC}, \mathit{view}\mathsf{)}$ to $\mathsf{leader}(\mathit{view})$

\smallskip
\State \hskip2em \textbf{as} $\mathsf{leader}(\mathit{view})$:
\State \hskip4em \textbf{wait for} $2f + 1$ \textsc{view-change} messages:
\State \hskip6em $M \gets \{m \,|\, \mathsf{matching\_msg(}m, \textsc{view-change}, \mathit{view}\mathsf{)}\}$
\State \hskip4em $\mathsf{Quorum\_Certificate} \text{ } \mathit{highQC} \gets \mathit{qc}$ with the highest $\mathit{qc}.\mathit{view}$ in $M$ \label{line:highest_qc}
\State \hskip4em $\mathsf{Value} \text{ } \mathit{proposal} \gets \mathit{highQC}.\mathit{value}$
\State \hskip4em \textbf{if} $\mathit{proposal} = \bot$:
\State \hskip6em $\mathit{proposal} \gets \mathit{proposal}_i$ \BlueComment{$\mathit{proposal}_i$ denotes the proposal of $P_i$}
\State \hskip4em \textbf{broadcast} $\mathsf{msg(}\textsc{prepare}, \mathit{proposal}, \mathit{highQC}, \mathit{view}\mathsf{)}$

\smallskip
\State \hskip2em \textbf{as} a process: \BlueComment{every process executes this part of the pseudocode}
\State \hskip4em \textbf{wait for} message $m$: $\mathsf{matching\_msg(}m, \textsc{prepare}, \mathit{view}\mathsf{)}$ from $\mathsf{leader}(\mathit{view})$
\State \hskip4em \textbf{if} $m.\mathit{qc}.\mathit{value} = m.\mathit{value}$ and ($\mathit{lockedQC}.\mathit{value} = m.\mathit{value}$ or $\mathit{qc}.\mathit{view} > \mathit{lockedQC}.\mathit{view}$): \label{line:view_core_check}
\State \hskip6em \textbf{send} $\mathsf{vote\_msg(}\textsc{prepare}, m.\mathit{value},\bot, \mathit{view}\mathsf{)}$ to $\mathsf{leader}(\mathit{view})$ \label{line:proposal_support}

\smallskip
\State \hskip2em \textcolor{blue}{\(\triangleright\) Precommit phase}
\State \hskip2em \textbf{as} $\mathsf{leader}(\mathit{view})$:
\State \hskip4em \textbf{wait for} $2f + 1$ votes: $V \gets \{\mathit{vote} \,|\, \mathsf{matching\_msg(}\mathit{vote}, \textsc{prepare}, \mathit{view}\mathsf{)}\}$
\State \hskip4em $\mathsf{Quorum\_Certificate} \text{ } \mathit{qc} \gets \mathsf{qc(}V\mathsf{)}$
\State \hskip4em \textbf{broadcast} $\mathsf{msg(}\textsc{precommit}, \bot, \mathit{qc}, \mathit{view}\mathsf{)}$

\smallskip
\State \hskip2em \textbf{as} a process: \BlueComment{every process executes this part of the pseudocode}
\State \hskip4em \textbf{wait for} message $m$: $\mathsf{matching\_qc(}m.\mathit{qc}, \textsc{prepare}, \mathit{view}\mathsf{)}$ from $\mathsf{leader}(\mathit{view})$
\State \hskip4em $\mathit{prepareQC} \gets m.\mathit{qc}$ \label{line:quad_update_prepare_qc}
\State \hskip4em \textbf{send} $\mathsf{vote\_msg(}\textsc{precommit}, m.\mathit{qc}.\mathit{value}, \bot, \mathit{view}\mathsf{)}$ to $\mathsf{leader}(\mathit{view})$

\smallskip
\State \hskip2em \textcolor{blue}{\(\triangleright\) Commit phase}
\State \hskip2em \textbf{as} $\mathsf{leader}(\mathit{view})$:
\State \hskip4em \textbf{wait for} $2f + 1$ votes: $V \gets \{\mathit{vote} \,|\, \mathsf{matching\_msg(}\mathit{vote}, \textsc{precommit}, \mathit{view}\mathsf{)}\}$
\State \hskip4em $\mathsf{Quorum\_Certificate}$ $\mathit{qc} \gets \mathsf{qc(}V\mathsf{)}$
\State \hskip4em \textbf{broadcast} $\mathsf{msg(}\textsc{commit}, \bot, \mathit{qc}, \mathit{view}\mathsf{)}$

\smallskip
\State \hskip2em \textbf{as} a process: \BlueComment{every process executes this part of the pseudocode}
\State \hskip4em \textbf{wait for} message $m$: $\mathsf{matching\_qc(}m.\mathit{qc}, \textsc{precommit}, \mathit{view}\mathsf{)}$ from $\mathsf{leader}(\mathit{view})$
\State \hskip4em $\mathit{lockedQC} \gets m.\mathit{qc}$ \label{line:update_locked_qc}
\State \hskip4em \textbf{send} $\mathsf{vote\_msg(}\textsc{commit}, m.\mathit{qc}.\mathit{value}, \bot, \mathit{view}\mathsf{)}$ to $\mathsf{leader}(\mathit{view})$

\smallskip
\State \hskip2em \textcolor{blue}{\(\triangleright\) Decide phase}
\State \hskip2em \textbf{as} $\mathsf{leader}(\mathit{view})$:
\State \hskip4em \textbf{wait for} $2f + 1$ votes: $V \gets \{\mathit{vote} \,|\, \mathsf{matching\_msg(}\mathit{vote}, \textsc{commit}, \mathit{view}\mathsf{)}\}$
\State \hskip4em $\mathsf{Quorum\_Certificate}$ $\mathit{qc} \gets \mathsf{qc(}V\mathsf{)}$
\State \hskip4em \textbf{broadcast} $\mathsf{msg(}\textsc{decide}, \bot, \mathit{qc}, \mathit{view}\mathsf{)}$ \label{line:broadcast_decide_message}

\smallskip
\State \hskip2em \textbf{as} a process: \BlueComment{every process executes this part of the pseudocode}
\State \hskip4em \textbf{wait for} message $m$: $\mathsf{matching\_qc(}m.\mathit{qc}, \textsc{commit}, \mathit{view}\mathsf{)}$ from $\mathsf{leader}(\mathit{view})$
\State \hskip4em \textbf{trigger} $\mathsf{decide}(m.\mathit{qc}.\mathit{value})$ \label{line:decide_view_core}
\end{algorithmic}
\end{algorithm}

\smallskip
\noindent \textbf{Proof of correctness.}
In this paragraph, we show that \name ensures weak validity, termination and agreement.
Recall that the main body of \name is given in \Cref{algorithm:quad}, whereas its view synchronizer \rare is presented in \Cref{algorithm:synchronizer} and its view core in \Cref{algorithm:view_1}.
We underline that the proofs concerned with the view core of \name can be found in~\cite{Yin2019}, as \name uses the same view core as HotStuff.

We start by proving that \name ensures weak validity.

\begin{theorem} [Weak validity]
\name ensures weak validity.
\end{theorem}
\begin{proof}
Suppose that all processes are correct.
Whenever a correct process updates its $\mathit{prepareQC}$ variable (line~\ref{line:quad_update_prepare_qc} of \Cref{algorithm:view_1}), it updates it to a quorum certificate vouching for a proposed value.
Therefore, leaders always propose a proposed value since the proposed value is ``formed'' out of $\mathit{prepareQC}$s of processes (line~\ref{line:highest_qc} of \Cref{algorithm:view_1}).
Given that a correct process executes line~\ref{line:decide_view_core} of \Cref{algorithm:view_1} for a value proposed by the leader of the current view, which is proposed by a process (recall that all processes are correct), the weak validity property is ensured.
\end{proof}

Next, we prove agreement.

\begin{theorem} [Agreement] \label{theorem:quad_agreement}
\name ensures agreement.
\end{theorem}
\begin{proof}
Two conflicting quorum certificates associated with the same view cannot be obtained in the view core of \name (\Cref{algorithm:view_1}); otherwise, a correct process would vote for both certificates, which is not possible according to \Cref{algorithm:view_1}.
Therefore, two correct processes cannot decide different values from the view core of \name in the same view.
Hence, we need to show that, if a correct process decides $v$ in some view $\mathit{view}$ in the view core (line~\ref{line:decide_view_core} of \Cref{algorithm:view_1}), then no conflicting quorum certificate can be obtained in the future views.

Since a correct process decides $v$ in view $\mathit{view}$ in the view core, the following holds at $f + 1$ correct processes: $\mathit{lockedQC}.\mathit{value} = v$ and $\mathit{lockedQC}.\mathit{view} = \mathit{view}$ (line~\ref{line:update_locked_qc} of \Cref{algorithm:view_1}).
In order for another correct process to decide a different value in some future view, a prepare quorum certificate for a value different than $v$ must be obtained in a view greater than $\mathit{view}$.
However, this is impossible as $f + 1$ correct processes whose $\mathit{lockedQC}.\mathit{value} = v$ and $\mathit{lockedQC}.\mathit{view} = \mathit{view}$ will not support such a prepare quorum certificate (i.e., the check at line~\ref{line:view_core_check} of \Cref{algorithm:view_1} will return false).
Thus, it is impossible for correct processes to disagree in the view core even across multiple views.
The agreement property is ensured by \name.
\end{proof}

Finally, we prove termination.

\begin{theorem} [Termination] \label{theorem:quad_termination}
\name ensures termination.
\end{theorem}
\begin{proof}
\rare ensures that, eventually, all correct processes remain in the same view $\mathit{view}$ with a correct leader for (at least) $\Delta = 8\delta$ time after $\mathit{GST}$.
When this happens, all correct processes decide in the view core.

Indeed, the leader of $\mathit{view}$ learns the highest obtained locked quorum certificate through the \textsc{view-change} messages (line~\ref{line:highest_qc} of \Cref{algorithm:view_1}).
Therefore, every correct process supports the proposal of the leader (line~\ref{line:proposal_support} of \Cref{algorithm:view_1}) as the check at line~\ref{line:view_core_check} of \Cref{algorithm:view_1} returns true.
After the leader obtains a prepare quorum certificate in $\mathit{view}$, all correct processes vote in the following phases of the same view.
Thus, all correct processes decide from the view core (line~\ref{line:decide_view_core} of \Cref{algorithm:view_1}), which concludes the proof.
\end{proof}

Thus, \name indeed solves the Byzantine consensus problem with weak validity.

\begin{corollary}
\name is a partially synchronous Byzantine consensus protocol ensuring weak validity.
\end{corollary}

\smallskip
\noindent \textbf{Proof of complexity.}
Next, we show that \name achieves $O(n^2)$ communication complexity and $O(f)$ latency complexity.
Before we start the proof, we clarify one point about \Cref{algorithm:quad}: as soon as $\mathsf{advance}(v)$ is triggered (line~\ref{line:synchronizer_advance}), for some view $v$, the process immediately stops accepting and sending messages for the previous view.
In other words, it is as if the ``stop accepting and sending messages for the previous view'' action immediately follows the $\mathsf{advance}(\cdot)$ upcall in \Cref{algorithm:synchronizer}.\footnote{Note that this additional action does not disrupt \rare (nor its proof of correctness and complexity).}

We begin by proving that, if a correct process sends a message of the view core associated with a view $v$ which belongs to an epoch $e$, then the last entered epoch prior to sending the message (in the behavior of the process) is $e$ (this result is similar to the one of \Cref{lemma:epoch_completed_last_enter}).
A message is a \emph{view-core} message if it is of \textsc{view-change}, \textsc{prepare}, \textsc{precommit}, \textsc{commit} or \textsc{decide} type.

\begin{lemma} \label{lemma:quad_sent_previously_entered}
Let $P_i$ be a correct process and let $P_i$ send a view-core message associated with a view $v$, where $v$ belongs to an epoch $e$.
Then, $e$ is the last epoch entered by $P_i$ in $\beta_i$ before sending the message.
\end{lemma}
\begin{proof}
Process $P_i$ enters the view $v$ before sending the view-core message (since $\mathsf{start\_executing}(v)$ is invoked upon $P_i$ entering $v$; line~\ref{line:start_executing_logic} of \Cref{algorithm:quad}).
By \Cref{lemma:no_jump}, $P_i$ enters the first view of the epoch $e$ (and, hence, $e$) before sending the message.
By contradiction, suppose that $P_i$ enters another epoch $e'$ after entering $e$ and before sending the view-core message.

By \Cref{lemma:increasing_views}, we have that $e' > e$.
However, this means that $P_i$ does not send any view-core messages associated with $v$ after entering $e'$ (since $(e' - 1) \cdot (f + 1) + 1 > v$ and $P_i$ enters monotonically increasing views by \Cref{lemma:increasing_views}).
Thus, a contradiction, which concludes the proof.
\end{proof}

Next, we show that a correct process sends (at most) $O(n)$ view-core messages associated with a single epoch.

\begin{lemma} \label{lemma:quad_constant_in_epoch}
Let $P_i$ be a correct process.
For any epoch $e$, $P_i$ sends (at most) $O(n)$ view-core messages associated with views that belong to $e$.
\end{lemma}
\begin{proof}
Recall that $P_i$ enters monotonically increasing views (by \Cref{lemma:increasing_views}), which means that $P_i$ never invokes $\mathsf{start\_executing}(v)$ (line~\ref{line:start_executing_logic} of \Cref{algorithm:quad}) multiple times for any view $v$.

Consider a view $v$ that belongs to $e$.
We consider two cases:
\begin{compactitem}
    \item Let $P_i$ be the leader of $v$.
    In this case, $P_i$ sends (at most) $O(n)$ view-core messages associated with $v$.
    
    \item Let $P_i$ not be the leader of $v$.
    In this case, $P_i$ sends (at most) $O(1)$ view-core messages associated with $v$.
\end{compactitem}
Given that $P_i$ is the leader of at most one view in every epoch $e$ (since $\mathsf{leader(\cdot)}$ is a round-robin function), $P_i$ sends (at most) $1 \cdot O(n) + f \cdot O(1) = O(n)$ view-core messages associated with views that belong to $e$.
\end{proof}



Finally, we prove the complexity of \name.

\begin{theorem} [Complexity] \label{theorem:quad_complexity_appendix}
\name achieves $O(n^2)$ communication complexity and $O(f)$ latency complexity.
\end{theorem}
\begin{proof}
As soon as all correct processes remain in the same view for $8\delta$ time, all correct processes decide from the view core.
As \rare uses $\Delta = 8\delta$ in the implementation of \name (line~\ref{line:init_rare_sync} of \Cref{algorithm:quad}), all processes decide by time $t_s + 8\delta$, where $t_s$ is the first synchronization time after $\mathit{GST}$ (\Cref{definition:t_s}). 
Given that $t_s + 8\delta - \mathit{GST}$ is the latency of \rare (see \Cref{subsection:view_synchronization_problem_definition}) and the latency complexity of \rare is $O(f)$ (by \Cref{theorem:latency}), the latency complexity of \name is indeed $O(f)$.

Fix a correct process $P_i$.
For every epoch $e$, $P_i$ sends (at most) $O(n)$ view-core messages associated with views that belong to $e$ (by \Cref{lemma:quad_constant_in_epoch}).
Moreover, if $P_i$ sends a view-core message associated with a view that belongs to an epoch $e$, then $e$ is the last epoch entered by $P_i$ prior to sending the message (by \Cref{lemma:quad_sent_previously_entered}).
Hence, in the time period $[\mathit{GST}, t_s + 8\delta]$, $P_i$ sends view-core messages associated with views that belong to (at most) $O(1)$ epochs (by \Cref{lemma:constant-number-of-epochs-after-GST}).
Thus, $P_i$ sends (at most) $O(1) \cdot O(n) = O(n)$ view-core messages in the time period $[\mathit{GST}, t_s + 8\delta]$, each containing a single word.
Moreover, during this time period, the communication complexity of \rare is $O(n^2)$ (by \Cref{theorem:latency}).
Therefore, the communication complexity of \name is $n \cdot O(n) + O(n^2) = O(n^2)$.
\end{proof}
\section{\sname: Proof of Correctness and Complexity} \label{section:squad_appendix}






First, we show that the certification phase of \sname ensures computability and liveness.

\begin{lemma} [Computability \& liveness]  \label{lemma:groundwork}
Certification phase (\Cref{algorithm:groundwork}) ensures computability and liveness.
Moreover, every correct process sends (at most) $O(n)$ words and obtains a certificate by time $\mathit{GST} + 2\delta$.
\end{lemma}
\begin{proof}
As every correct process broadcasts \textsc{disclose}, \textsc{certificate} or \textsc{allow-any} messages at most once and each message contains a single word, every correct process sends (at most) $3 \cdot n \cdot 1 = O(n)$ words. 
Next, we prove computability and liveness.

\smallskip
\noindent \textbf{Computability.} 
Let all correct processes propose the same value $v$ to \sname.
Since no correct process broadcasts a \textsc{disclose} message for a value $v' \neq v$, no process ever obtains a certificate $\sigma_{v'}$ for $v'$ such that $\mathit{CombinedVerify}(v', \sigma_{v'}) = \mathit{true}$ (line~\ref{line:verify_v}).

Since all correct processes broadcast a \textsc{disclose} message for $v$ (line~\ref{line:send_disclose}), the rule at line~\ref{line:rule_allow_any} never activates at a correct process.
Thus, no correct process ever broadcasts an \textsc{allow-any} message (line~\ref{line:broadcast_allow_any}), which implies that no process obtains a certificate $\sigma_{\bot}$ such that $\mathit{CombinedVerify}(\text{``allow any''}, \sigma_{\bot}) = \mathit{true}$ (line~\ref{line:verify_any_value}).
The computability property is ensured.

\smallskip
\noindent \textbf{Liveness.} 
Every correct process receives all \textsc{disclose} messages sent by correct processes by time $\mathit{GST} + \delta$ (since message delays are $\delta$ after $\mathit{GST}$; see \Cref{section:model}).
Hence, all correct processes receive (at least) $2f + 1$ \textsc{disclose} messages by time $\mathit{GST} + \delta$.
Therefore, by time $\mathit{GST} + \delta$, all correct processes send either (1) a \textsc{certificate} message upon receiving $f + 1$ \textsc{disclose} messages for the same value (line~\ref{line:broadcast_certificate_1}), or (2) an \textsc{allow-any} message upon receiving $2f + 1$ \textsc{disclose} messages without a ``common value'' (line~\ref{line:broadcast_allow_any}).
Let us consider two possible scenarios:
\begin{compactitem}
    \item There exists a correct process that has broadcast a \textsc{certificate} message upon receiving $f + 1$ \textsc{disclose} messages for the same value (line~\ref{line:broadcast_certificate_1}) by time $\mathit{GST} + \delta$.
    Every correct process receives this message by time $\mathit{GST} + 2\delta$ (line~\ref{line:receive_certificate}) and obtains a certificate.
    Liveness is satisfied by time $\mathit{GST} + 2\delta$ in this case.
    
    \item Every correct process broadcasts an \textsc{allow-any} message (line~\ref{line:broadcast_allow_any}) by time $\mathit{GST} + \delta$.
    Hence, every correct process receives $f + 1$ \textsc{allow-any} messages by time $\mathit{GST} + 2\delta$ (line~\ref{line:f+1_allow_any}) and obtains a certificate (line~\ref{line:combine_allow_any}).
    The liveness property is guaranteed by time $\mathit{GST} + 2\delta$ in this case as well.
\end{compactitem}
The liveness property is ensured by time $\mathit{GST} + 2\delta$.
\end{proof}

Finally, we show that \sname is a Byzantine consensus protocol with $O(n^2)$ communication complexity and $O(f)$ latency complexity.

\begin{theorem} \label{theorem:squad}
\sname is a Byzantine consensus protocol with (1) $O(n^2)$ communication complexity, and (2) $O(f)$ latency complexity.
\end{theorem}
\begin{proof}
If a correct process decides a value $v'$ and all correct processes have proposed the same value $v$, then $v' = v$ since (1) correct processes ignore values not accompanied by their certificates (line~\ref{line:start_quad_squad}), and (2) the certification phase of \sname ensures computability (by \Cref{lemma:groundwork}).
Therefore, \sname ensures validity.

Fix an execution $E_{\sname}$ of \sname.
We denote by $t_{\mathit{last}}$ the time the last correct process starts executing $\name_{\mathit{cer}}$ (line~\ref{line:start_quad_squad}) in $E_{\sname}$; i.e., by $t_{\mathit{last}}$ every correct process has exited the certification phase.
Moreover, we denote the global stabilization time of $E_{\sname}$ by $\mathit{GST}_1$.
Now, we consider two possible scenarios:
\begin{compactitem}
    \item Let $\mathit{GST}_1 \geq t_{\mathit{last}}$.
    $\name_{\mathit{cer}}$ solves the Byzantine consensus problem with $O(n^2)$ communication and $O(f)$ latency (by \Cref{theorem:quad_complexity_appendix}).
    As processes send (at most) $O(n)$ words associated with the certification phase (by \Cref{lemma:groundwork}), consensus is solved in $E_{\sname}$ with $n \cdot O(n) + O(n^2) = O(n^2)$ communication complexity and $O(f)$ latency complexity.
    
    \item Let $\mathit{GST}_1 < t_{\mathit{last}}$.
    Importantly, $t_{\mathit{last}} - \mathit{GST}_ 1 \leq 2\delta$ (by \Cref{lemma:groundwork}).
    Now, we create an execution $E_{\name}$ of the original \name protocol in the following manner:
    \begin{compactenum}
        \item $E_{\name} \gets E_{\sname}$. 
        If a process sends a value with a valid accompanying certificate, then \emph{just} the certificate is removed in $E_{\name}$ (i.e., the corresponding message stays in $E_{\name}$).
        Otherwise, the entire message is removed.
        Note that no message sent by a correct process in $E_{\sname}$ is removed from $E_{\name}$ as correct processes only send values accompanied by their valid certificates.
        
        \item We remove from $E_{\name}$ all events associated with the certification phase of \sname.
        
        \item The global stabilization time of $E_{\name}$ is set to $t_{\mathit{last}}$. We denote this time by $\mathit{GST}_2 = t_{\mathit{last}}$.
        Note that we can set $\mathit{GST}_2$ to $t_{\mathit{last}}$ as $t_{\mathit{last}} > \mathit{GST}_1$.
    \end{compactenum}
    In $E_{\name}$, consensus is solved with $O(n^2)$ communication and $O(f)$ latency.
    Therefore, the consensus problem is solved in $E_{\sname}$.
    
    Let us now analyze the complexity of $E_{\sname}$:
    \begin{compactitem}
        \item The latency complexity of $E_{\sname}$ is $t_{\mathit{last}} - \mathit{GST}_1 + O(f) = O(f)$ (as $t_{\mathit{last}} - \mathit{GST}_1 \leq 2\delta$).
        
        \item The communication complexity of $E_{\sname}$ is the sum of (1) the number of words sent in the time period $[\mathit{GST}_1, t_{\mathit{last}})$, and (2) the number of words sent at and after $t_{\mathit{last}}$ and before the decision, which is $O(n^2)$ since that is the communication complexity of $E_{\name}$ and each correct process sends (at most) $O(n)$ words associated with the certification phase (by \Cref{lemma:groundwork}).
        
        Fix a correct process $P_i$.
        Let us take a closer look at the time period $[\mathit{GST}_1, t_{\mathit{last}})$:
        \begin{compactitem}
            \item Let $\mathit{epochs}_{\rare}$ denote the number of epochs for which $P_i$ sends \textsc{epoch-completed} or \textsc{enter-epoch} messages in this time period.
            By \Cref{lemma:within_delta}, $P_i$ enters (at most) $2 = O(1)$ epochs in this time period.
            Hence, $\mathit{epochs}_{\rare} = O(1)$ (by lemmas~\ref{lemma:epoch_completed_last_enter} and~\ref{lemma:enter_epoch_in_e}).
            
            \item Let $\mathit{epochs}_{\name_{\mathit{cer}}}$ denote the number of epochs for which $P_i$ sends view-core messages in this time period.
            By \Cref{lemma:within_delta}, $P_i$ enters (at most) $2 = O(1)$ epochs in this time period.
            Hence, $\mathit{epochs}_{\name_{\mathit{cer}}} = O(1)$ (by \Cref{lemma:quad_sent_previously_entered}).
        \end{compactitem}
        For every epoch $e$, $P_i$ sends (at most) $O(n)$ \textsc{epoch-completed} and \textsc{enter-epoch} messages (by lemmas~\ref{lemma:epoch_completed_n} and~\ref{lemma:enter_epoch_n}).
        Moreover, for every epoch $e$, $P_i$ sends (at most) $O(n)$ view-core messages associated with views that belong to $e$ (by \Cref{lemma:quad_constant_in_epoch}).\footnote{Note that lemmas~\ref{lemma:epoch_completed_last_enter},~\ref{lemma:enter_epoch_in_e},~\ref{lemma:epoch_completed_n},~\ref{lemma:enter_epoch_n},~\ref{lemma:within_delta},~\ref{lemma:quad_sent_previously_entered} and~\ref{lemma:quad_constant_in_epoch}, which we use to prove the theorem, assume that all correct processes have started executing \rare and \name by $\mathit{GST}$. In \Cref{theorem:squad}, this might not be true as some processes might start executing \rare after $\mathit{GST}$ (since $t_{\mathit{last}} > \mathit{GST}$). However, it is not hard to verify that the claims of these lemmas hold even in this case.}
        As each \textsc{epoch-completed}, \textsc{enter-epoch} and view-core message contains a single word and $P_i$ sends at most $O(n)$ words during the certification phase (by \Cref{lemma:groundwork}), we have that $P_i$ sends (at most) $\mathit{epochs}_{\rare} \cdot O(n) + \mathit{epochs}_{\name_{\mathit{cer}}} \cdot O(n) + O(n) = O(n)$ words during the time period $[\mathit{GST}_1, t_{\mathit{last}})$.
        Therefore, the communication complexity of $E_{\sname}$ is $n \cdot O(n) + O(n^2) + O(n^2) = O(n^2)$.\footnote{The first ``$n \cdot O(n)$'' term corresponds to the messages sent during the time period $[\mathit{GST}_1, t_{\mathit{last}})$, the second ``$O(n^2)$'' term corresponds to the messages sent during the certification phase, and the third ``$O(n^2)$'' term corresponds to the messages sent at and after $t_{\mathit{last}}$ and before the decision has been made.}
    \end{compactitem}
    Hence, consensus is indeed solved in $E_{\sname}$ with $O(n^2)$ communication complexity and $O(f)$ latency complexity.
\end{compactitem}
The theorem holds.
\end{proof}

\end{document}